\documentclass{amsart}

\usepackage[reset,a4paper,hmargin=3truecm]{geometry}
\usepackage{amsmath,amsthm,amssymb,amscd,epic}
\usepackage{graphicx,bbm,amsthm,graphicx}
\usepackage{mathtools}
\usepackage{leftidx}
\usepackage{hyperref}
\usepackage{color}
\usepackage{cool}
\usepackage{bm}
\usepackage{bbold}
\usepackage{subfig}
\usepackage{tikz}

\theoremstyle{plain}
\newtheorem{thm}{Theorem}[section]
\newtheorem{lem}[thm]{Lemma}
\newtheorem{prop}[thm]{Proposition}

\newtheorem{cor}[thm]{Corollary}

\newtheorem{conject}[thm]{Conjecture}

\theoremstyle{definition}
\newtheorem{dfn}{Definition}[section]
\newtheorem{ex}{Example}[section]
\newtheorem{prob}{Problem}[section]

\theoremstyle{remark}
\newtheorem{rem}{Remark}[section]

\DeclareMathOperator{\Diag}{Diag}

\newcommand{\Z}{\mathbb{Z}} 
\newcommand{\Q}{\mathbb{Q}} 
\newcommand{\R}{\mathbb{R}} 
\newcommand{\C}{\mathbb{C}} 

\newcommand{\bA}{\mathbf{A}}
\newcommand{\bJ}{\mathbf{J}}

\newcommand{\sh}{{\rm sh}}
\newcommand{\ch}{{\rm ch}}

\usepackage{ctable,tabularx}

\newcommand{\mat}[1]{\begin{bmatrix}#1\end{bmatrix}}

\newcommand{\mycomment}[1]{}

\DeclareMathOperator{\tridiag}{tridiag}
\newcommand{\Tridiag}[4]{\tridiag\mat{#1 & #2 \\ #3}_{#4}}

\newcommand{\e}{\varepsilon}

\title[Two photon AQRM - degeneracy and symmetry]{Two-photon quantum Rabi models -- \\ Spectral degeneracy  and symmetries}
\author{Cid Reyes Bustos and Masato Wakayama}

\date{\today}

\begin{document}

\begin{abstract}
    The quantum Rabi model occupies a distinguished role in the study of quantum light and matter, describing
    the most fundamental interactions. Elucidation of its spectral properties, along with those of its generalizations, is necessary for quantum optics and applications in areas such as quantum information technologies.

    In this paper, we explore a precise description of energy degeneracy of the two-photon asymmetric quantum Rabi model and the nature of the system's symmetries; as a result, we demonstrate that these have significant
    relationships with the structure of certain algebraic curves (hyperelliptic curves) and number theory.
    Concretely, we first prove the existence of degenerate eigenvalues using the monodromy data of the
    Fuchsian ODE, which draws the eigenvalue problem of the system, and the existence of associated symmetry
    operators, previously discovered and studied heuristically in physics.

    Subsequently, we give a complete characterization of spectral degeneracies and present a series of conjectures showing that the degeneracy may determine the entire energy spectral structure. 
    More precisely, for sufficiently strong interaction, we obtain an excellent approximation of the energy curves  by the zero locus of a two-variable polynomial associated to degenerate eigenstates. In addition, we postulate an explicit description of the relation between the missing symmetry at the Hamiltonian level (hidden symmetry) and the degeneracy of the two-photon quantum Rabi model.
\end{abstract}

\maketitle

\tableofcontents

\section{Introduction}

The quantum Rabi model (QRM), widely regarded as one of the fundamental models of quantum light and matter
interaction, describes the interaction between a two-level system and a bosonic mode \cite{B2011}.
Motivated by potential applications for quantum information technologies \cite{HR2008}, along with advances in experimental technology, the QRM and its generalizations are the subject of intense research in both experimental and theoretical physics (e.g., \cite{Yea2017, Yea2018}). The mathematical structures underlying QRM and related models have also been explored extensively in different areas: asymptotic theory \cite{BMZ2021,BMZ2021B,SZ2026}, analytical formula of heat kernel  and spectral zeta function theory \cite{HS2025,M2023,N2026,RW2020hk,RBW2024}, among others \cite{KR2024,MP2026,R2024}.

The Hamiltonian of the QRM is given by
\[
  H_{\rm{QRM}}= a^\dag a +\Delta \sigma_z +g(a+a^\dag)\sigma_x,
\]
where
\[
  \sigma_x:=\begin{bmatrix}
    0 & 1 \\
    1 & 0
  \end{bmatrix}, \quad \sigma_y:= \begin{bmatrix}
    0 & -i \\
    i & 0
  \end{bmatrix}, \quad \sigma_z:=\begin{bmatrix}
    1 & 0 \\
    0 & -1
  \end{bmatrix}, 
\]
are the usual Pauli matrices, $2\Delta>0$ is the difference between the two levels, $g>0$ is the coupling strength and  $a^\dag$ and $a$ the creation and annihilation operators for harmonic oscillators of frequency $\omega$ (we may set $\omega=1$), that is,  $[a,a^\dag]=1$. The Hamiltonian $H_{\rm{QRM}}$ is a self-adjoint operator acting on the Hilbert space $\mathcal{H} = \mathcal{H}_0 \otimes \C^2$, where $\mathcal{H}_0$ is either $L^2(\R)$ or the Fock space. 

An essential feature of the QRM is the existence of a parity (or $\Z_2$) symmetry, realized by the operator
\[
  \hat{\mathcal{P}}_{\rm{QRM}}:= e^{i\pi a^\dag a}\sigma_z.
\]
The existence of parity symmetry was crucial for the proof of the exact solvability of the QRM by Daniel Braak \cite{B2013MfI}, which revitalized research on the QRM and similar quantum interaction models \cite{bcbs2016}. The fact that the Hamiltonian $H_{\rm{QRM}}$ commutes with the involutive operator $\hat{\mathcal{P}}_{\rm{QRM}}$ gives a decomposition of the Hilbert space $\mathcal{H}$ into two invariant subspaces, labeling the spectrum into two parities. Spectral degeneration of the QRM occurs only for different parities \cite{K1985JMP} and the eigenfunctions are of a special type called Juddian, or quasi-exact \cite{KRW2017}.

The intimate relation between the symmetry and the spectral degeneracy is better observed in the {\em asymmetric quantum Rabi model} (AQRM), a model obtained by the introduction of a bias term $\e \sigma_x$ to the QRM Hamiltonian. In general, the effect of introducing the bias term is the elimination of both the parity symmetry  and the spectral degeneracies. However, an interesting phenomenon appears for half-integer values of the bias parameter. In this case, spectral degeneracies appear again \cite{LB2015JPA}, a along with a special type of commuting operator, known as a hidden symmetry operator since it does not immediately arise from the Hamiltonian (see, e.g., \cite{KRW2017,MBB2021, RBBW2021}). We note that hidden symmetry operators cannot be independent of the system parameters, as first demonstrated numerically in \cite{A2020}.
In addition, a geometric picture for the spectrum based on the hyperelliptic curve and surfaces, coming from polynomials related to the degeneracy, allows for a new perspective of the spectrum (see \cite{RBW2022} for a detailed discussion).

A motivation for the present study is to investigate the relation between symmetry and spectral degeneracies, including the possible existence algebro-geometric picture for spectral curves, in other quantum interaction models. 
The {\em two-photon quantum Rabi model} (2pQRM) is the quantum model describing the interaction of a two-level system and a bosonic mode where two photons are absorbed or emitted. 
In physics, research on the 2pQRM has been motivated due to its potential applications to quantum information technologies and superconducting circuits \cite{FRRSFD2018,TZY2025}.
The Hamiltonian  $H_{2p}$ of the 2pQRM is given by
\begin{align}
  \label{eq:tpHamilt0}
  H_{2p} = \omega (a^t a + \frac12) + g ( a^2 + (a^\dagger)^2 ) \sigma_z + \Delta \sigma_x, 
\end{align}
with parameters sharing the physical interpretation with those of the QRM. For an overview of the general properties of this model, including analogs for higher number of photons, we refer the reader to \cite{B2025}.
In this paper, we assume $g < 1/2$, a necessary and sufficient condition for $H_{2p}$ to have a pure point spectrum (see, e.g., \cite{HS2024,DXBC2016}). In this case, the asymptotic properties of the high energy spectrum are described in detail in \cite{BMZ2021B}. For other values of $g$, there is a more complicated spectral structure, including continuous spectrum (see, e.g., \cite{SZ2026}). Another feature of the 2pQRM Hamiltonian is that its Hamiltonian is compatible with the action of the generators of the oscillator representation of the Lie algebra $\mathfrak{sl}_2$, a fact that has been exploited in physics (see, e.g., \cite{D2009,X2025,Z2016}). 

The 2pQRM Hamiltonian commutes with a parity operator $\hat{\mathcal{P}}^2 = \exp(i \pi a^\dag a)$, inducing a parity decomposition
\begin{align}
  \label{eq:2pQRMpar}
    \mathcal{H}= \mathcal{H}_{\rm{even}} \oplus \mathcal{H}_{\rm{odd}}
\end{align}
of the  Hilbert space $\mathcal{H}$ into invariant subspaces. Since $\hat{\mathcal{P}}^2$ acts diagonally on $\mathcal{H}$, the parity decomposition corresponds exactly to even and odd functions in both components in the Hilbert space $\mathcal{H}$. Consequently, the eigenvalues can be labeled according to their parity. We
show in Figure \ref{fig:tpQRMeigencurves}(a) the spectral curves for the 2pQRM, the plot of the eigenvalues for different values of the coupling parameter $g$ for fixed $\Delta$.

\begin{figure}[ht]
  \centering
  \includegraphics[height=4.5cm]{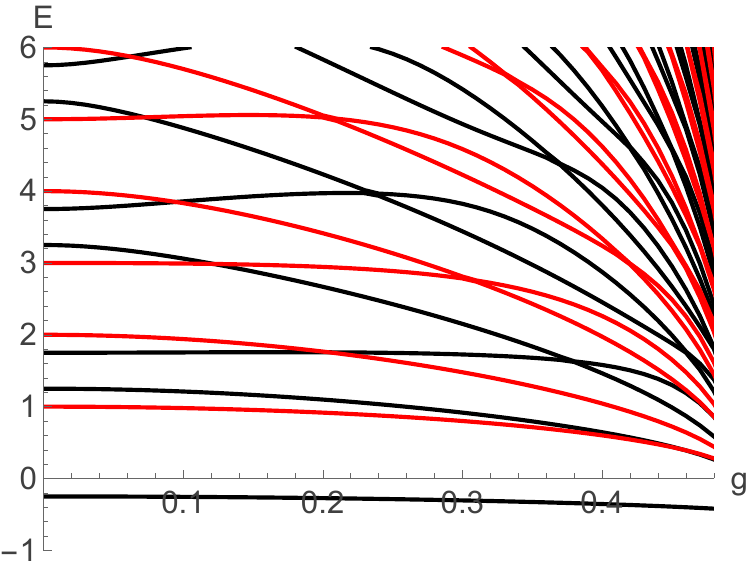}
  \caption{Spectral curves for the 2pQRM for $\Delta=0.75$. Black (resp. red) curves correspond to even (resp. odd) eigenvalue curves.}
  \label{fig:tpQRMeigencurves}
\end{figure}

Note that as the value of the coupling parameter reaches the critical value $g_c=\frac12$, the eigenvalue curves appear to converge, in a process sometimes called spectral collapse. Actually, the critical value $g_c$ is the point where the spectrum transitions from pure point spectrum to a continuous spectrum (see, e.g.,  \cite{DXBC2016,RMVRL2020} for a detailed discussion).

As shown in Figure  \ref{fig:tpQRMeigencurves}, there are crossings of the spectral curves both between spectral curves of \emph{shared-parity} (that is, even-even or odd-odd) and spectral curves of \emph{opposite-parity}. The opposite-parity degeneracies correspond to the $\Z_2$-symmetry, while the shared-parity crossings appear due to the presence of a $\Z_4$-symmetry. The Hamiltonian  $H_{2p}$ commutes with the operator
\[
  \hat{\mathcal{P}}_x = \exp(\tfrac{i \pi a^\dag a}2)\sigma_x,
\]
and each of the subspaces $\mathcal{H}_{\rm{even}}$ and $\mathcal{H}_{\rm{odd}}$  decomposes further into two invariant subspaces associated with the shared-parity spectral crossings. It turns out that the shared-parity degenerate eigenvalues behave similarly to the spectral degeneracies of the QRM. For instance, the eigenfunctions are of Juddian type and their existence is determined by a polynomial constraint condition. On the other hand, there is no analog in the QRM for the opposite-parity degenerate, for instance, it is known that their eigenvalues can be expressed as finite sums of special functions (see e.g. \cite{X2020} for a detailed description in the asymmetric case).

In this paper, we focus on the symmetry breaking generalization of the 2pQRM, the {\em asymmetric two-photon quantum Rabi model} (2pAQRM), with Hamiltonian $H_{\e}$ given by
\begin{align*}
  H_{\e}= (a^\dag a + \frac12) + g ( a^2 + (a^\dagger)^2 ) \sigma_z + \Delta \sigma_x + \tau \e \sigma_z,
\end{align*}
obtained by adding a bias term to the 2pQRM Hamiltonian and where $\tau = \sqrt{1-4g^2}$ is an auxiliary parameter. Here, we also assume $g < 1/2$ so that the spectrum to consist only of eigenvalues. It is immediate to verify that $H_{\e}$ commutes with the parity operator $\hat{\mathcal{P}}^2$ for any real value of $\e$, and that the opposite-parity degeneracies may be present for any bias, as shown in Figure \ref{fig:tpQRMeigencurves2}(a). 

\begin{figure}[ht]
  \centering
  \subfloat[$\e=0.2$]{
    \includegraphics[height=4.5cm]{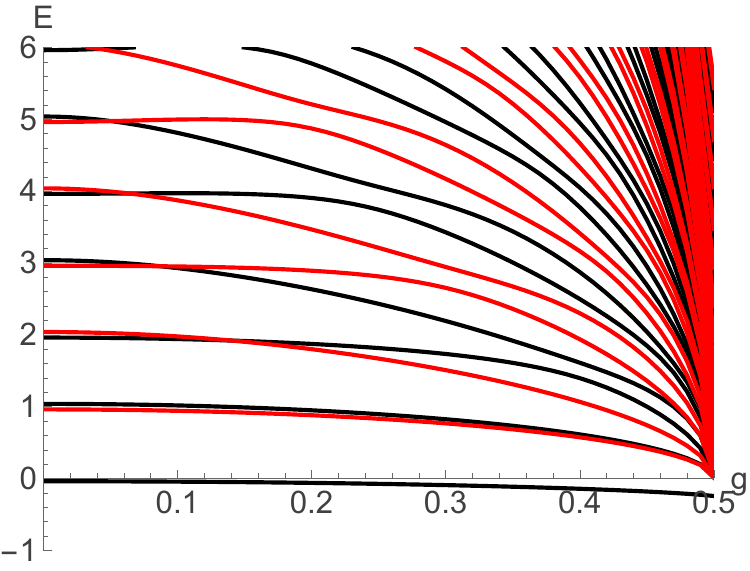}}
  ~ \qquad
  \subfloat[$\e=1$]{
    \includegraphics[height=4.5cm]{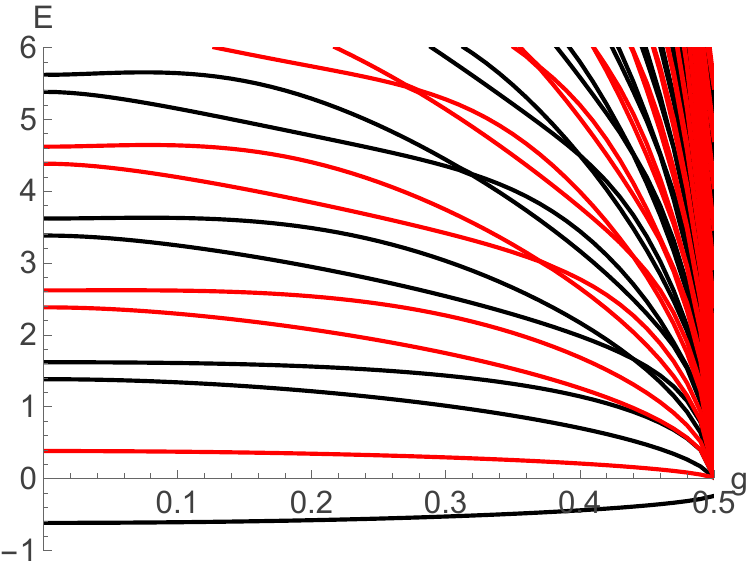}}
  \caption{Spectral curves of the 2pAQRM for $\Delta=0.5$ and $\e=0.2$ (left) and $\e=1$ (right). Black (resp. red) curves correspond to even (resp. odd) eigenvalue curves. The opposite-parity crossings remain for any value of $\e$ (left) while the shared-parity crossings appear only for integer values of $\e$ (right).}
  \label{fig:tpQRMeigencurves2}
\end{figure}

While the $\Z_2$-symmetry remains, the introduction of the bias parameter breaks the $\Z_4$-symmetry and eliminates the shared-parity crossings of the 2pQRM in general (see Figure \ref{fig:tpQRMeigencurves2}(a)). However, it was observed \cite{XC2021} that for integral values of the bias parameter $\e$, the shared-parity crossings reappear (see Figure \ref{fig:tpQRMeigencurves2}(b)) associated with the existence of a hidden symmetry operator.

Although there has been considerable research on Juddian solutions for the 2pAQRM using different approaches (see, e.g.,  \cite{XC2021} or \cite{BNZ2026} for a recent approach using Bethe ansatz), a complete analysis of the existence of the shared-parity degenerate eigenvalues was still incomplete. For example, there was no general proof of the existence of shared-parity degeneracies, that is, that the numerically observed crossings are actual degeneracies and not avoided crossings or numerical errors. Similarly, for the integer $\e$, a hidden symmetry operator was discovered in \cite{XC2022} where a heuristic method for its computation was also provided, but proofs of existence and other important properties were not provided.

In this paper, we provide the proofs for both the degeneracy and the hidden symmetry for the integer bias parameter $\e$. In particular, we prove the existence of the shared-parity degeneracies by studying the constraint polynomials associated with the Juddian solutions, and the properties of the hidden symmetry operator by a careful analysis of the algebraic properties of the operators in the algebra ${\rm Mat}_2[\C[a,a^\dag]]$ where $\C[a,a^\dag]$ is the usual Weyl algebra. An important tool used to describe the multiplicity of the eigenvalues is a Heun picture of the eigenvalue problem of the 2pAQRM. This approach allows us to use the general theory of Heun ODEs, including the monodromy theory of Fuchsian equations, and obtain the full characterization the degenerate eigenvalues of the 2pAQRM for all bias parameters.

The proofs for both symmetry and spectral degeneracy rely heavily on algebraic properties of the 2pAQRM and its eigenfunctions. In fact, the spectrum of the 2pAQRM itself seems to be determined in a large extent by algebro-geometric structures associated with polynomials appearing in the settings discussed in this paper. For instance, the \emph{inter-constraint polynomial}, the quotient of the two constraint polynomials associated with the eigenfunctions a shared-parity degeneracy, gives an approximation to the low energy eigenvalues of the 2pAQRM, as shown in Figure \ref{fig:EA0}. 
On the other hand, the joint eigenvalues of the 2pAQRM Hamiltonian and its hidden symmetry operator lie in certain hyperelliptic curves, which in turn are conjecturally related to the aforementioned inter-constraint polynomial. In this paper, we explore this geometric picture, which constitutes a promising line of research that may give further insights into the spectrum of quantum interaction models and provide new connections with geometry and arithmetics. 

\begin{figure}[ht]
  \centering
    \includegraphics[height=4.5cm]{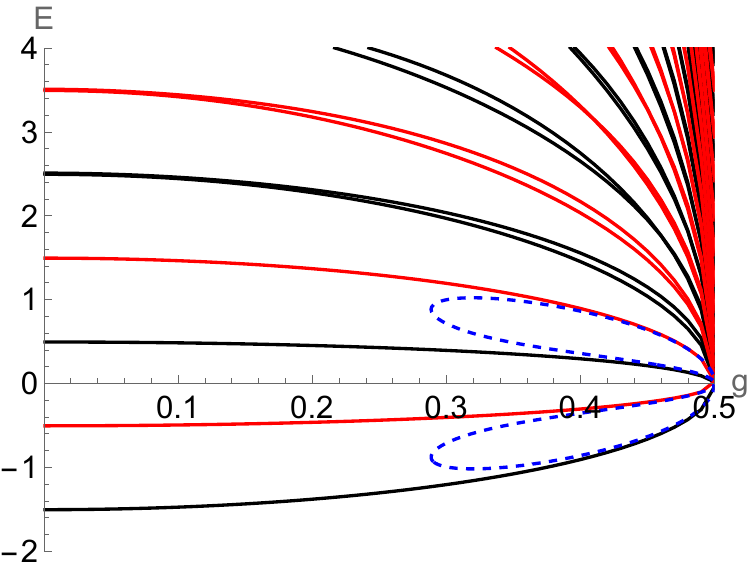}
    \caption{Excellent approximation of spectral curves by the inter-constraint polynomial for $\Delta=\frac16$ for $\e=2$.}
  \label{fig:EA0}
\end{figure}

In this paper, we make extensive use of the monodromy data of the Heun picture of the 2pAQRM to describe the
properties of the eigenvalues and eigenfunctions. The development of the Heun picture for the 2pAQRM is closely
related a similar for the non-commutative harmonic oscillator (NCHO). The NCHO is a mathematical model, introduced by Alberto Parmeggiani and the second author in \cite{PW2001, PW2002}, with  deep arithmetical properties appearing unexpectedly in its spectrum (see, e.g., \cite{KW2023}). Similarly to quantum interaction models, an asymmetric NCHO (aNCHO) version can be defined by adding an appropriate bias term (see \cite{RBW2023} and Definition \ref{dfn:aNCHO} below).

For a long time, no relation of the NCHO with any physical model was known; however, when describing the eigenvalue problem using the Heun ODE model, it was found that the QRM eigenvalue problem (given by a confluent Heun ODE) may be obtained from that of NCHO through the confluence of regular singularities, which means that NCHO can be considered a covering model of QRM \cite{W2015, RBW2023}. Subsequently, a more decisive fact emerged, namely,  in \cite{N2024} it was shown that the NCHO and the 2pQRM are related by a non-unitary equivalence (i.e. equivalence of eigenvalue problems). In particular, this discovery has made possible to draw parallels between the properties of the two models. We give a short overview of these relation between the models in Section \ref{sec:NCHOp} and remark that a detailed study of the relations between the two models is the subject of a forthcoming paper by the authors \cite{RBW2027}.

The paper is organized as follows. In Section \ref{sec:preliminaries} we introduce the two main models, the 2pAQRM and aNCHO, as well as the relationship between the two models. In Section \ref{sec:deg}, we set the stage for the analysis of the spectrum by describing the eigenvalue problem of the 2pAQRM in the Fock space and its equivalent Heun ODE picture. Next, in Section \ref{sec:juddian} we describe the shared-parity spectral degeneracies for the integral bias $\e$, showing that it is reduced to the proof of a divisibility and positivity property of certain \emph{constraint polynomials} associated with the solutions. 

Section \ref{sec:divpos} is concerned with the algebraic properties of constraint polynomials, a key factor in proving the existence of shared-parity degenerate eigenvalues. Specifically, it is shown that when a divisibility relationship exists between two constraint polynomials and that their quotient is also a polynomial. In Section \ref{sec:qtpol}, numerical analysis reveals that the curves on the energy $E$ and $g$ plane (for fixed $\Delta$) defined by the zero of these inter-constraint polynomials provide a remarkably accurate approximation of the energy curves for low-energy states—starting from the ground state—where there is no degeneracy; consequently, the section presents conjectures and open problems regarding this phenomenon from a mathematical perspective.
In Section \ref{sec:degstruct}, using these results and the monodromy data of the Heun ODE picture of the eigenvalue problem, we fully describe the structure of the solutions for the degenerate spectrum of the 2pAQRM. The existence of the hidden symmetry operators and their properties are discussed in Section \ref{sec:hiddenS} and as a conclusion, in Section \ref{sec:symdeg} we propose a conjecture that explicitly relates spectral degeneracy to hidden symmetry. This explicit relationship focuses on the hyperelliptic curves (see, e.g. \cite{CF2006}) that arise naturally in both settings, and it can be said to embody certain universal properties of the spectrum.

\section{Definitions and scope of models} 
\label{sec:preliminaries}

In this section, we introduce the two-photon asymmetric quantum Rabi model and some of the basic properties
used in this paper. We also introduce the asymmetric non-commutative harmonic oscillator (aNCHO),
a generalization of the quantum harmonic oscillator defined independently of physics, and give a short overview of the relations between these two families of models.

In addition to the standard $L^2(\R)$ realization of Hilbert spaces, we also consider the Fock (or Bargman) space 
\[
  \mathcal{F}(\C):= \Big\{f\in \mathcal{O}(\C)\,\big|\, \frac1\pi \int_\C |f(z)|^2e^{-|z|^2} dz < \infty\Big\},
\]
where $\mathcal{O}(\C)$ is the space of entire functions on $\C$. Then, $L^2(\R)$ and $\mathcal{F}(\C)$ are unitarily isomorphic by the
Bargman transform $\mathcal{B}$.
In the $L^2(\R)$ realization, the annihilation and creation operators acting on $L^2(\R)\otimes \C^2$
$$
a=\frac1{\sqrt2}(x+\partial_x), \qquad a^\dag=\frac1{\sqrt2}(x-\partial_x) \qquad \left(\partial_x:= \frac{d}{dx}\right),
$$
and we remark that $\mathcal{B}$ maps the annihilation and creation operators to $\partial_z$ and $z$, respectively, in the Fock space.

For certain parameters the models discussed in this paper have a $\Z_4$-symmetry operator. To unify the notation, let us define the operators
\[
  \hat{\mathcal{P}}_x = e^{\pi i a^\dag a/2} \sigma_x, \quad  \hat{\mathcal{P}}_z = e^{\pi i a^\dag a/2} \sigma_z,
\]
and we call them {\em reciprocity operators}. It is obvious from this definition that these are operators of 
fourth order. The second power of the reciprocity operators, that is,
\[
  \hat{\mathcal{P}}_x^2 = \hat{\mathcal{P}}_z^2 = \hat{\mathcal{P}}^2 := e^{\pi i a^\dag a} 
\]
is the {\em parity operator}. We use the notations above for both the $L^2(\R)$ and the Fock space $\mathcal{F}(\C)$ realizations without distinction as it should be obvious in the context.

\subsection{Two-photon quantum Rabi model}
\label{sec:tpQRM}

The main object of study of this paper is the two-photon asymmetric quantum Rabi model (2pAQRM), the model describing the interaction between a two level atom and a light field for the case where the interaction involves emission or absorption of two photons.

The Hamiltonian of the 2pAQRM $H_{\e}$ is given by
\begin{align}
  \label{eq:tpHamilt}
  H_{\e}= (a^\dag a + \frac12) + g ( a^2 + (a^\dagger)^2 ) \sigma_z + \Delta \sigma_x + \tau \e \sigma_z,
\end{align}
where $g>0$ is the coupling strength, $2\Delta$ the distance between the two levels and $\e \in \R$ is the bias parameter
and where we tacitly set the oscillator frequency $\omega = 1$ as before.  The auxiliary parameter $\tau = \sqrt{1-4g^2}$ is
fully determined by $g$ and since we consider the case $g < 1/2$  we have $0<\tau(< 1)$. In this case, the spectrum of
$H_{\e}$ consists only of eigenvalues. For other values of $g$ the spectrum also contains a continuous part, we refer to the reader to the analysis in  \cite{HS2024} for the details for the 2pQRM, which can be extended to the asymmetric version 2pAQRM.

We note that we made two nonstandard choices with respect to the usual definition of the Hamiltonian of the 2pAQRM. The first one is that we take $\tau \e$ as the coefficient of the bias parameter to simplify the derivation and definition of constraint polynomials for Juddian solutions in Section \ref{sec:deg} (see Remark \ref{rem:biascoeff}). The second one is a translation by $\frac{\omega}2\big(=\frac12\big)$ with respect to the usual definition to obtain a more natural expression of the polynomials appearing from the hidden symmetry operators in Section \ref{sec:symdeg}. Moreover, this type of translation appears naturally when the 2pAQRM is related to the asymmetric non-commutative harmonic oscillator, as we explain in Section~\ref{sec:NCHOp}. Since we are assuming $\tau < 1$ so that $H_{\e}$ has only a point spectrum, neither of the choices introduce any essential changes to the spectrum of the 2pAQRM.

To simplify the visualization of spectrum, it is convenient to normalize the spectral curves (or equivalently, the Hamiltonian) by $\tau=\sqrt{1-4g^2} \in \R$ as shown in Figure \ref{fig:tpAQRMeigencurves1}. In Section \ref{sec:symdeg}, there is further evidence that such a  normalization is natural according to the 
hidden symmetry and degeneracy picture of the asymmetric model.

\begin{figure}[ht]
  \centering
  \subfloat[$\e=0$]{
    \includegraphics[height=4.5cm]{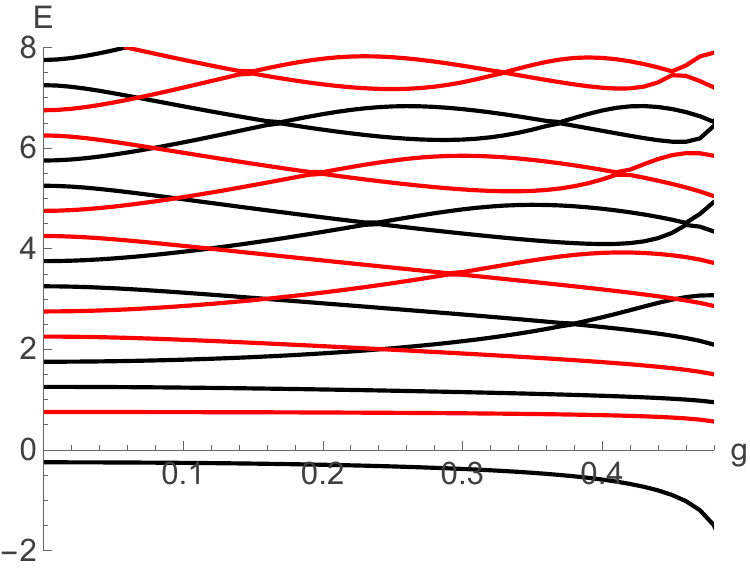}}
  ~ \qquad
  \subfloat[$\e=0.25$]{
    \includegraphics[height=4.5cm]{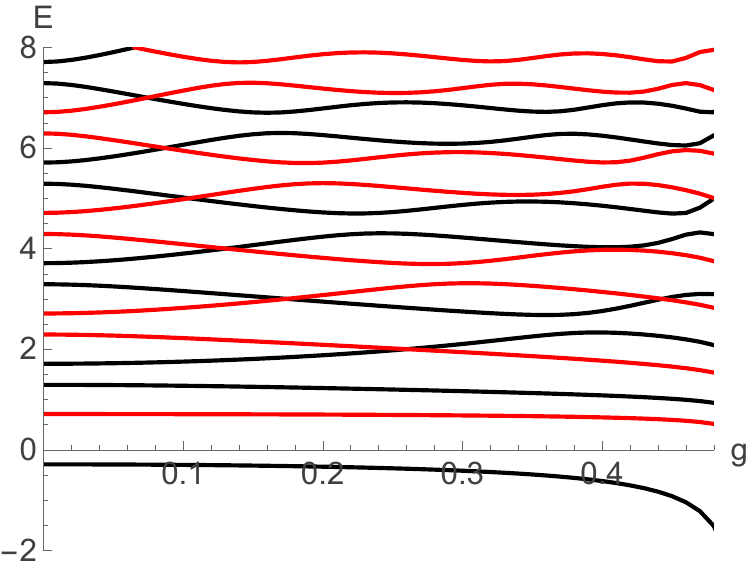}}
  \caption{Normalized spectral curves for the 2pQRM for $\Delta=3/4$ for $\e=0$ (2pQRM) and $\e= 1/4$. Black (resp. red) curves correspond to even (resp. odd) eigenvalue curves.}
  \label{fig:tpAQRMeigencurves1}
\end{figure}

The existence of symmetry operators, operators that commute with the 2pAQRM Hamiltonian, is of fundamental importance for the characterization of the spectrum and, as we see in this paper, they endows the model with a geometric structure based on algebraic curves.

For any $\e \in \R$, the Hamiltonian $H_{\e}$  commutes with the involutive operator $\hat{\mathcal{P}}_x^2$, thus we say the 2pAQRM has a $\Z_2$-symmetry . Since the a operator $\hat{\mathcal{P}}_x^2$
is a parity operator (in the usual sense) on the Hilbert space $\mathcal{H}$, it gives a classification of the spectrum of the 2pAQRM into even or odd according to the parity of the eigenfunctions.

For the symmetric version, i.e. $\e=0$, it is immediate to verify
that
\[
  [H_{0},  \hat{\mathcal{P}}_x]=0,
\]
thus in the 2pQRM the $\Z_2$-symmetry of $\hat{\mathcal{P}}_x^2$ is generated from the $\Z_4$-symmetry  $\hat{\mathcal{P}}_x$. The introduction of the bias term $\tau \e \sigma_z$ breaks the $\Z_4$-symmetry, that is, $H_\e$ no longer commutes with $\hat{\mathcal{P}}_x$. Such a symmetry only reappears in a nontrivial way for special values of the bias parameter but it no longer generates the $\Z_2$-symmetry (cf. Section \ref{sec:hiddenS}).

We recall the classification of degenerate eigenvalues based on the $\Z_2$-symmetry described in the introduction. We say that a degenerate eigenvalue is of shared-parity type if both of the eigenfunctions are either even or odd. On the other hand, we say that a degenerate eigenvalue is of opposite-type one eigenvalue is even and the other is odd.  We show shared-parity spectral crossings in Figure \ref{fig:tpAQRMeigencurves1}(a) and we note that the bias term has the effect of making these type of degeneracies disappear, as shown in Figure \ref{fig:tpAQRMeigencurves1}(b). Shared-parity spectral degeneracies appear once more for integer $\e$ as shown in Figure \ref{fig:tpAQRMeigencurves2}. On the other hand, we note that the opposite-parity spectral crossings may appear for any value of $\e \in \R$.

\begin{figure}[ht]
  \centering
    \includegraphics[height=4.5cm]{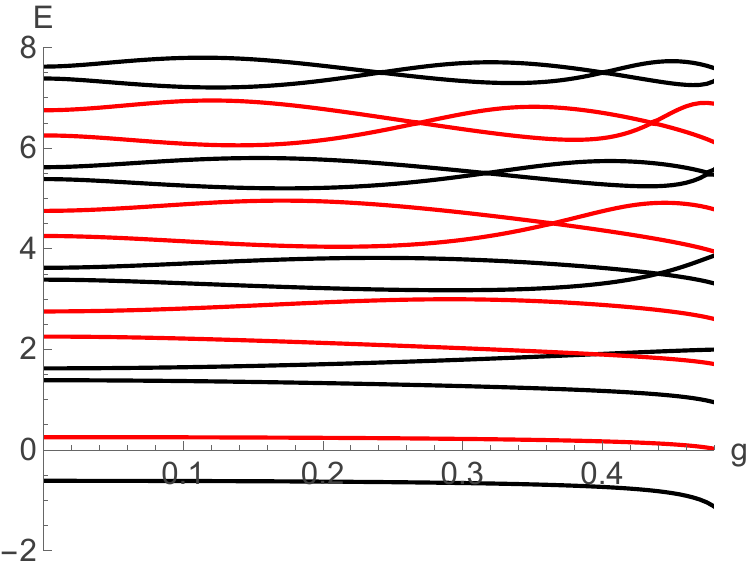}
  \caption{Normalized spectral curves for the 2pQRM for $\Delta=3/4$ for $\e=1$. Note the return of the shared-parity degeneracies.}
  \label{fig:tpAQRMeigencurves2}
\end{figure}

One of the objectives of this paper is the complete characterization of the degeneracies of the 2pAQRM in Section \ref{sec:degstruct}. To do this, we introduce a classification of the spectrum according to the form of the eigenvalues
and the type of solution.

\begin{dfn}
  An eigenvalue of the 2pAQRM is said to be
  \begin{itemize}
  \item \textbf{Exceptional} if it is of the form $\lambda = \tau (2 N + \frac12 \pm \e)$ for a non-negative integer $N$,
  \item \textbf{Regular} if it is not exceptional. 
  \end{itemize}

  Moreover, we say that an exceptional eigenvalue $\lambda$ is called
  \begin{itemize}
  \item \textbf{Juddian} if there is a realization of the eigenvalue problem where its eigenfunction can be written as a polynomial (i.e. it is quasi-exact),
  \item \textbf{Non-Juddian} if it is not Juddian.  
  \end{itemize}
\end{dfn}

The definition for Juddian eigenvalue follows the usual convention used for the AQRM in mathematical physics. However, we note that it depends on a particular realization of the Hamiltonian (i.e. the associated  ODE picture) and it may be difficult to verify in practice since in different realizations a quasi-exact solution may not be quasi-exact in other. In Section \ref{sec:degstruct}, we revisit this problem from the point of the view of the Heun ODE picture associated to the 2pAQRM.

\begin{rem}
  The 2pAQRM is not a generalization of the AQRM in the usual sense, that is, no choice of parameters for the 2pAQRM Hamiltonian results in the QRM Hamiltonian except for the trivial cases $g\to0$ or $\Delta\to0$. Nevertheless, there is a nontrivial relation, called covering, between the two models based on confluence of the ODE pictures associated to the eigenvalue problems of the models (see \cite{RBW2023} for the original definition and \cite{N2024} for the case of the 2pAQRM).

  In the setting of covering models, the $\Z_2$-symmetry for the AQRM may be considered as coming from the
  $\Z_4$-symmetry of the 2pAQRM via the covering relation, and other properties of the models may be studied in
  this way. 
\end{rem}

\subsection{Non-commutative harmonic oscillators and quantum Rabi models }
\label{sec:NCHOp}

The non-commutative harmonic oscillator (NCHO) was originally introduced in \cite{PW2001} as a non-commutative generalization of the quantum harmonic
oscillator
\[
  H_{\rm{HO}}:= a^\dag a+\frac12 = -\frac12 \frac{d^2}{d x^2}  + \frac12 x^2,
\]
for applications on number theory and differential equations. The NCHO and its asymmetric version were defined and originally studied with no regard to mathematical physics or quantum interaction models. 

Let us recall the definition of the Hamiltonian of the aNCHO, introduced as $\eta$-shifted NCHO in \cite{RBW2023}. For the arithmetical properties of the NCHO we refer the reader to \cite{KW2023}.

\begin{dfn}
  \label{dfn:aNCHO}
  The asymmetric non-commutative harmonic oscillator (aNCHO) is the model with Hamiltonian $Q^{(\eta)}=Q^{(\eta)}_{\alpha,\beta}(x,D)$ acting on $L^2(\R)\otimes \C^2$ given by
  \begin{align*}
      Q^{(\eta)}_{\alpha,\beta}(x, D) &:=  \left( -\frac12 \frac{d^2}{d x^2}  + \frac12 x^2 \right)\bA +  
                            \bJ  \left( x \frac{d}{d x} + \frac12 \right) + 2 \eta i \det(\bA \pm i\bJ)^{\frac12} \bJ \\
    &=\left(a^\dag a+\frac12 \right) \bA -  
                            \frac{i}2\left(a^2-(a^\dag)^2 \right)\sigma_y + 2 \eta \det(\bA \pm \sigma_y)^{\frac12} \sigma_y 
\end{align*}
 where $\eta \in \R$, and 
 \[
   \bA:=\begin{bmatrix}
     \alpha & 0 \\
     0 & \beta
   \end{bmatrix}, \qquad 
   \bJ:=  \begin{bmatrix}
     0 & -1 \\
     1 & 0
   \end{bmatrix}. 
 \]
\end{dfn}

In this paper, we assume $\alpha \beta>1$ in order to guarantee that $Q^{(\eta)}_{\alpha,\beta}(x,D)$ has only discrete spectrum and may assume $\alpha, \beta>0$ (see \cite{RBW2023}). When $\eta=0$, $Q^{(\eta)}_{\alpha,\beta}$ is just the NCHO, the original non-commutative harmonic oscillator introduced in \cite{PW2001}. When $\alpha=\beta$, it is known that $Q^{(0)}_{\alpha,\alpha}$ is unitarily equivalent to a couple of harmonic oscillators, whence the eigenvalues are easily calculated as $\{\sqrt{\alpha^2-1}(n+\frac12)\,|\, n\in \Z_{\geq0}\}$ with multiplicity $2$. 

As it is apparent, the Hamiltonian of the aNCHO resembles the ones used for quantum interactions and for a long time a relation between the two models was expected to exist. The first connection was the covering relation between the aNCHO and the AQRM \cite{RBW2023} (see also \cite{W2015}) via a confluence process on the corresponding ODE pictures. Recently, it was discovered in \cite{N2024} that the eigenvalue problem of the 2pAQRM is equivalent to the eigenvalue problem of the aNCHO (see also \cite{HS2024}). 
These discoveries suggests that theory developed for one of the models may be used in another, as we demonstrate in this paper with the Heun picture of the 2pAQRM.

Let us recall the main result of \cite{N2024}, giving a certain equivalence between the eigenvalue problems of the 2pAQRM and the aNCHO.

\begin{thm}[\cite{N2024}] \label{them:N2024}
  Let $\alpha,\beta >0$, $\eta \sqrt{\alpha \beta -1}$ and $\lambda \in \R$, then by setting
  \begin{gather*}
    g:=\frac{1}{2\sqrt{\alpha \beta}}, \qquad \e := - 2 \eta, \\ 
    \mu := \frac{\lambda}2 \left(\frac1{\alpha} +\frac1{\beta} \right), \qquad \Delta := \frac{\lambda}2 \left(\frac1{\alpha} - \frac1{\beta} \right), \\
    M :=
    \begin{bmatrix}
      \sqrt{\alpha} e^{-\pi i/4} & 0 \\
      0 & \sqrt{\beta} e^{\pi i/4}
    \end{bmatrix},
  \end{gather*}
  we see that for $f(x) \in  L^2(\R) \otimes \C^2$, the eigenvalue problem
  \[
    Q^{(\eta)}_{\alpha,\beta} f = \lambda f,
  \]
  is equivalent to
  \[
    H_{\e} e^{\frac{\pi i}4 H_{\rm{HO}}} M f = \mu e^{\frac{\pi i}4 H_{\rm{HO}}} M f.
  \]
\end{thm}

The transformation in Theorem \ref{them:N2024} is a gauge transformation and is thus not an equivalence in the usual sense. Indeed, since the definition of $\Delta$ includes the eigenvalue $\lambda$, it is not possible to relate a single aNCHO Hamiltonian to a single 2pAQRM Hamiltonian. Because of this, in \cite{HS2024} this relation is called a ``fiber decomposition'' of the aNCHO into 2pAQRMs. In the same paper, the authors realize the transformation between the two models as a unitary transformation in an appropriate Hilbert, however since this Hilbert space is defined in terms of the parameters (including the eigenvalue) it does not give a unitary equivalence between the models.

\begin{rem}
As is well known, the oscillator representation of the Lie algebra $\mathfrak{sl}_2(\R)$ does not lift to the Lie group $SL_2(\R)$ but instead defines a representation of its double cover, the metaplectic group $Mp_2(\R)$. In the same context, when dealing with the AQRM or 2pAQRM with a general bias parameter $\e \in \R$, it is natural—from the perspective of group representation theory—to view the object as a character of the universal covering group $\widetilde{SL_2(\R)}$ of $SL_2(\R)$ (see, e.g., \cite{Puk1964, K2000} and also \cite{N2025} for a relevant discussion). We develop the details of this picture in a forthcoming paper by the authors  \cite{RBW2027}.
\end{rem}

As a consequence of Theorem \ref{them:N2024}, we see that the properties of the spectrum of the 2pAQRM can be elucidated directly from the properties of the aNCHO, and vise versa. For example, Theorem \ref{them:N2024} gives an immediate proof that almost all the eigenvalues of the 2pAQRM have multiplicity bounded above by $2$. The result for the complete spectrum does not follow directly from Theorem \ref{them:N2024} since the inequality $|\mu| > \Delta$ is implicitly imposed. Indeed, rewriting the relations between the parameters of the two models in \ref{them:N2024} as
\begin{align}
\label{para-correspondence2}
  \alpha = \frac{1}{2 g} \sqrt{\frac{\mu-\Delta}{\mu+\Delta}},\quad \beta = \frac{1}{2 g} \sqrt{\frac{\mu+\Delta}{\mu-\Delta}}, \quad  \lambda = \frac{1}{2 g} \sqrt{\mu^2 - \Delta^2}. 
\end{align}
Thus, the necessity of condition $|\mu| > \Delta$ is evident. We return to the question of the multiplicity in Section \ref{sec:HeunPicture}.

Another important property of the aNCHO is that its eigenvalue problem is equivalent to the existence of solutions of certain Heun (second order) differential equations. In addition of being a lower degree differential equation system than the one using a direct derivation in the Fock space, it allows us to use the theory of monodromy to study the properties of solutions based on the singularity data. In Section \ref{sec:HeunPicture} we develop an equivalent theory for the 2pAQRM to be later used to characterize the degenerate spectrum in Section \ref{sec:degstruct}.

Recall the classification of the spectrum of the aNCHO for comparison with that of 2pAQRM. Denoting the spectrum of the aNCHO by ${\rm Spec}(Q^{(\eta)}_{\alpha,\beta})$, we then have
\begin{align*}
  {\rm Spec}(Q^{(\eta)}_{\alpha,\beta}) &= \Sigma_{0} \cup \Sigma_{\infty} \\
                        &= \Sigma^{+}_{0} \cup \Sigma^{-}_{0} \cup \Sigma^{+}_{\infty} \cup \Sigma^{-}_{\infty},
\end{align*}
where
\begin{itemize}
\item $\Sigma_{0}$ is the set of eigenvalues of ``finite type'': eigenvalues with eigenfunctions consisting of finite sums with respect to the Hermite functions basis, that is, elements of $L^2(\R)_{{\rm fin}}\otimes \C^2$,
\item $\Sigma_{\infty}$ is the set of eigenvalues of ``infinite type'': eigenvalues that are not of finite type. 
\item $\Sigma_{i}^{+}$ (resp. $\Sigma_{i}^{-}$) for $i \in \{0,\infty\}$ are the eigenvalues having even (resp. odd) eigenfunctions.
\end{itemize}

We remark that the definition of the eigenvalues of finite type does not depend on the choice of basis.
Moreover, while they appear to be similar, the notion of finite type is different to that of Juddian eigenvalues
for the 2pAQRM. For instance, the shared--parity degenerated eigenvalues for the aNCHO correspond to eigenvalues  $\lambda \in \Sigma_0^{\pm} \cap \Sigma_\infty^{\pm}$, while the opposite-parity degeneracies correspond to eigenvalues $\lambda\in \Sigma_\infty^{+} \cap \Sigma_\infty^{-}$. In Section \ref{sec:degstruct2}, using the Heun picture, we clarify the relation between the classification of the spectrum
for the two models.

Although we mainly limit ourselves in this paper to the properties of the aNCHO that are helpful for the clarification of the symmetry and degeneracy aspects of the 2pAQRM, the relation between the two models can be explored in other directions. Applications to representation theory are the subject of a forthcoming sequel to this paper \cite{RBW2027}.

\subsection{A view of quantum interaction models and their relations}
\label{sec:relationsModels}

We conclude this section by remarking that the relation between the aNCHO and the 2pAQRM can be extended to other
models, and when combined with the concept of covering models \cite{RBW2023} we can begin to describe a panorama of the relation between models with the quantum Rabi model at the center.

For instance, the non-unitary equivalence between the 2pAQRM and the aNCHO is remarkably similar to the relation used in \cite{BCEJT2024} to reduce the analysis of the spectrum quantum Rabi-Stark model (QRSM) to that of a certain QRM. In particular, similar to the relation in Theorems \ref{them:N2024} and \eqref{para-correspondence2}, the parameters of the new QRM depend on the eigenvalues of the original QRSM.

Let us recall that the notion of covering model implies that one can obtain one model from the other by a confluence relation of the corresponding (Heun) ODE pictures \cite{RBW2023,N2024}. 

In Figure \ref{fig:paths2},  the horizontal arrows correspond to non-unitary equivalence (cf. \cite{N2024} or
\cite{HS2024} and \cite{BCEJT2024}) and vertical (and diagonal) arrows correspond to covering relations.  Dashed lines correspond to relations currently not described in the literature but that due to the structure of the models, are  likely to exist. In the diagram, 1pNCHO is the one photon NCHO defined by Hiroshima and Shirai \cite{HS2024}, 2pQRSM is a two photon generalization of the quantum Rabi-Stark model defined in a natural way. In all cases, asymmetric versions are obtained by the introduction of an appropriate bias term.

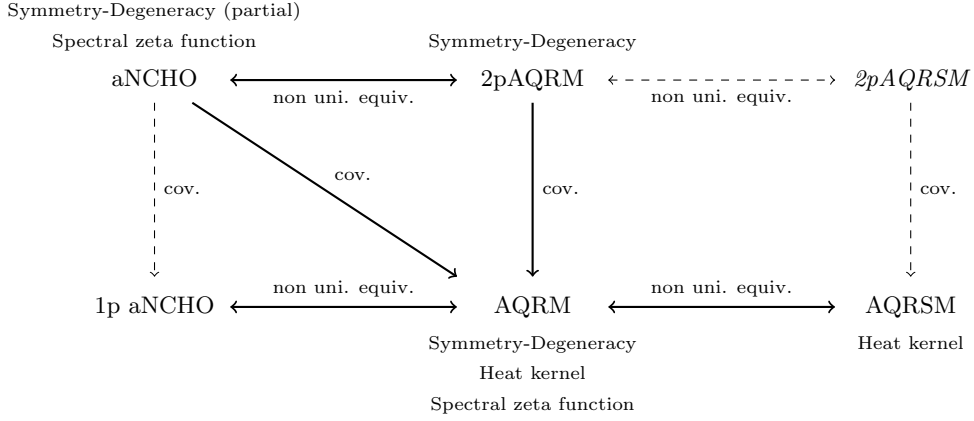
\begin{figure}[!ht]
  \centering
  \begin{tikzpicture}[domain=0:4]
    

    \node at (0,0) { {\small\begin{tabular}{c}  2pAQRM \end{tabular}}};
    
    \node at (0,0.5) { {\scriptsize\begin{tabular}{c} Symmetry-Degeneracy \end{tabular}}};

    \node at (5,0){ {\small\begin{tabular}{c}  {\em 2pAQRSM} \end{tabular}}};

    \node at (-5,0.9) { {\scriptsize\begin{tabular}{c} Symmetry-Degeneracy (partial) \end{tabular}}};   
    \node at (-5,0.5) { {\scriptsize\begin{tabular}{c} Spectral zeta function \end{tabular}}};

    \node at (-5,0) { {\small\begin{tabular}{c}  aNCHO \end{tabular}}};


    \node at (0,-3.5) { {\scriptsize\begin{tabular}{c} Symmetry-Degeneracy \end{tabular}}};
    \node at (0,-3.9) { {\scriptsize\begin{tabular}{c} Heat kernel \end{tabular}}};
    \node at (0,-4.3) { {\scriptsize\begin{tabular}{c} Spectral zeta function \end{tabular}}};

    \node at (0,-3) { {\small\begin{tabular}{c}  AQRM \end{tabular}}};

    \node at (5,-3.5) { {\scriptsize\begin{tabular}{c} Heat kernel \end{tabular}}};

    \node at (5,-3){ {\small\begin{tabular}{c}   AQRSM \end{tabular}}};


    \node at (-5,-3) { {\small\begin{tabular}{c}  1p aNCHO \end{tabular}}};
    

    \draw[<->,dashed] (1,0) --  (4,0) node[midway,below] {\scriptsize{non uni. equiv.}}; 

    \draw[<->,thick] (-1,0) --  (-4,0) node[midway,below] {\scriptsize{non uni. equiv.}}; 


    \draw[<->,thick] (1,-3) --  (4,-3) node[midway,above] {\scriptsize{non uni. equiv.}}; 

    \draw[<->,thick] (-1,-3) --  (-4,-3) node[midway,above] {\scriptsize{non uni. equiv.}}; 


    \draw[->,thick] (0,-0.3) --  (0,-2.6) node[midway,right] {\scriptsize{cov.}}; 

    \draw[->,dashed] (-5,-.3) --  (-5,-2.6) node[midway,right] {\scriptsize{cov.}} ;

   \draw[->,thick] (-4.5,-.3) --  (-1,-2.6) node[midway,above right] {\scriptsize{cov.}}; 
   
    \draw[->,dashed] (5,-.3) --  (5,-2.6) node[midway,right] {\scriptsize{cov.}}; 

  \end{tikzpicture}
  \caption{Schematic diagram of models related with the quantum Rabi model.}
  \label{fig:paths2}
\end{figure}

We expect that properties of the models in the scheme above should be reflected, with certain modifications, in the other models, for instance, the existence of symmetry operators for the asymmetric versions or Juddian solutions (see the case of the aNCHO and AQRM in \cite{RBW2023}). Additionally, relations between physical properties and arithmetical properties (i.e. special values at integer points, see \cite{KW2023} and references therein for the NCHO) are other potential research directions. We note that  hidden symmetry operators have already been studied heuristically in other quantum interaction models \cite{LLMB2022}.

It would also be an interesting problem to develop a similar picture for the Spin Boson model \cite{HHF2014}, which can be considered as a generalization of the QRM in quantum field theory, or the Kondo effect  \cite{K2012}. In addition, the explicit correlation function for the Kondo model about a magnetic impurity \cite{AYH1970} is remarkably similar to the analytical formula of the heat kernel developed for the quantum Rabi model \cite{RW2020hk,R2023}


\section{Eigenvalue problems of the 2pAQRM}
\label{sec:deg}

In this section, we consider the eigenvalue problem, the time independent Schr\"odinger equation, associated with the Hamiltonian of the 2pAQRM. By using an appropriate unitary transformation we consider a couple of equivalent Hamiltonians which allow a simpler analysis of the exceptional solutions. More importantly, from these Hamiltonians we can show that the eigenvalue problem is equivalent to a Heun ordinary differential equation (i.e. a Fuchsian ODE with four regular singular points) with certain regularity conditions around the singularities.

As a preparation, let us recall the definition of the Bogoliubov transforms, which is 
a standard technique in the study of models of quantum interaction (see, e.g. \cite{E2019}). 

We consider the Hamiltonian $H_{\e}$ as an element of ${\rm Mat}_2[\C[a,a^\dag]]$ where $\C[a,a^\dag]$ is the Weyl algebra generated by the operators $a$ and $a^\dag$. Then, a Bogoliubov transform on the Weyl algebra is the extension of the polynomial algebra $\C[a,a^\dag]$ transform given by
\begin{equation}
\begin{bmatrix}
    b \\ b^{\dag}
\end{bmatrix}
 = \bm{M} 
 \begin{bmatrix}
    a \\ a^{\dag}
\end{bmatrix}
\end{equation}
for $\bm{M} \in SL_2(\C)$. Under these conditions, the transformed variables $b,b^{\dag}$ satisfy $[b,b^\dag]=1$. The Bogoliubov  transforms defined in this way are actually automorphisms of the Weyl algebra $\C[a,a^\dag]$ (see \cite{D1968} Théor\`eme 8.10, and also \cite{ML1984}). 

For the unitary case, the Bogoliubov transform may also be realized on the Hilbert space in which $a,a^\dag$ act by the unitary operator 
\[
  I_\theta[\psi](x) = e^{\theta/2} \psi(e^\theta x) = \exp(\tfrac{\theta}2 (a^2 - (a^{\dag})^2))\psi(x)
\]
for a parameter $\theta \in \R$. It is easy to verify that 
\[
  I_\theta a I_{\theta}^{-1} = \ch(\theta) a + \sh(\theta) a^\dag,
\]
and
\[
  I_\theta a^\dag I_{\theta}^{-1} = \sh(\theta) a + \ch(\theta) a^\dag,
\]
showing the equivalence to the definition above. In this paper, we make use of both approaches depending of the situation. In the context of quantum interactions, the operator $I_\theta[\psi]$ is sometimes called a squeezing transformation (see, e.g., Braak \cite{B2023}).

\subsection{Approach to solutions for the 2pAQRM}
\label{sec:solutions}

The first step for the analysis of the eigenvalues of the 2pAQRM is to consider unitary equivalent Hamiltonians,
realized by  unitary Bogoliubov transformations. Here, we follow the approach of Emary and Bishop \cite{EB2002} in the form used by Braak \cite{B2023} and Xie and Chen \cite{XC2021} with some modifications.

Using the Bogoliubov transform $I_\theta$ with parameter $\theta$, we verify directly that
\[
  I_\theta a^\dag a I_{\theta}^{-1} = (2 \ch(\theta)^2 - 1) a^t a + (\ch(\theta)^2 - 1) + \ch(\theta) \sh(\theta)(a^2+ (a^\dag)^2)
\]
and
\[
  I_\theta (a^2 + (a^{\dag})^2)I_{\theta}^{-1}  = (2 \ch(\theta)^2 -1 ) (a^2+ (a^\dag)^2) + 2 \ch(\theta) \sh(\theta) (1 + 2 a^\dag a).
\]

By setting
\[
  \theta = \log\left(\frac{1-2g}{1+2g} \right)^{\tfrac14},
\]
the Hamiltonian $I_\theta  H_{\e} I_{\theta}^{-1}$ in matrix form is given in
\begin{equation}
  \label{eq:Hamiltonian1}
  H_{A}^{(\e)} =
  \begin{bmatrix}
    \tau (a^\dag a + \frac12) + \tau \e & \Delta \\
    \Delta & \frac{1}{\tau}( (1+4g^2)(a^\dag a + \frac12) - 2g (a^2+ (a^\dag)^2))- \tau \e 
 \end{bmatrix}, \tag{A}
\end{equation}
with $\tau = \sqrt{1-4g^2}$. On the other hand, with $\theta = -\log\left(\frac{1-2g}{1+2g} \right)^{\tfrac14}$, the 
Hamiltonian $I_\theta  H_{\e} I_{\theta}^{-1}$ is given by
\begin{align}
  \label{eq:Hamiltonian2}
  H_{B}^{(\e)} =
    \begin{bmatrix}
\frac{1}{\tau}( (1+4g^2)(a^\dag a + \frac12) + 2g (a^2+ (a^\dag)^2)) + \tau \e & \Delta \\
    \Delta &  \tau (a^t a + \frac12) - \tau  \e
  \end{bmatrix}. \tag{B}
\end{align}

Then, it follows that the eigenvalue problem for \eqref{eq:Hamiltonian1} with the eigenvalue $\lambda$ and eigenfunctions $\phi_1,\phi_2 \in \mathcal{F}(\C)$ is given by
\begin{equation}
\label{eq:ODE_A}
\begin{aligned}
  \tau (z \partial_z  + \tfrac12)\phi_1 + \tau ( \e - \tfrac{\lambda}{\tau}) \phi_1 + \Delta \phi_2 &= 0,  \\
  \tau \Delta \phi_1 -2g (\partial_z^2 + z^2) \phi_2 + (1+4g^2) (z \partial_z +\tfrac12)\phi_2 - \tau^2( \e+\tfrac{\lambda}{\tau})\phi_2 &=0.
\end{aligned}
\end{equation}
 The ODE system is clearly parity invariant, therefore we consider solutions of the type
\[
  \phi_1(z) = \sum_{n=0}^\infty a^{(\rho)}_n z^{2n+\rho}, \qquad \phi_2(z) = \sum_{n=0}^\infty b^{(\rho)}_n z^{2n+\rho}, 
\]
for $\rho=0,1$, we obtain the relations
\begin{align}
  \label{eq:relA1}
  (\tau(2n+ \rho + \frac12 + \e)-\lambda) a^{(\rho)}_n =  \Delta b^{(\rho)}_n,
\end{align}
and
\begin{align}
  \label{eq:recurr1}
  &(2(n+1)+\rho)(2(n+1)+\rho-1) b_{n+1}^{(\rho)} = \nonumber \\
  &\qquad \frac{1}{2g}\left[(1+4g^2)(2n+\rho+\tfrac12) - \tau^2\left(\frac{\lambda}{\tau} + \e \right) - \frac{\Delta^2}{(2n+\rho+\frac12 + \e) - \frac{\lambda}{\tau} } \right] b_{n}^{(\rho)}  -  b_{n-1}^{(\rho)}. 
\end{align}

Let us note here that the exceptional solutions, that is, those that arise as poles of the G-function (see \cite{XC2021}) correspond to eigenvalues of the form $\lambda = \tau(2 n + \rho + \frac12 + \e)$ for $n\geq 0$.

For the case of \eqref{eq:Hamiltonian2}, with solutions
\[
  \psi_1(z) = \sum_{n=0}^\infty (-1)^n f^{(\rho)}_n z^{2n+\rho}, \qquad \phi_2(z) = \sum_{n=0}^\infty (-1)^n e^{(\rho)}_n z^{2n+\rho}, 
\]
we obtain the relations
\begin{align}
  \label{eq:relB1}
  (\tau(2n+ \rho +\frac12 - \e)-\lambda)e^{(\rho)}_n =  \Delta f^{(\rho)}_n,
\end{align}
and
\begin{align}
  \label{eq:recurr2}
  &(2(n+1)+\rho)(2(n+1)+\rho-1) f_{n+1}^{(\rho)} = \nonumber\\
  & \qquad \frac{1}{2g}\left[(1+4g^2)(2n+\rho+\tfrac12) - \tau^2\left(\frac{\lambda}{\tau}- \e\right)- \frac{ \Delta^2}{ (2n+\rho+\frac12 - \e) - \frac{\lambda}{\tau}} \right] f_{n}^{(\rho)} -  f_{n-1}^{(\rho)}.
\end{align}
Notice that $\e \mapsto -\e$ in \eqref{eq:recurr2} gives \eqref{eq:recurr1}. In this case, the exceptional solutions correspond to eigenvalues of the form $\lambda = \tau (2 N + \rho + \frac12 - \e) $ with $N \in \Z_{\geq 0}$.

\begin{rem}
  It is worth remarking the relation between the two Hamiltonians $ H_{A}^{(\e)}$ and $H_{B}^{(\e)}$ in terms of the
  reciprocity operator $\mathcal{P}_x$. A simple computation shows that
  \begin{equation}
    \label{eq:HaHbrel}
    H_{A}^{(\e)} =  \mathcal{P}_x H_{B}^{(-\e)} \mathcal{P}_x^{-1}.
  \end{equation}
  For $\e=0$, this is just the expression of the $\Z_4$-symmetry and in the general case it does not give a unitary   equivalence between the two Hamiltonians. Indeed, if this were the case, then the 2pAQRM would have an obvious $\Z_4$-symmetry. 
\end{rem}

The case of degenerate solutions arising from Juddian solutions, that is, exceptional solutions corresponding to
polynomials (or quasi-exact) solutions, is treated in Section \ref{sec:deg}. 

\subsection{Heun picture of the 2pAQRM}
\label{sec:HeunPicture}

A basic result on the analysis of the aNCHO is the equivalence of its eigenvalue problem with a pair of Heun ODE according to the parity (see \cite{W2015} and \cite{RBW2023}). Similarly, it is well-known that the eigenvalue problem of the AQRM is equivalent to a pair of Heun ODE (originally considered in \cite{O2001CMP,O2004} for the odd parity, and completed in  \cite{W2017}). In light of the relation of the 2pAQRM with the aNCHO, it is reasonable to expect a similar result in this setting. This is indeed the case, and in this section we describe the Heun picture of the 2pAQRM.

To obtain the Heun ODE picture, we make use of a generalization of the oscillator representation of $\mathfrak{sl}_2$. We refer the reader to any standard reference for the background material, e.g. \cite{HT1992}.

We denote the standard basis of $\mathfrak{sl}_2$ by
\[
 E_+:=  \begin{bmatrix}
    0 & 1 \\
    0 & 0
  \end{bmatrix}
  \quad 
  H:=  \begin{bmatrix}
    1 & 0 \\
    0 & -1
  \end{bmatrix}
 \quad 
  E_-:=  \begin{bmatrix}
    0 & 0 \\
    1 & 0
  \end{bmatrix}.
\]

For $a \in \C$, define the representation $\pi'_a$ of $\mathfrak{sl}_2$ on the space $y^{a-1} \C[y]$ given on generators by
\begin{align}
\label{Rep_a'}
  \pi_a'(H) = y \frac{d}{d y} + \frac12, \quad \pi_a'(E_+) = \frac{y^2}{2}, \quad \pi'_a(E_-) = -\frac12 \frac{d^2}{d y^2} + \frac{(a-1)(a-2)}{2} \cdot \frac{1}{y^2}.
\end{align}
Clearly, for $a=1,2$, the representation $\pi_{a}'$ is isomorphic to the oscillator representation $\mathcal{F}(\C)$ of $\mathfrak{sl}_2$ for . We use the same symbol $a$ for the parameter of $\pi'_a$ and the operator $a$, but the difference should be clear from the context.

The Hilbert space $\overline{\C[y]}$, equipped with the Fisher inner product, is isometric to $L^2(\C)$ and $\mathcal{F}(\C)$ via  intertwiners between the corresponding oscillator representations, so the eigenvalue problems in these spaces is equivalent (see \cite{RBW2023} for details). Here, we remark that the representation $\pi_a'(H)$ is equivalent to the representation used in \cite{Puk1964} for the (second series of) representations of the simply-connected covering group of $SL(2,\R)$ and the representation $D_r$ appearing in \cite{K2000}. 

First, we consider elements of $\mathcal{U}(\mathfrak{sl}_2)$ that captures the eigenvalue problems associated with the 2pAQRM (namely, those determined by \eqref{eq:Hamiltonian1} and \eqref{eq:Hamiltonian2}). First, we consider elements $\mathcal{S}_X^{(\e)} = \mathcal{S}_X^{(\e)}(\lambda,g,\Delta) \in \mathcal{U}(\mathfrak{sl}_2)$ for $X \in \{A,B\}$ given by
\begin{align*}
  \label{eq:Rop}
  \mathcal{S}_A^{(\e)} &=  \left[ 4g (E_+ - E_-) - (1+4g^2) H  + \tau^2 \left( \frac{\lambda}{\tau} + \e \right) \right] \left(H - \frac{\lambda}{\tau} + \e \right) + \Delta^2, \\
  \mathcal{S}_B^{(\e)} &=  \left[ 4g (E_+ - E_-) + (1+4g^2) H  - \tau^2 \left(\frac{\lambda}{\tau}  - \e \right) \right] \left(H - \frac{\lambda}{\tau} - \e \right) - \Delta^2. \nonumber
\end{align*}
then the following proposition follows directly.

\begin{prop}
  \label{prop:EigensystemPi}
  For $a = 1,2 $ and $X \in \{A,B\}$, the eigenvalue problem for $H_X^{(\e)}$ (defined in equation \eqref{eq:Hamiltonian1} or \eqref{eq:Hamiltonian2}) for an eigenvalue $\lambda$ is equivalent to the solution of
  \[
    \pi'_a(\mathcal{S}_X^{(\e)}) \phi = 0.
  \]
\end{prop}

\begin{proof}
  It is enough to consider the system of equations $(H_X^{(\e)} - \lambda )\psi = 0$ for $\psi = ^{t} (\psi_1,\psi_2)$ where $\psi_i$ are
  entire functions in $\mathcal{F}(\C)$. Then, we reduce the system to a single second order ODE for $\psi_2$ and
  the result follows by comparison with $\pi'_a(\mathcal{S}^{(\e)}) \phi = 0$ and the existence of the intertwining isometries
  between the two spaces.
\end{proof}

The relation between the elements $\mathcal{S}_A^{(\e)}$ and $\mathcal{S}_B^{(\e)}$ is exhibited by 
\[
   \pi'_a(\mathcal{S}_A^{(-\e)}(\lambda,-g,\Delta)) \phi =  -\pi'_a(\mathcal{S}_B^{(\e)}(\lambda,g,\Delta)) \phi.
\]
Therefore, in what follows we consider only the element  $\mathcal{S}_A^{(\e)}$ with no loss of generality.

\begin{rem}
  In the proof of Proposition \ref{prop:EigensystemPi}, it is possible to reduce the equation $(H_X^{(\e)} - \lambda )\psi = 0$
  in a different way. In this case, the same result holds with respect to the elements
  \begin{align*}
  \widetilde{\mathcal{S}_A^{(\e)}} &=  \left(H - \frac{\lambda}{\tau} + \e \right) \left[ 4g (E_+ - E_-) - (1+4g^2) H  + \tau^2 \left( \frac{\lambda}{\tau} + \e \right) \right]  + \Delta^2, \\
  \widetilde{\mathcal{S}_B^{(\e)}} &=  \left(H - \frac{\lambda}{\tau} - \e \right) \left[ 4g (E_+ - E_-) + (1+4g^2) H  - \tau^2 \left(\frac{\lambda}{\tau}  - \e \right) \right]  - \Delta^2
  \end{align*}
  of  $\mathcal{U}(\mathfrak{sl}_2)$. We do not consider the operators $\widetilde{\mathcal{S}_X^{(\e)}}$ in this paper, however, we
  remark that the analogous $\mathcal{U}(\mathfrak{sl}_2)$ elements have a role in the theory of the aNCHO \cite{RBW2023}.
\end{rem}

Note that $\pi'_a(\mathcal{S}_A^{(\e)}) \phi = 0$ is a differential operator of third order, so it does not give directly the Heun picture for the 2pAQRM. To overcome this problem, we follow the method of \cite{RBW2023} and use an intertwiner with a another representation of $\mathfrak{sl}_2$.

For $a \in \R$, a representation $\varpi_a$ on $\C[z,z^{-1}]$ is defined on generators by
\[
  \varpi_a(H) = z \partial_z + \frac12, \qquad \varpi_a(E_+) = \frac12 z^2 (z \partial_z + a), \qquad \varpi_a(E_-) = - \frac1{2 z^2} (z \partial_z + 1- a )
\]
For $a=1,2$, $\mathcal{L}_a$ is a quasi-intertwiner of $\pi'_a$ and $\varpi_a$ in the sense that
\begin{align*}
  \mathcal{L}_a \pi'_a (X) &= \varpi_a(X) \mathcal{L}_a , \qquad (\text{ for } X = H,E_+),  \\
  (\mathcal{L}_1 \pi' (E_-))(z) &= \varpi_1(E_-) (\mathcal{L}_1 u(z)) + u'(0)/(2z),   \\
  (\mathcal{L}_2 \pi' (E_-))(z) &= \varpi_2(E_-) (\mathcal{L}_2 u(z)) - u(0)/(2z^2).
\end{align*}

Then, by setting $\omega = \frac{z^2}{2g}$ and $t = \frac{1}{4g^2}$, by a direct computation we note that
\[
  z^{1-a} \varpi_a(\mathcal{S}_A^{(\e)}) z^{a-1} = 16 g^2 \omega (\omega-1)(\omega-t) \Lambda_a,
\]
where the Heun ODE operator $\Lambda_a$ is given by
  \begin{align*}
    \Lambda_a &= \frac{d^2}{d \omega^2}  
          +  \left( \frac{ \frac12( a - \frac{\lambda}{\tau} + \frac32 + \e)}{2 \omega}  + \frac{ \frac12(a - \frac{\lambda}{\tau} -\frac12 - \e) }{(\omega -1)} + \frac{ \frac12 (a - \frac12 + \frac{\lambda}{\tau} + \e) }{(\omega - t)}  \right) \frac{d}{d \omega}  \\
    & \quad \qquad \qquad + \frac{ \frac12 ( a - \frac{\lambda}{\tau} -\frac12 + \e)(a-\frac12) \omega + \bar{q}_a}{\omega (\omega -1)(\omega - t)},
  \end{align*}
  with regular singularities at $\omega \in \{ 0,1,t,\infty\}$ and accessory parameter given by
  \begin{align}
    \bar{q}_a = \frac14 \left[ - \left(a-\frac12 -  \frac{\lambda}{\tau}\right)^2 + \e^2 + \frac{\Delta^2}{\tau^2} \right] (t - 1) - \frac12 (a-\frac12)(a - \frac12 - \frac{\lambda}{\tau} + \e).
  \end{align}

The local properties of the Heun ODE are given by the Riemann scheme (see, e.g. \cite{R1995,SL2000})
\[
  \begin{bmatrix}
    0 & 1 & t & \infty &; \omega \\
    0 & 0 & 0 & a-\frac{1}{2}&; \bar{q}_a \\
    \frac12\left(\frac{\lambda}{\tau} + \frac12 - a - \e\right) &  \frac12 \left(\frac{\lambda}{\tau} + \frac52 - a + \e \right) & \frac12 \left( \frac52 -a -\frac{\lambda}{\tau} - \e \right) & \frac12 \left(a -\frac{\lambda}{\tau} - \frac12 + \e \right)
  \end{bmatrix}.
\]

It is important to note here that $g$ appears in the Heun ODE only to the second power. Therefore, using the operator $\mathcal{S}_B^{(\e)}$ only changes the  Heun ODE picture by $\e \to -\e$. In the next section, this observation is fundamental for the study of the constraint polynomials corresponding to the Juddian solutions.

Finally, we obtain the Heun picture for the 2pAQRM following the case of the aNCHO.

\begin{prop}
  \label{Prop:32}
  The even and odd solutions of the eigenvalue problem of the 2pAQRM for the eigenvalue $\lambda$
  are equivalent respectively to the existence of solutions of Heun equations
  \begin{align*}
     \{ \varphi \in \mathcal{F}(\C) \otimes \C^2 \, | \, H_B^{(\e)} \varphi = \lambda \varphi, \, \varphi(-x) = \varphi(x) \} & \simeq \{ f \in \mathcal{O}(\Omega) \, | \,  \Lambda_1 f = 0 \},\\
     \{ \varphi \in \mathcal{F}(\C) \otimes \C^2 \, | \, H_B^{(\e)} \varphi = \lambda \varphi, \, \varphi(-x) = -\varphi(x) \} & \simeq \{ f \in \mathcal{O}(\Omega) \, | \,  \Lambda_2 f = 0 \},
  \end{align*}
  where $\mathcal{O}(\Omega)$ is the set of holomorphic solutions in a (simply connected) domain $\Omega \subset \C$ with $1,0 \in \Omega$ but
  $\frac{1}{4g^2} \notin \Omega$.
\end{prop}

Interestingly, the Heun picture for the 2pAQRM is obtained in a rather direct way compared to the aNCHO. Indeed, it does not require multiple intertwinning between different representations as in the latter.

The Heun picture of the 2pAQRM allows us to prove certain results in a way analogous to the aNCHO since the Heun picture is analogous. For instance, we can show that the eigenvalues of the 2pAQRM are bounded above by $2$ (cf. Thm 6.1 in \cite{W2015}).

\begin{thm}
  The multiplicity of any eigenvalue of the 2pAQRM is bounded above by $2$.
\end{thm}

We also use the Heun picture in Section \ref{sec:linindep} to give a detailed description of the multiplicity of the degenerate solutions and to fully describe the degeneracy picture of the 2pAQRM

\begin{rem}
  We can also obtain the Heun picture of the 2pAQRM from that of the aNCHO (see \cite{RBW2023} and Section \ref{sec:NCHOp}) using the relation between the two models discovered by Nakahama \cite{N2024}, at least for eigenvalues $\lambda$ satisfying $|\lambda| > \Delta$.
\end{rem}

\section{Spectral degeneracies and Juddian solutions for 2pAQRM}
\label{sec:juddian}

For the 2pAQRM, it is known \cite{XC2021} that degeneracies corresponding to Juddian solutions occur only for $\e \in \Z$ and are always of shared-parity type.  This is in contrast with the case of the QRM where exceptional degenerate (Juddian) solutions are of opposite-parity type, and the same is true for the AQRM with half-integer bias parity  (see \cite{RBW2022} for the full discussion including the discussion of how parity is defined for that case). 

In this section, for an integral bias $\e \in \Z$ we consider the shared-parity degenerate eigenvalues of the 2pAQRM, given by Juddian solutions, and give the associated constraint conditions . As mentioned above, Juddian solutions are exactly the quasi-exact (i.e., polynomial) solutions of either of the systems \eqref{eq:Hamiltonian1} and \eqref{eq:Hamiltonian2}. 

Based on the considerations of Section \ref{sec:deg}, to deal with both cases at once we consider the recurrence equation defined by 
\begin{align}
  \label{eq:coeffK}
  & (2(n+1)+\rho)(2(n+1)+\rho-1) K_{n+1}^{(\rho,\e)} = \nonumber \\
  & \qquad \frac{1}{2g} \left[(1+4g^2)(2n+\rho+\tfrac12) - \tau^2\left(\frac{\lambda}{\tau} + \e \right) - \frac{\Delta^2}{(2n+\rho+\frac12 + \e)  - \frac{\lambda}{\tau}} \right] K_{n}^{(\rho,\e)} - K_{n-1}^{(\rho,\e)}.
\end{align}
with initial conditions $K_{-1}^{(\rho,\e)}=0$ and $K_{0}^{(\rho,\e)}=1$. 

Let us now consider a polynomial solution of degree $2N+\rho$ for the exceptional eigenvalue $\lambda = \tau (2 N + \rho + \frac12 + \e)$.
Depending on the sign of $\e$, the solution corresponds to \eqref{eq:Hamiltonian1}  or \eqref{eq:Hamiltonian2}. 
Concretely, for $\e \geq 0$, from \eqref{eq:Hamiltonian1} we obtain
\[
  \phi_1(z) = \sum_{n=0}^N a^{(\rho,\e)}_n z^{2n+\rho}, \qquad \phi_2(z) = \sum_{n=0}^N b^{(\rho,\e)}_n z^{2n+\rho}, 
\]
while for $\e\leq0$, from \eqref{eq:Hamiltonian2} we have 
\[
  \phi_1(z) = \sum_{n=0}^N (-1)^n b^{(\rho,\e)}_n z^{2n+\rho}, \qquad \phi_2(z) = \sum_{n=0}^N (-1)^n a^{(\rho,\e)}_n z^{2n+\rho}.
\]
In both cases, the coefficients $\{b^{(\rho,\e)}_n\}$ satisfy the recurrence relation \eqref{eq:coeffK}. In particular, for the eigenvalue $\lambda = \tau (2 N + \rho + \frac12 + \e)$ the recurrence relation is
\begin{align}
  \label{eq:coeffKjudd}
  & (2(n+1)+\rho)(2(n+1)+\rho-1) K_{n+1}^{(\rho,\e)} = \nonumber \\
  & \qquad \frac{1}{2g} \left[(1+4g^2)(2n+\rho+\tfrac12) - \tau^2\left(2 N + \rho + \frac12 + 2\e\right) - \frac{\Delta^2}{2(n-N)} \right] K_{n}^{(\rho,\e)} - K_{n-1}^{(\rho,\e)}.
\end{align}

In addition, equations \eqref{eq:relA1} or \eqref{eq:relB1} impose the condition
\begin{align}
  \label{eq:relA2}
  (\tau(2N+ \rho +\frac12 + \e)-\lambda) a^{(\rho,\e)}_N = \Delta b^{(\rho,\e)}_N,
\end{align}
and since $(\tau(2N+ \rho +\frac12 + \e)-\lambda)=0$, we have $b^{(\rho,\e)}_N = 0$ and $a^{(\rho,\e)}_N \neq  0$. It follows that for $N$, from the recurrence relation we have
\[
  \tau \Delta a^{(\rho,\e)}_N \mp 2 g b_{N-2}^{(\rho,\e)} = 0,
\]
and the value of $a^{(\rho,\e)}_N$ is completely determined. The condition imposed by the existence of the polynomial solution is 
\begin{equation}
  \label{eq:constCondK}
  K_{N}^{(\rho,\e)} = 0,
\end{equation}
which is given as the determinant of  a $N \times N$ tridiagonal matrix. Concretely, we have
\[
   K_{N}^{(\rho,\e)}  = \det \Tridiag{\frac{1}{2g(2(n+1)+\rho)(2(n+1)+\rho-1)} \alpha_i^{(\rho,\e)}}{\frac1{(2(n+1)+\rho)}}{\frac{1}{(2(n+1)+\rho-1)}}{1\le i\le k},
 \]
 with
 \[
  \alpha_i^{(\rho,\e)} =  \left[(1+4g^2)(2n+\rho+\tfrac12) - \tau^2\left(2 N + \rho + \frac12 + 2\e\right) - \frac{\Delta^2}{2(n-N)} \right].
 \]
Here, we used the notation \begin{equation*}
  \Tridiag{a_i}{b_i}{c_i}{1\le i\le n}
  :=\begin{bmatrix}
    a_1 & b_1 & 0 &  \cdots & 0    \\
    c_1 & a_2 & b_2 &  \cdots  & 0\\
    \vdots & \ddots   & \ddots &  \ddots & \vdots   \\
    0 &  \cdots &  0  & a_{n-1} & b_{n-1} \\
    0 & \cdots  & 0  & c_{n-1} & a_n
  \end{bmatrix}
\end{equation*}
for a tridiagonal matrix. 

It is straightforward to verify from the derivation of the recurrence relations \eqref{eq:recurr1} and \eqref{eq:recurr2} that \eqref{eq:constCondK} is actually a constraint condition for the Juddian solution $\lambda = \tau (2 N + \rho + \frac12 + \e)$, that is, that it is a necessary and sufficient condition for the existence of Juddian solution and that the said is necessarily of degree $2N+\rho$.

\subsection{Constraint polynomials for 2pAQRM}
\label{sec:constraint}

In general, the coefficient $K_{N}^{(\rho,\e)}$ is a rational function on the parameters $g,\Delta,\e$. Following the idea developed in \cite{LB2015JPA} and \cite{KRW2017} for the AQRM, we give an equivalent description for the existence for Juddian solutions by a polynomial equation. In addition to being simpler, 
the polynomial description is fundamental for describing the underlying structures behind the 2pAQRM spectrum, as we see in Section \ref{sec:qtpol} and  Section \ref{sec:symdeg}.

\begin{dfn}
  \label{dfn:copo}
  The family of polynomials $P_{n}^{(N,\rho,\e)}(x,y)$ is given by
\begin{align}
    \label{eq:defConsP}
  P_{0}^{(N,\rho,\e)}(x,y) &= 1 \nonumber \\
  P_{1}^{(N,\rho,\e)}(x,y) &= y + 2 x ( 4 N + 2\rho + 2\e - 1) - 4(1+\e) \nonumber \\
  P_{k}^{(N,\rho,\e)}(x,y) &=  (y + 2 x k (4 N + 2\rho- 2k + 2\e + 1) - 4 k (k+\e))P_{k-1}^{(N,\rho,\e)}(x,y) \nonumber \\
                        & - 4 k (k-1) (2(N-k+1)+\rho)(2(N-k+1)+\rho-1) x P_{k-2}^{(N,\rho,\e)}(x,y).
\end{align}
for $k \geq 2$. The polynomial $P_{N}^{(N,\rho,\pm\e)}(x,y)$ is called constraint polynomial.
\end{dfn}

\begin{rem}
  \label{rem:biascoeff}
  If we have the $\e$ as the coefficient of the bias in  the definition of $H_{\e}$, that is, if we use 
  \begin{align*}
      H_{\e}= (a^\dag a + \frac12) + g ( a^2 + (a^\dagger)^2 ) \sigma_z + \Delta \sigma_x + \e \sigma_z,  
  \end{align*}
  as the definition, then it not difficult to see that we cannot obtain a polynomial equation on $(2g)^2$ (resp. $\tau^2$) unless we assume $\e$ it is of the form $\e=\tau \gamma$ for some $\gamma \in \R$.
\end{rem}

The basic result for constraint polynomials is that its zeros define a constraint relation for Juddian solutions 
equivalent to the natural one associated with the coefficients $K_{N}^{(\rho,\e)}$.

\begin{prop}
  \label{prop:constraint}
  Let $g,\Delta>0$ with $\tau = \sqrt{1-(2g)^2} \in \R$ (that is, $(2g)^2 \le 1$), then for $N>0$
  \begin{align*}
    P_{N}^{(N,\rho,\e)}((2g)^2,\Delta^2)=0
  \end{align*}
  if and only if $\lambda = \tau(2 N + \rho + \frac12 + \e)$ is a Juddian eigenvalue of the 2pAQRM Hamiltonian \eqref{eq:tpHamilt}.
\end{prop}

\begin{proof}
  It is enough to verify that $P_{N}^{(N,\rho,\e)}((2g)^2,\Delta^2)$ and $K_{N}^{(\rho,\e)}$ have the same positive roots. In fact,
  we verify as in Prop. 2.1 of \cite{KRW2017}\footnote{In the proof of Prop. 2.1 of  \cite{KRW2017} the matrix in the last equation of page 14 must be $
    \begin{bmatrix}
      i(2g)^2 + \Delta^2-i^2-2i \e & 2ig \\
      2(i+1)(N-i)g & 
    \end{bmatrix}
    $. Aside from this correction, the proof remains unchanged.} that
  \[
    (2g)^N 2^N (2N + \rho)! N! K_{N}^{(\rho,\e)}(g,\Delta) = P_{N}^{(N,\rho,\e)}((2g)^2,\Delta^2),
  \]
  and the result follows.
\end{proof}

An important consequence of Proposition \ref{prop:constraint} is that it
clearly shows that the existence of Juddian solutions does not depend on the sign of the system parameters $g$ and $\Delta$ (see also the comments before Proposition \ref{Prop:32}). Moreover, it also shows that only the zeros of the constraint polynomials in the region $x,y>0$ are relevant to the 2pAQRM spectrum. 

\begin{ex}
  We have
  \begin{align*}
    P_{1}^{(N,\rho,\e)}(x,y&) = y + 2 x ( 4 N + 2\rho + 2\e - 1) - 4(1+\e), \\
    P_{2}^{(N,\rho,\e)}(x,y&) = y^2  + 8  x^2 (4 N+2 \rho +2 \e -3) (4 N+2 \rho +2 \e -1) - 2 x y (12 N+6 \rho +6 \e -7) \\
       & \qquad -8 x \left(4 N^2+4 N \rho +16 N \e +14 N+\rho ^2+7 \rho +8 \e ^2+8 \rho  \e +4 \e -4\right)  - 4 y(3 \e +5) \\
       & \qquad + 32 (\e +1) (\e +2) 
  \end{align*}
\end{ex}

Example computations and numerical experiments in \cite{XC2021}  suggest that the shared-parity degeneracies for the 2pAQRM correspond to the Juddian eigenvalues, mirroring the situation of the AQRM. In Figure \ref{fig:tpAQRMeigencurves3} (see also Figure \ref{fig:tpQRMeigencurves}(b) above) we illustrate the situation for $\e=1,2$ showing the apparent shared-parity degenerate eigenvalues.

\begin{figure}[ht]
  \centering
  \subfloat[$\e=1$]{
    \includegraphics[height=4.5cm]{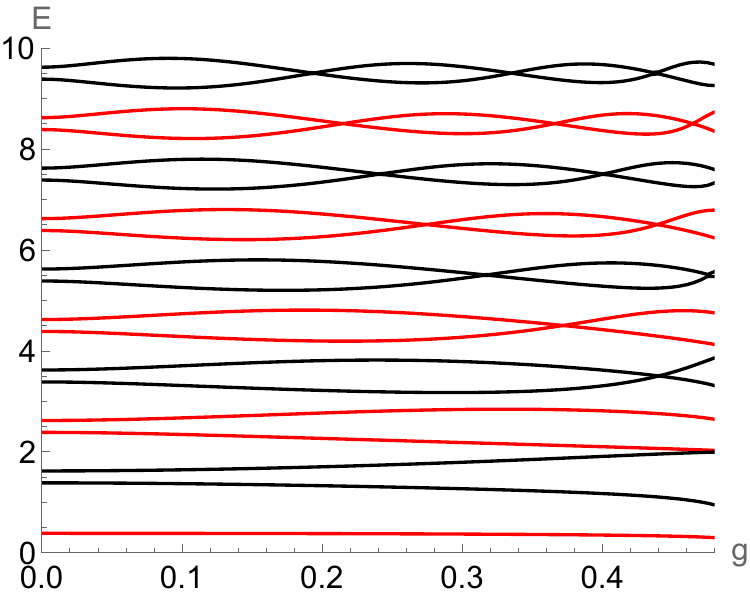}}
  ~ \qquad
  \subfloat[$\e=2$]{
    \includegraphics[height=4.5cm]{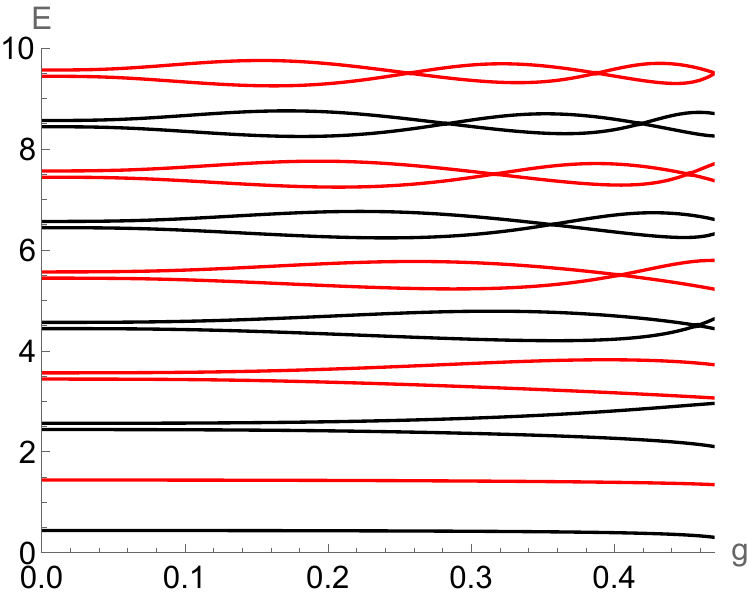}}
  \caption{Normalized spectral curves of the 2pAQRM for $\Delta=0.5$. Black (resp. red) curves correspond to even
    (resp. odd) eigenvalue curves.}
  \label{fig:tpAQRMeigencurves3}
\end{figure}

For a positive integer $\ell \geq 1$, shared-parity degeneracies may appear between two Juddian eigenvalues of the form
\[
  \lambda_1 = \tau (2 N + \rho + \e + \frac12  ) \quad \text{ and } \quad  \lambda_2 = \tau (2 (N+M) + \rho - \e  + \frac12  )
\]
for $\e = \ell$, provided that the corresponding eigenfunctions are linearly independent.

For the Juddian case, in \cite{XC2021} the question of the simultaneous vanishing of  $K_{N+\ell}^{(\rho,-\ell)}$ and $K_{N}^{(\rho,\ell)}$ for $N,\ell\geq 0$ is posed. In terms of constraint polynomials this corresponds to a divisibility relation of the form
\begin{align}
  \label{eq:div1}
  P_{N+\ell}^{(N+\ell,\rho,-\ell)}(x,y)= A_{N}^{(\ell,\rho)}(x,y)P_{N}^{(N,\rho,\ell)}(x,y).
\end{align}
for polynomials $A_{N}^{(\ell,\rho)}(x,y)$ with $A_{N}^{(\ell,\rho)}(x,y) >0$ for $x,y>0$, at least for the region $x \leq 1$, corresponding to $g \leq \frac12$.

For small values of $N$ and $\ell$, the divisibility condition in \eqref{eq:div1} can be verified directly, for instance, we have
\begin{align*}
  P_{1+1}^{(1+1,0,-1)}(x,y)/ P_{1}^{(1,0,1)}(x,y) &= 12 x+y \\
  P_{2+2}^{(2+2,0,-2)}(x,y)/ P_{2}^{(2,0,2)}(x,y) &= 1680 x^2+82 x y+y^2+4 y. 
\end{align*}
In addition, note that in the examples the polynomial $A_{N}^{(\ell,\rho)}(x,y)$ has positive coefficients, thus the positivity holds automatically.

The remainder of this section is devoted to the proof of the divisibility \eqref{eq:div1} between constraint polynomials, which is the analog of the study done in \cite{KRW2017} for the AQRM and answers positively the question posed in \cite{XC2021} about the vanishing of the higher degree coefficients of degenerate Juddian solutions for the 2pAQRM. In addition to the proof of \eqref{eq:div1}, another objective of the present paper is to show that the polynomial $A_{N}^{(\ell,\rho)}(x,y)$, which we call {\bf inter-constraint polynomial}, occupies a significant role in the spectrum of the 2pAQRM.

We leave the issue of linear independence of solutions to Section \ref{sec:linindep} where it is settled using the properties of the Heun picture.

\begin{rem} 
  The divisibility relation \eqref{eq:div1} does not hold for general $\rho \neq 0,1$ even for small $N$. For instance, we have
  \[
    P_{1+1}^{(1+1,\rho,-1)}(x,y) = P_{1}^{(1,\rho,1)}(x,y) a(x,y,\rho) + b(x,y,\rho)
  \]
  with
  \[
    a(x,y,\rho) = 4 (2 \rho +3) x+y
  \]
  and
  \[
    b(x,y,\rho) = -8 (\rho -1) \rho  x
  \]
  Note that $a(x,y,\rho)$ and $b(x,y,\rho)$ are rational functions in $\rho$ and $b(x,y,\rho)=0$ for $\rho=0,1$.
\end{rem}

\begin{rem}
    \label{rem:figuEllip}
  To get another perspective for the divisibility of constraint polynomials, it is convenient to visualize the curves  described in the $(g,\Delta)$-plane. Concretely, the curves determined by
\begin{align}
  \label{eq:constcurves}
  P_{N}^{(N,\rho,\e)}((2g)^2,\Delta^2)=0,
\end{align}
for different values of $N \geq 0, \rho \in \{0,1\}$ and $\e \in \R$, are shown in Figure \ref{fig:curveCPoly1}.

The curves determined by \eqref{eq:constcurves} appear as $N$ concentric ellipses, a feature shared with those of the AQRM \cite{RBW2022}. In addition, we observe that the curves are contained in the strip $|g|\leq 1/2$, reflecting the condition for the Hamiltonian $H_{\e}$ to have a point spectrum structure. 

\begin{figure}[ht]
  \centering
  \subfloat[$\rho=0, \e=0$ and $N=1,2,3$ (blue, green,red)]{
    \includegraphics[height=4.5cm]{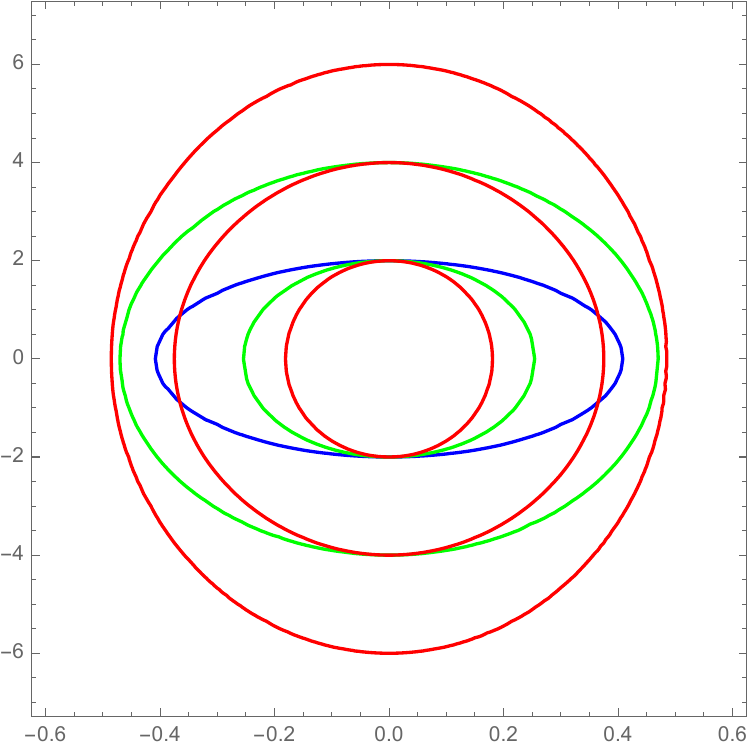}}
  ~ \qquad
  \subfloat[$\rho=1, \e=2$ and $N=5,7,9$ (black,red,blue)]{
    \includegraphics[height=4.5cm]{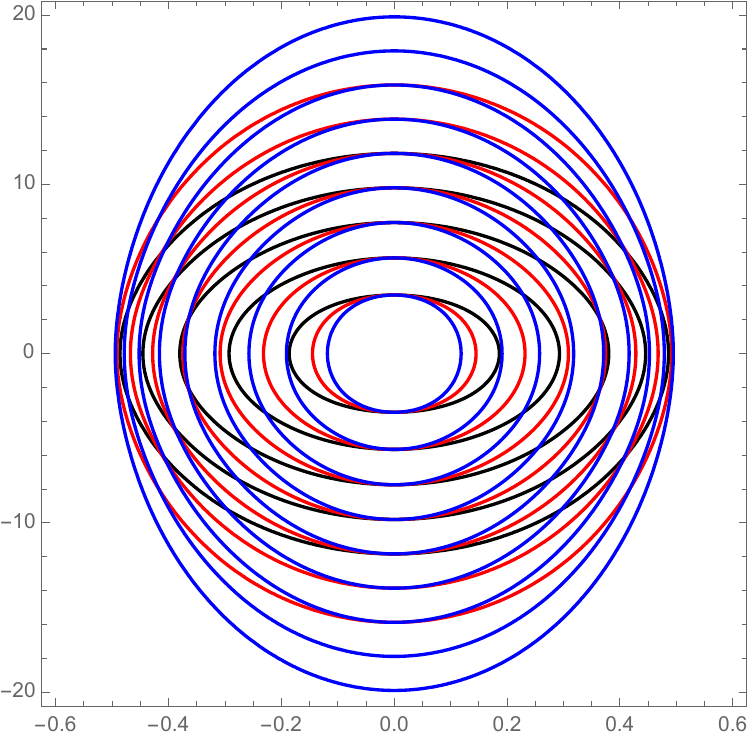}}
  \caption{Curves $P_{N}^{(N,\rho,\e)}((2g)^2,\Delta^2)=0$}
  \label{fig:curveCPoly1}
\end{figure}

Then, by fixing the parity $\rho$ and the parameter $N$ and varying the bias parameter $\e$ we can visualize the simultaneous vanishing of the constraint polynomials. For instance, let $\rho=0$,$N=2$ and the curves of $P_{N+\e}^{(N+\e,\rho,-1)}((2g)^2,\Delta^2)=0$ and $P_{N+2}^{(N+2,\rho,-\e)}((2g)^2,\Delta^2)=0$ for $\e=1.5,1.9,2$ in Figure \ref{fig:curveDeg}. As $\e \to 2$, the two curves completely coincide,  showing that the two constraint polynomials have the same real zeros, or equivalence that divisibility holds and 
$A_{N}^{(\ell,\rho)}(x,y) \geq 0$ for $x,y\geq 0$.

\begin{figure}[ht]
  \centering
  \subfloat[$\e=1.5$]{
    \includegraphics[height=4cm]{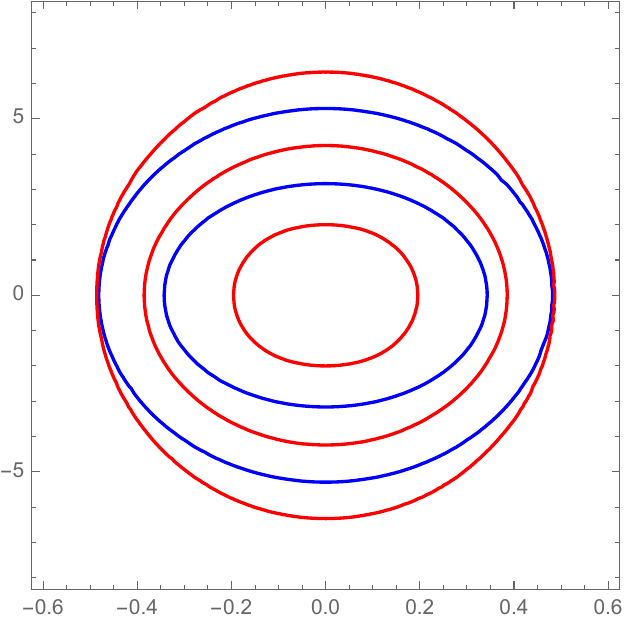}}
  ~ \quad
  \subfloat[$\e=1.9$]{
    \includegraphics[height=4cm]{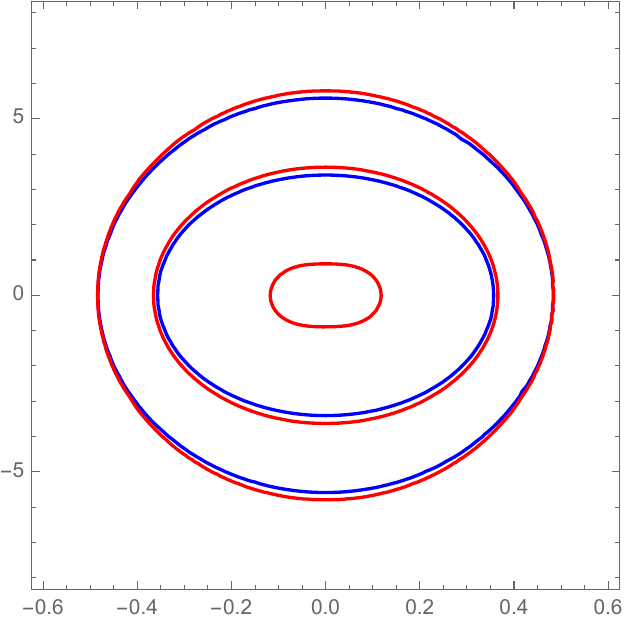}}
  ~ \quad
  \subfloat[$\e=2$]{
    \includegraphics[height=4cm]{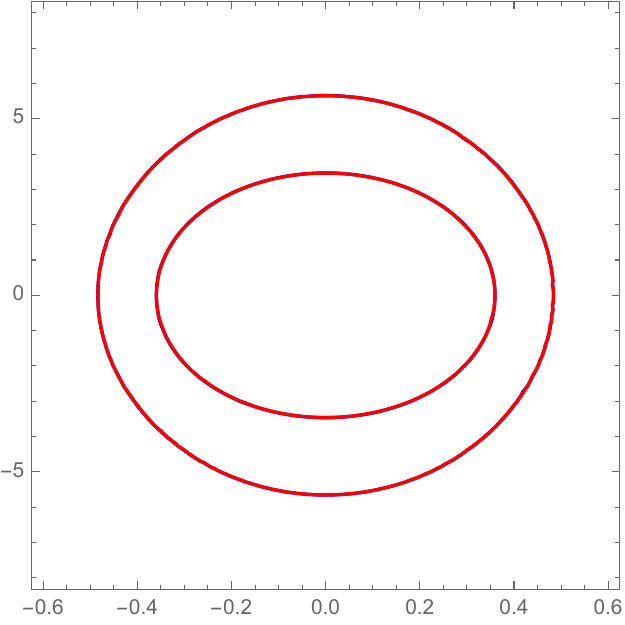}}
  \caption{Visualization of divisibility by the curves of $P_{N+\e}^{(N+\e,\rho,-1)}((2g)^2,\Delta^2)=0$ and $P_{N+2}^{(N+2,\rho,-\e)}((2g)^2,\Delta^2)=0$}
  \label{fig:curveDeg}
\end{figure}  
\end{rem}

\subsection{Questions on the number and distribution of Juddian solutions}
\label{sec:openconstp}

In the remainder of this section, we discuss certain aspects of constraint polynomials that are 
not needed for this paper but that are nevertheless of mathematical or physical interest.

For a given $N$ and general $\e$, not necessarily integer, the constraint equation \eqref{eq:thmconst} in Theorem \ref{thm:judddeg} (Section \ref{sec:degstruct2}) is satisfied when the parameters correspond to a 2pAQRM with spectrum containing a Juddian solution. It is natural to consider a type of inverse problem: for given parameters $g,\Delta,\e$, how many Juddian solutions are in the spectrum of the corresponding 2pAQRM?

\begin{prob}
  \label{prob:count}
  For fixed $g,\Delta,\e$ and $\rho \in \{0,1\}$. Compute the number of non-negative integers $N$ such that the equation
  \[
    P_{N}^{(N,\rho,\e)}(x,y) = 0
  \]
  is satisfied.
\end{prob}

The recent paper \cite{R2024} focus on this problem for the QRM. For general parameters, it is generally believed that the number of Juddian solutions is finite. In fact, since  the constraint conditions 
are polynomial equations with integer coefficients, it is easy to construct parameter families $g,\Delta,\e$ such that the spectrum do not constraint any Juddian solutions. 

\begin{conject}
  For fixed $g,\Delta,\e$, the number of Juddian solutions of the 2pAQRM is finite. Moreover, it is uniformly bounded with respect to parameters $g,\Delta,\e$.
\end{conject}

It is also natural to also consider whether the family of curves \eqref{eq:constcurves} for $N\geq 1$ are dense in the strip $|g|\leq 1/2$. As an illustration, in Figure \ref{fig:curvesDensity} we illustrate the curves for $N=1,2,\cdots,8$ for $\rho=1$, $\e=1$, note that as more curves are added, the region strip gets covered by the curves.

\begin{figure}[ht]
  \centering
  \includegraphics[height=6cm]{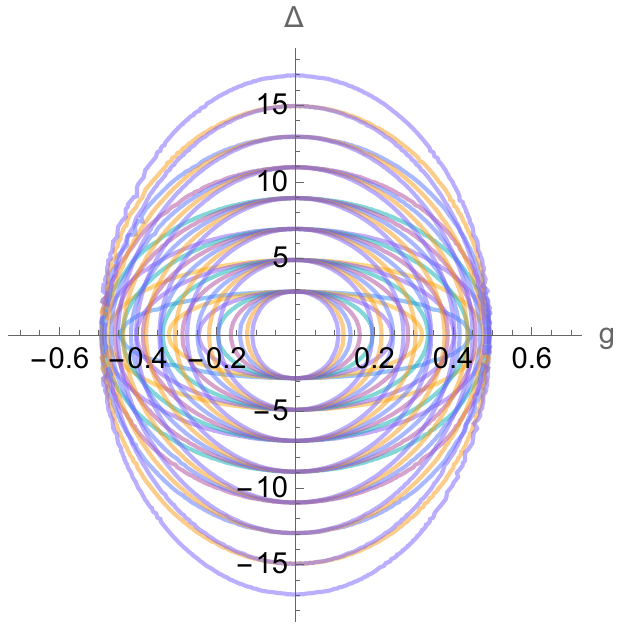}
  \caption{Constraint curves for $\rho=1$,$\e=1$ and $N=1,2,\cdots,8$}
  \label{fig:curvesDensity}
\end{figure}

\begin{conject}
  For fixed $\rho \in \{0,1\}$, the family of $(g,\Delta)$-plane curves 
  \[
    P_{N}^{(N,\rho,\e)}((2g)^2,\Delta^2)=0,
  \]
  is dense in the strip $|g|<1/2$.
\end{conject}

The analogous problem for the AQRM is discussed in length in \cite{RBW2022}. We remark that this density conjecture for the (symmetric) QRM was recently affirmatively proved in \cite{KR2024}, where also advances to the analog of Problem \ref{prob:count} were presented. The proofs make heavy use of the fact that for the QRM, in the case $y=0$, the constraint polynomials are multiples of generalized Laguerre polynomials. 

Another interesting feature of the constraint polynomials is that they appear to have information not only of the Juddian solutions, but also can describe partially features of the full spectrum, as we see in Sections \ref{sec:qtpol} and \ref{sec:symdeg}. Here, we just give a simple basic example related to the number of Juddian solutions.

Recall that to guarantee the spectrum of the 2pAQRM Hamiltonian contains only eigenvalues, we make the assumption $|g| < 1/2$ and that in the critical value $g= 1/2$ (corresponding to $x=1$), the model contains continuous spectrum. In principle, constraint polynomials only concern the case of pure point spectrum, that is, $|g| < 1/2$, but it is interesting to consider their behavior at the critical point.

The following result is obtained immediately from the three-term recurrence relation \eqref{eq:defConsP} of the
constraint polynomials.

\begin{prop}
  For $N\geq 0$, for $x=1$, the constraint polynomial $P_{N}^{(N,\rho,\pm\e)}(1,y)$ does not depend on $\e$. \qed
\end{prop}

\begin{ex}
  We have
  \begin{align*}
    P_{1}^{(1,0,\e)}(1,y) &= y + 2, \qquad P_{1}^{(1,1,\e)}(1,y) = y+6 \\
    P_{2}^{(2,0,\e)}(1,y) &= (y+2) (y+12),\\
    P_{2}^{(2,1,\e)}(1,y) &=  (y+6) (y+20), \\
    P_{5}^{(5,1,\e)}(1,y) &= (y+2) (y+12) (y+30) (y+56) (y+90).
  \end{align*}
\end{ex}

The constraint appears to have a simple form at the critical point. Moreover, past the critical point the polynomials do not have zeros (see also Remark \ref{rem:figuEllip}). We leave a complete description of this as an open problem.

\begin{conject}
  We have
  \[
    P_{N}^{(N,\rho,\e)}(1,y) = \prod_{n=1}^{N} (y + 2n(2n + 2\rho - 1 )),
  \]
  Moreover, for $x>1$, the polynomial $P_{N}^{(N,\rho,\e)}(x,y)$ has positive coefficients, and therefore, no positive roots for $y$.
\end{conject}

The behavior of the zeros of the constraint polynomials past the critical point seems consistent with recent results \cite{SZ2026} on the spectrum of the 2pQRM in that parameter regime. In particular, we see that the polynomial $P_{N}^{(N,\rho,\e)}(x,y)$ is having no zeros for $x>1,y>0$ coincides with the fact that for 
$g > 1/2$, there is no embedded point spectrum in the continuous spectrum of the 2pQRM. 

\subsection{Orthogonal polynomial families associated  to constraint polynomials}
\label{sec:orthopoly}

For the AQRM, when $\Delta = 0$ (i.e. $y=0$), sometimes called ``degenerate atomic limit'', the constraint polynomials actually determine a family of orthogonal polynomials equal to generalized Laguerre polynomials, up to constant multiples (see \cite{KRW2017}). 

The Laguerre polynomials are considered to appear from the fact that for $\Delta = 0$, the Hamiltonian of the AQRM is that of a displaced quantum harmonic oscillator, in which the Laguerre polynomials appear in a natural way \cite{Sc1967AP}. 

Let us consider the analogous situation for the two-photon case. For $\rho=0,1$, by factoring out $i ( 2 N +\e - i  +\rho +\tfrac12)$ from each row in the determinant,  we rewrite \eqref{eq:threetermwR} as
\begin{align}
    P_{N}^{(N,\rho,\e)}(x,y) =  N! (N+\e + \rho+ \tfrac12)_N \det \left( \hat{\bm{D}} y + \bm{I}_N (4 x-2) + \bm{V}_N \right) ,
\end{align}
with
\[
     \hat{\bm{D}}_N  = \Diag\left\{ \frac{1}{i ( 2 N +\e - i  +\rho +\tfrac12)} \right\}_{1\leq i \leq N},
\]
and 
\begin{align*}
    \bm{V}_N = \Tridiag{ \frac{ (1-4\e^2)(2(N-i)+\e+\rho+\frac12)^2}{ 2 (2(N-i)+\e + \rho-\frac12)_2(2(N-i)+\e+\rho+\frac12)_2 } }{ \frac{(2(N+\e-i) + \rho -1)_2}{( 2(N-i) +\e + \rho - \frac{1}2)_2} }{ \frac{(2(N-i)+\rho-1)_2 }{( 2(N-i) +\e + \rho - \frac{3}2)_2} }{1\le i\le N}.
\end{align*}

Setting $z= 4x-2$, defining
\[
  R_N^{(\e,\rho)}(z) = \frac{P_{N}^{(N,\rho,\e)}(x,0)}{N! (N+\e + \rho+ \tfrac12)_N},
\]
and expanding the continuant, we see that $R_N^{(\e,\rho)}$ satisfy the recurrence relations
\begin{align*}
  R_{-1}^{(\e,\rho)}(z) &= 0, \qquad R_{0}^{(\e,\rho)}(z) = 1 \\
  R_{k+1}^{(\e,\rho)}(z) &=  (y - \alpha_{k}) R_{k}^{(\e,\rho)}(z) - \beta_{k} R_{k-2}^{(\e,\rho)}(z),.
\end{align*}
for $k\geq 1$. Here, 
\[
  \alpha_{k} := \alpha_k(\rho,\e) = \frac{ 2(\e^2-\frac{1}{4})(2 k +\e+\rho+\frac12)^2}{(2 k +\e + \rho-\frac12)_2(2 k +\e+\rho+\frac12)_2 }
\]
and
\[
  \beta_{k} := \beta_{k}(\rho,\e) = \frac{(2k +2\e + \rho -1)_2 (2 k +\rho-1)_2  }{( 2 k  +\e + \rho - \frac{1}2)_2 ( 2 k +\e + \rho - \frac{3}2)_2}.
\]
By Favard's theorem (see e.g. \cite{C1978}), for $\e>-\frac12$, the polynomials $R_{k}^{(\e,\rho)}(z)$ determine a family of orthogonal polynomials. 

\begin{ex}
  For the case $\e=\frac12$, we have $R_{k}^{(\e,\rho)}(z) \in \Z[z]$ for all $k\geq 0$. Indeed, note that
  \[
    \beta_{k}(\rho,\tfrac12) = \frac{(2k+\rho)_2 (2k+\rho-1)_2}{(2k+\rho)_2 (2k+\rho-1)_2} = 1
  \]
  and clearly $\alpha_k=0$ for $k>0$ and
  \begin{align*}
     \alpha_{0}(\rho,\e) = \frac{2(\e-\tfrac12)(\e+\tfrac12)}{(\e + \rho - \tfrac12)(\e + \rho + \tfrac32)},
  \end{align*}
  so that $\alpha_{0}(\rho,\tfrac12) = \delta_{0,\rho}$.

  In particular, by comparing the three-term recurrence relations, we see that $R_k^{(\frac12,1)}(z) = U_k(\frac{z}2)$, where $\{U_n\}$ are the Chebyshev polynomials of the second kind.
\end{ex}

Another family of orthogonal polynomials can be obtained from the constraint polynomials.
For fixed $N \geq 1$, the three-term recurrence relation \eqref{eq:defConsP} gives
\[
  P_{N+1}^{(N,\rho,\e)}(x,y) = (y + 4 x (N+1) (N + \rho + \e - \frac12) - 4 (N+1) (N+1+\e))P_{N}^{(N,\rho,\e)}(x,y),
\]
in particular, $P_{N}^{(N,\rho,\e)}(x,y)$ divides $P_{N+k}^{(N,\rho,\e)}(x,y)$ for $k\geq 1$. Then, by defining
\[
  \mathcal{Q}_{k}^{(N,\rho,\e)}(x,y) = \frac{P_{N+k}^{(N,\rho,\e)}(x,y)}{P_{N}^{(N,\rho,\e)}(x,y)},
\]
for $k\geq 0$ and $\mathcal{Q}_{-1}^{(N,\rho,\e)}(x,y) =0$. Then, the polynomials $\mathcal{Q}_{k}^{(N,\rho,\e)}(x,y)$ satisfy the three-term recurrence relation
\begin{align*}
   \mathcal{Q}_{k}^{(N,\rho,\e)}(x,y) &=  (y + 2 x (N+k) (2 N + 2\rho + 2\e + 1-k) - 4 (N+k) (N+k+\e))\mathcal{Q}_{k-1}^{(N,\rho,\e)}(x,y)\\
                         & - 4 (N+k) (N+k-1) (2k-\rho-2)(2k -\rho -1) x  \mathcal{Q}_{k-2}^{(N,\rho,\e)}(x,y).
\end{align*}
Therefore, for $x<0$ the family of polynomials  $\{\mathcal{Q}_{k}^{(N,\rho,\e)}(x,y) \}_{k\geq 0}$ form a family of orthogonal polynomials by Favard's theorem. However, since the case $x<0$ corresponds to imaginary $g$, it is not straightforward how to relate it to the spectral properties of the 2pAQRM spectrum.

It would be interesting to investigate the possible relation of these families of orthogonal polynomials with the representation theoretical structures underlying the quantum interaction models, for instance, with a reasonable $SL_2(\R)$-structure of the 2pAQRM Hamiltonian. 

For example, in order to obtain the Heun picture of the eigenvalue problem of the AQRM (resp. 2pAQRM) Hamiltonian (see Section \ref{sec:HeunPicture}, \cite{W2015,W2017} and the forthcoming \cite{RBW2027}) we need to employ the Lie algebra representations $\pi_a'$ (and its equivalent representation $\varpi_a$) of  $\mathfrak{sl}_2(\R)$. 

Therefore,  the presence of Laguerre  polynomials is not surprising since these polynomials have a distinguished role, they span the representation space of the group representation corresponding to (an equivalent representation of) $\pi_a'$, in the theory of unitary representations of the simply-connected covering group of $SL(2,\R)$ (see \cite{K2000}), as a counterpart of the holomorphic discrete series (i.e., square integrable representations) of $SL(2,\R)$ (c.f. \cite{S1975}).

\section{Divisibility and positivity of constraint polynomials}
\label{sec:divpos}

In this section, we prove the algebraic properties of the constraint polynomials required to characterize the shared-parity degeneracies of the 2pAQRM. We follow the general scheme of proof of \cite{KRW2017,CRB2020} for the AQRM and, whenever possible we consider $\rho$ to be a real variable to simplify the proofs.

\subsection{Divisibility relations for constraint polynomials}
\label{sec:div}

The divisibility relation \eqref{eq:div1} between constraint polynomials follows from a special
determinant expression that is different from the usual one arising from the three-term recurrence 
relation in Definition \ref{dfn:copo}.  The method in this section is based on the one given in \cite{KRW2017,CRB2020} for the AQRM.

We start by considering the natural determinant expression for the polynomials $P_{k}^{(N,\rho,\e)}$, 
that is,
\begin{align}
  \label{eq:threeterm}
  P_{k}^{(N,\rho,\e)} = \det \left( \bm{I}_k y + 4 \bm{A}_k^{(N,\rho,\e)} x +    \bm{U}_k^{(N,\e)} \right)
\end{align}
with
\[
  \bm{A}_k^{(N,\rho,\e)} = \Tridiag{i(2 N - i + \rho  + \e +\tfrac{1}2 )}{0}{i (i+1) (2(N-i)+\rho-1) (2(N-i)+\rho)}{1\le i\le k},
\]
and
\[
  \bm{U}_k^{(N,\e)} = \Tridiag{- 4 k (k+\e)}{1}{0}{1\le i\le k}.
\]

The special determinant expression for the case $k=N$, corresponding to constraint polynomials, is obtained  by conjugation of the matrix in the determinant. Indeed, conjugation of the matrix inside the determinant by an appropriate matrix leaves the value of the determinant invariant. It turns out that the appropriate matrix for this purpose is the matrix that diagonalizes $\bm{A}_k^{(N,\rho,\e)}$. 

\begin{prop}
  \label{prop:matrixE}
  The eigenvalues of $\bm{A}_k^{(N,\rho,\e)}$ are
  \[
    \{ i ( 2N - i + \rho + \e  +\tfrac{1}2 )\}
  \]
  for $i=1,\cdots,k$, and the eigenvectors are given by the columns of the lower triangular
  matrix $\bm{E}_k^{(N,\rho,\e)}$ given by
  \[
    (\bm{E}_k^{(N,\rho,\e)})_{i,j} = \binom{i}{j} \frac{(i-1)!}{(j-1)!} \frac{(2 j-2 N -\rho)_{2(i-j)}}{(2j - 2 N - \rho - \e + \tfrac12)_{i-j}},
  \]
  where $(a)_n = a(a+1)\cdots(a+n-1)$  is the Pochhammer symbol.
\end{prop}

\begin{proof}
  Let us write  $f(i) = i ( 2N - i + \rho + \e  +\tfrac{1}2 )$, $g(i) = i (i+1) (2(N-i)+\rho-1) (2(N-i)+\rho)$ and
  $\bm{D}_k = \Diag\{f(i)\}_{1\leq i \leq k}$. We need to verify the identity
  \[
    \bm{A}_k^{(N,\rho,\e)} \bm{E}_k^{(N,\rho,\e)}=  \bm{E}_k^{(N,\rho,\e)} \bm{D}_k.
  \]
  Considering the $i,j$-th component of the matrix, this is reduced to
  \[
    g(i-1) (\bm{E}_k^{(N,\rho,\e)})_{i-1,j} = (f(j)-f(i))  (\bm{E}_k^{(N,\rho,\e)})_{i,j},
  \]
  and by the definition, this reduces to
  \[
    f(j)-f(i) = (i-j)(i+j-2 N - \rho - \e - \tfrac12),
  \]
  which is verified immediately.
\end{proof}

Then, we proceed to obtain the special determinant expression for the constraint polynomials.

\begin{prop}
  \label{prop:detexp2}
  For $N \geq 0$, we have
  \begin{align}
    \label{eq:threetermwR}
     P_{N}^{(N,\rho,\e)}(x,y) = \det \left( \bm{I}_N y + \bm{D}_N (4 x-2) + \bm{C}_N \right) + \rho (1-\rho) R_N^{(\rho,\e)}(x,y),
  \end{align}
  for a polynomial $R_N^{(N,\rho,\e)}(x,y) \in \Q[\rho,\e]$, 
  \[
    \bm{D}_N = \bm{D}_N(\e,\rho) = \Diag\left\{ i ( 2 N +\e - i  +\rho +\tfrac12) \right\}_{1\leq i \leq N},
  \]
  and  $\bm{C}_N=\bm{C}_N(\e,\rho)$ is given by
  \begin{align*}
    \bm{C}_N &=
     \Tridiag{ i(2N+\e-i+\rho+\tfrac12) \hat{a}_i^{(N,\rho,\e)} }{\frac{i (2N+\e - i + \rho +\frac{1}2) (2(N+\e-i) + \rho -1)_2}{( 2(N-i) +\e + \rho - \frac{1}2)_2} }{ \frac{(i+1) (2 N+\e - i + \rho -\tfrac{1}2) (2(N-i)+\rho-1)_2 }{( 2(N-i) +\e + \rho - \frac{3}2)_2} }{1\le i\le N},
  \end{align*}
  with
  \[
    \hat{a}_i^{(N,\rho,\e)}  = \frac{1-4\e^2}{2(2(N-i)+\e+\rho+\tfrac{3}2)(2(N-i)+\e + \rho-\tfrac{1}{2})}.
  \]
\end{prop}

\begin{proof}
  First, we show the equality
  \begin{equation}
    \label{eq:threetermMod}
    P_{N}^{(N,\rho,\e)}(x,y) = \det \left( \bm{I}_N y + 4 \bm{D}_N x + \bm{M}_N +  \rho (1-\rho) \bm{e}_N \,\,^{T}\bm{u}_N \right),
  \end{equation}
  where $\bm{e}_N  \in \R^N$ is the $N$-th basis vector, $\bm{u}_N = \bm{u}_N(\e,\rho)$ is the vector given entry-wise by
  \begin{align*}
    (\bm{u}_N)_i = (-1)^{N-i+1} \frac{N!}{(i-1)!} \binom{N+1}{i} \frac{(\rho+1)_{2(N-i)}}{(N-i+\e+\rho-\tfrac12)_{N-i+1}},
  \end{align*}
  and
  \begin{align*}
    \bm{M}_N &=  \Tridiag{a_i^{(N,\rho,\e)}}{1}{c_i^{(N,\rho,\e)}}{1\le i\le N}, 
  \end{align*}
  with
  \begin{align*}
    a_i^{(N,\rho,\e)} =& - 2i(2N+\e-i+\rho+\tfrac12)\left( 1 - \frac{1-4\e^2}{4(2N+\e-2i+\rho+\tfrac{3}2)(2N+\e-2i+\rho-\tfrac{1}{2})} \right)  \\
    c_i^{(N,\rho,\e)} =& i(i+1)(2(N-i)+\rho)(2(N-i)+\rho-1)(2(N+\e-i) + \rho)(2(N+\e-i) + \rho -1) \\
     &\qquad \times  \frac{(2N+\e - i + \rho +\tfrac{1}2)(2 N+\e - i + \rho -\tfrac{1}2)}{( 2(N-i) +\e + \rho - \tfrac{3}2)( 2(N-i) +\e + \rho - \tfrac{1}2)^2( 2(N-i) +\e+  \rho + \tfrac{1}2)}.
  \end{align*}
  
   Using Proposition \ref{prop:matrixE}, it is enough to verify that
   \begin{align}
     \label{eq:proofDetexp2}
     \bm{U}_k^{(N,\e)} \bm{E}_N^{(N,\rho,\e)}= \bm{E}_N^{(N,\rho,\e)} \bm{M}_N + \bm{E}_N^{(N,\rho,\e)} \bm{e}_N \,\,^{T}\bm{u}_N,
   \end{align}
   and since $(\bm{E}_N^{(N,\rho,\e)})_{N,N}=1$ we note that $\bm{E}_N^{(N,\rho,\e)} \bm{e}_N \,\,^{T}\bm{u}_N = \bm{e}_N \,\,^{T}\bm{u}_N$.

   Define
   \[
     d_{i,j} = \binom{i}{j} \frac{(i-1)!}{(j-1)!} \frac{(2 j-2 N -\rho)_{2(i-j)}}{(2j - 2 N - \rho - \e + \tfrac12)_{i-j}},
   \]
   and note that $d_{i,j}=0$ if $j>i$. It is straightforward to verify the identities
   \begin{align*}
     d_{i+1,j} - d_{i,j-1} &= \frac{1}{(i-j+1)}\bigg( \frac{i (i+1)(2 i-2 N-\rho +1) (2 i-2 N-\rho )}{ \left(i+j-2 N-\rho -\e +\frac{1}{2}\right)} \\
                           &\qquad -\frac{j (j-1) (2 j-2 N-\rho -2) (2 j-2 N-\rho -1) \left(i+j-2 N-\rho -\e +\frac{1}{2}-1\right)}{\left(2 j-2 N-\rho -\e +\frac{1}{2}-1\right)\left(2 j-2 N-\rho -\e +\frac{1}{2}-2\right) } \bigg) d_{i,j} \\
     c_j^{(N,\rho,\e)} d_{i,j+1} &= (i-j) (2 (-j+n+\e )+\rho ) (2 (-j+n+\e )+\rho -1) \\
                           &\qquad \times \frac{((-j+2 n+\e )+\rho +\frac{1}{2}) ((-j+2 n+\e )+\rho -\frac{1}{2})}{\left(2 (n-j)+(\rho +\e )-\frac{1}{2}\right) \left(2 (n-j)+(\rho +\e )+\frac{1}{2}\right) \left(i+j-2 n-\rho -\e +\frac{1}{2}\right)}d_{i,j}.     
   \end{align*}

   Then we verify the identity
   \begin{align}
     \label{eq:proofDetexp3}
     (d_{i,j-1} - d_{i+1,j}) + (4i(i+\e) + a_j^{(N,\rho,\e)})d_{i,j} + c_j^{(N,\rho,\e)} d_{i,j+1} = 0,
   \end{align}
   after writing everything in terms of $d_{i,j}$ and a long but elementary computation. For $1\leq i,j \leq N$, we have
   $(\bm{E}_k^{(N,\rho,\e)})_{i,j}=d_{i,j}$ and thus \eqref{eq:proofDetexp3} gives \eqref{eq:proofDetexp2} for the $i,j$-th
   component for $1\leq i\leq N-1$ and $1\leq j\leq N$.

   For $i=N$, from \eqref{eq:proofDetexp3}  we get 
   \[
     \left(\bm{U}_k^{(N,\e)} \bm{E}_N^{(N,\rho,\e)} - \bm{E}_N^{(N,\rho,\e)} \bm{M}_N\right)_{N,j} = - d_{N+1,j},
   \]
   and using the identities
   \begin{align*}
     (2 j - 2N - \rho)_{2(N+1-j)} &= \rho (\rho-1) (\rho+1)_{2(N-j)}, \\
     (2j - 2 N - \rho - \e + \tfrac12) &= (-1)^{N+1-j} (N-j + \rho + \e -\tfrac12)_{N+1-j},
   \end{align*}
   which proves \eqref{eq:threetermMod} by comparing the values of $d_{N,N}$ and $(u_N)_N$. Finally, by elementary linear algebra, we have
  \[
    \det\left( \bm{A} + \bm{v} ^{T}\bm{u}_k \right) =  \det(\bm{A} ) +  \bm{v} {\rm adj}(\bm{A}) \,^{T}\bm{u}_k,
   \]
   where ${\rm adj}$ is the adjugate matrix of $\bm{A}$, that is, the transpose of the matrix of the cofactors of
   $\bm{A}$.
   In particular, we have
   \begin{align}
     \label{eq:polyR}
     R_N^{(\rho,\e)}(x,y) = \bm{e}_N  {\rm adj}( \bm{I}_N y + 4 \bm{D}_N x + \bm{C}_N ) ^{T}\bm{u}_k,
   \end{align}
   and then, since it is the determinant of a tridiagonal matrix, replacing $\bm{M}_N$ with $\bm{C}_N$ does not change
   the value of the determinant.
\end{proof}

\begin{ex}
  For $N=1,2$ the polynomials $R_N^{(\rho,\e)}(x,y)$ are given by
  \begin{align*}
    R_1^{(\rho,\e)}(x,y) &=  \frac{-4 \rho (1-\rho) }{2 \rho +2 \e -1}, \\
    R_2^{(\rho,\e)}(x,y) &= -12 \rho (1-\rho)   \frac{ 2 x (2 \rho +2 \e +1) (2 \rho +2 \e +7)+
                         y (2 \rho +2 \e +1)-2 (\rho +2 \e )^2-2 (7 \rho +6 \e +4)}{(2 \rho +2 \e -1) (2 \rho +2 \e +1)}.
  \end{align*}
  In principle, it is possible to obtain an explicit expression $R_N^{(\rho,\e)}(x,y)$ from \eqref{eq:polyR}, but since we do not use it in the rest of the discussion we omit the details.
\end{ex}

We are now in the position to prove the divisibility relation \eqref{eq:div1} between constraint polynomials.

\begin{thm} \label{thm:div}
  Let $N\geq0$ and $\ell \geq 1$ and $\rho=0,1$, then we have
\begin{align*}
  P_{N+\ell}^{(N+\ell,\rho,-\ell)}(x,y)= A_{N}^{(\ell,\rho)}(x,y)P_{N}^{(N,\rho,\ell)}(x,y),
\end{align*}
where the polynomial $A_{N}^{(\ell,\rho)}(x,y)$ is given by
\[
  A_{N}^{(\ell,\rho)}(x,y) = (N+1)_\ell (N+\rho+\tfrac12)_\ell \det \left(2 \bm{I}_\ell (2x-1) + \bm{D}_\ell^{(N,\rho)} y + \bm{V}_\ell^{(\rho)} \right),
\]
where $\bm{D}_\ell^{(N,\rho)}  = \Diag\left\{\frac{1}{(N+i) ( N + \ell - i  +\rho +\tfrac12)}\right\}_{1\leq i \leq \ell}$ and
\[
  \bm{V}_\ell^{(\rho)} =  \Tridiag{\frac{1-4 \ell^2}{2 \left(-2 i+\ell+\rho -\frac{1}{2}\right)\left(-2 i+\ell+\rho +\frac{3}{2}\right)}}{\frac{(-2 i + \rho )(-2 i + \rho-1)}{\left(\ell-2 i+\rho -\frac{1}{2}\right)\left(\ell-2 i+\rho +\frac{1}{2}\right)  }}{\frac{ (2 (\ell-i)+\rho ) (2 (\ell-i)+\rho-1 )}{\left(\ell -2 i+ \rho -\frac{1}{2}\right) \left(\ell -2 i+\rho -\frac{3}{2}\right) }}{1\le i\le \ell}.
\]
\end{thm}

\begin{proof}
  For simplicity of notation, we use the expression \eqref{eq:threetermMod} for $P_{N}^{(N,\rho,\e)}(x,y)$. Let  $\rho=0,1$,  note
  that $c_i^{(N+\ell,\rho,-\ell)}$ vanishes for $1 \leq i \leq N+\ell-1$ only for $i=N$, then is straightforward to verify that
  \[
    P_{N+\ell}^{(N+\ell,\rho,-\ell)}(x,y) = q_1(x,y)q_2(x,y)
  \]
  where  $q_1(x,y)$ is of degree $N$ and $q_2$ of degree $\ell$ with respect to $x,y$. Now, since
  $2(N + \ell) - \ell  = 2 N + \ell$, we verify directly that
  \[
    a_i^{(N+\ell,\rho,-\ell)} = a_i^{(N,\rho,\ell)}   \qquad c_i^{(N+\ell,\rho,-\ell)} = c_i^{(N,\rho,\ell)}
  \]
  and similarly for the diagonal entries of $\bm{D}_N$. This gives $q_1(x,y)= P_{N}^{(N,\rho,\ell)}(x,y)$ and
  $q_2(x,y)=A_{N}^{(\ell,\rho)}(x,y)$. The desired form for $A_{N}^{(\ell,\rho)}(x,y)$ is obtained using elementary continuant transformations (i.e. determinant operations for tridiagonal matrices).
\end{proof}

Note that once $\rho \in \{0,1\}$ is fixed, $A_{N}^{(\ell,\rho)}(x,y)$ is a polynomial in the variables $x,y$ and $N$ with rational coefficients. In general $A_{N}^{(\ell,\rho)}(x,y)$ is a rational function on $x,y,\rho,N$,
for instance,
\[
   A_{N}^{(1,\rho)}(x,y) = 2(N+1)(2 N + 2 \rho +1)x +y -\frac{ (\rho -1) \rho (N+1)  (2 N+2 \rho +1)}{( \rho -\tfrac32) (\rho + \tfrac{1}2)}.
\]
Note that $A_{N}^{(\ell,\rho)}(x,y)$ is not defined as certain values of $\rho$, for instance, in the example above it has a pole at $\rho=\tfrac32$.

\begin{rem}
    Recall that the values of $\rho$ corresponding to the even and odd solutions of the 2pAQRM are $\rho \in \{0,1\}$, and that due to the presence  of the residual term $ R_N^{(\rho,\e)}(x,y)$, the polynomial $P_{N+\ell}^{(N+\ell,\rho,-\ell)}(x,y)$ does not factor for general $\rho$. However, there are other values of $\rho$ that give  factorizations or partial factorizations of these polynomials.

  For instance, for $\rho \in \{-1,-2\}$ (resp.  $\rho \in \{-3,-4\}$), the expression \eqref{eq:threetermMod} in Proposition \ref{prop:detexp2} is actually the determinant of a tridiagonal matrix. This lead to a factorization of $P_{N+\ell}^{(N+\ell,\rho,-\ell)}(x,y)$ into two polynomials of degree $N-1$ and $\ell+1$ (resp. $N-2$ and $\ell+2$).

  In addition, other values of $\rho$ give factorizations of the polynomial
  \[
    \bar{P}_{N+\ell}^{(N+\ell,\rho,-\ell)}(x,y) := \det \left( \bm{I}_{N+\ell} y + 4 \bm{D}_{N+\ell} x + \bm{C}_{N+\ell} \right).
  \]
  It is easy to verify that for $\rho=2 i + \delta$  with $-(N-1) \leq i \leq \ell$ and $\rho \in \{0,1\}$, the equation
  \[
    c_i^{(N+\ell,\rho,-\ell)} = 0,
  \]
  is valid for $2 \leq i \leq N+\ell$. Therefore, the polynomial $\bar{P}_{N+\ell}^{(N+\ell,\rho,-\ell)}(x,y)$ factors into two polynomials of degree $N+[\frac{\rho}2]$ and $\ell - [\frac{\rho}2]$. 

  Although only the cases $\rho \in \{0,1\}$ can be associated in a natural way to the 2pAQRM and its spectrum, using representation theory it is possible to define a generalization of the 2pAQRM that allows a general value of $\rho$, corresponding to a parameter of a certain $\mathfrak{sl}_2$-module. It may be an interesting problem to interpret the factorizations of these polynomials in the setting of the generalized models. The definition of these models and their properties will be treated in a forthcoming paper \cite{RBW2027}.
\end{rem}

\subsection{Positivity of the inter-constraint polynomial}
\label{sec:pos}

In this section, we prove that the polynomial  $A_{N}^{(\ell,\rho)}(x,y)$ is positive for $x,y>0$. To do this, we show that for fixed $x>0$, it is not possible for the polynomial $A_{N}^{(\ell,\rho)}(x,y)$ to have roots as a polynomial in $y$. The proof follows the method of \cite{KRW2017} for the case of AQRM, and is based on interpreting the zeros of the polynomial as eigenvalues of the matrices involved in the associated determinant expressions (cf. Theorem \ref{thm:div}). 

We start by obtaining  the value of $A_{N}^{(\ell,\bar{\rho})}(x,y)$ when $y=0$ for $\bar{\rho}=0,1$. 

\begin{prop}
  For $\rho=0,1$, we have
  \[
    A_{N}^{(\ell,\rho)}(x,0)  = 4^{\ell} (N+1)_\ell (N + \bar{\rho} +1/2)_\ell x^\ell.
  \]
  In particular, $A_{N}^{(\ell,\rho)}(x,0)>0$ for $x>0$.
\end{prop}

Note that, for general $\rho \neq 0,1$, $A_{N}^{(\ell,\rho)}(x,0)$ does not need to be a monomial in $x$.

\begin{proof}
  Let us define the polynomial $J_{k}(x)$ by
  \begin{align*}
    J_{k}(x) =& \left(4 x - 2 + \frac{(1-4\ell^2)}{2 \left(-2 k+\ell+\rho -\frac{1}{2}\right)\left(-2 k+\ell+\rho +\frac{3}{2}\right)} \right)J_{k-1}(x) \\
             & \qquad - \frac{(-2 (k-1) + \rho-1)_2 (2 (\ell-k+1)+\rho-1 )_2}{\left(\ell-2 (k-1)+\rho -\frac{1}{2}\right)_2 \left(\ell -2 (k-1)+\rho -\frac{3}{2}\right)_2 } J_{k-2}(x) 
  \end{align*}
  with initial conditions $J_{0}(x)=1$ and $J_{-1}(x)=0$, then
  $A_{N}^{(\ell,\rho)}(x,0) = (N+1)_\ell (N + \bar{\rho} +1/2)_\ell J_{\ell}(x)$. It is not difficult to verify 
  by induction that
  \[
    J_{k}(x) = 4^k \sum_{i=0}^k (-1)^{k-i} \binom{k}{k-i} \frac{(\ell-k)_{k-i}}{(\ell - \tfrac{4k+1-2\rho}{2})_{k-i}} x^i,
  \]
  thus $J_{\ell}(x) = 4^\ell x^\ell$, as desired.
\end{proof}

On the other hand, we have the following result for $x=0$.

\begin{prop}
  For $\rho=0,1$, we have
  \[
    A_{N}^{(\ell,\rho)}(0,y)  = \prod_{i=0}^{\ell-1} (y + 4i(\ell-i)).
  \]
  In particular, $y=0$ is a root of $A_{N}^{(\ell,\rho)}(0,y)$.
\end{prop}

\begin{proof}
  From the definition, we see that for $k\geq 1$ we have
  \[
    P_{k}^{(N,\rho,\e)}(0,y) = \prod_{i=1}^k \left( y - 4 i (i+\e) \right).
  \]
  It follows that
  \begin{align*}
    P_{N+\ell}^{(N+\ell,\rho,-\ell)}(0,y) &= 
     \prod_{i=1}^{\ell} \left( y - 4 i (i-\ell) \right) \prod_{i=\ell+1}^{N+\ell} \left( y - 4 i (i-\ell) \right) \\
       &= \prod_{i=1}^{\ell} \left( y - 4 i (i-\ell) \right) \prod_{i=1}^{N} \left( y - 4 i (i+\ell) \right) =
                                \prod_{i=0}^{\ell-1} \left( y + 4 i (\ell-i) \right) P_{N}^{(N,\rho,\ell)}(0,y). 
  \end{align*}
 The result is then compared with the factorization of $P_{N+\ell}^{(N+\ell,\rho,-\ell)}(0,y)$.
\end{proof}

Summarizing, for $x=0$ the roots of $A_{N}^{(\ell,\rho)}(0,y)$ are non-negative, and for $x>0$, $y=0$ is not a root of $A_{N}^{(\ell,\rho)}(x,0)$. We also note that for $x>0$, transforming the determinant expression into a determinant of a symmetric matrix, we see that the polynomials $P_{N+\ell}^{N+\ell,\rho,-\ell)}(x,y)$, $P_{N}^{N,\rho,\ell)}(x,y)$, and thus $A_{N}^{(\ell,\rho)}(x,y)$ have real roots with respect to $y$ (cf. Cor. 3.5 of \cite{KRW2017}). Therefore,  we can complete the proof if we show that for a fixed positive $x>1$ all the roots of $A_{N}^{(\ell,\rho)}(x,y)$ with 
respect to $y$ are positive.

From the expression given in Theorem \ref{thm:div}, we see that for fixed $x$  the zeros of $A_{N}^{(\ell,\rho)}(x,y)$ are given by the negative of the eigenvalues of the matrix
\[
  M(x;\ell,\rho) = \Tridiag{2 (2x-1) + \frac{1-4 \ell^2}{2 \left(-2 i+\ell+\rho -\frac{1}{2}\right)\left(-2 i+\ell+\rho +\frac{3}{2}\right)}}{\frac{(-2 i + \rho )(-2 i + \rho-1)}{\left(\ell-2 i+\rho -\frac{1}{2}\right)\left(\ell-2 i+\rho +\frac{1}{2}\right)  }}{\frac{ (2 (\ell-i)+\rho ) (2 (\ell-i)+\rho-1 )}{\left(\ell -2 i+ \rho -\frac{1}{2}\right) \left(\ell -2 i+\rho -\frac{3}{2}\right) }}{1\le i\le \ell}.
\]

\begin{thm}
  For $\rho=0,1$, we have $A_{N}^{(\ell,\rho)}(x,y)$ for $x,y>0$.
\end{thm}

\begin{proof}
  We proceed as in Lemma 3.10 of \cite{KRW2017}. As mentioned above, what it is left to verify that for a large enough $\bar{x}$ the polynomial $A_{N}^{(\ell,\rho)}(\bar{x},y)$ has all negative roots with respect to $y$. This is equivalent to checking that all the eigenvalues of the matrix
\[
  M(x;\ell,\rho) = \Tridiag{2 (2x-1) + \frac{1-4 \ell^2}{2 \left(-2 i+\ell+\rho -\frac{1}{2}\right)\left(-2 i+\ell+\rho +\frac{3}{2}\right)}}{\frac{(-2 i + \rho )(-2 i + \rho-1)}{\left(\ell-2 i+\rho -\frac{1}{2}\right)\left(\ell-2 i+\rho +\frac{1}{2}\right)  }}{\frac{ (2 (\ell-i)+\rho ) (2 (\ell-i)+\rho-1 )}{\left(\ell -2 i+ \rho -\frac{1}{2}\right) \left(\ell -2 i+\rho -\frac{3}{2}\right) }}{1\le i\le \ell}.
\]
are positive for $\bar{x}$. In this case, we verify this using the Gershgorin circle theorem with $\bar{x}$ such
that
\begin{align*}
  2(2\bar{x}-1) >& 1+ \frac{4 \ell^2-1}{2 \left(-2 i+\ell+\rho -\frac{1}{2}\right)\left(-2 i+\ell+\rho +\frac{3}{2}\right)} \\
                 &+ \frac{(-2 i + \rho )(-2 i + \rho-1) (2 (\ell-i)+\rho ) (2 (\ell-i)+\rho-1 )}{\left(\ell-2 i+\rho -\frac{1}{2}\right)\left(\ell-2 i+\rho +\frac{1}{2}\right)\left(\ell -2 i+ \rho -\frac{1}{2}\right) \left(\ell -2 i+\rho -\frac{3}{2}\right)  },
\end{align*}
for all $i=1,2,\cdots,\ell$, completing the proof. 
\end{proof}

These two results of this section, that is, the divisibility of the constraint polynomials and the positivity of the inter-constraint polynomial $A_{N}^{(\ell,\rho)}(x,y)$, along with the proof of linear independence of the solutions are sufficient to give a characterization of degeneracies and describe the spectral structure of the 2pAQRM in Section \ref{sec:degstruct}.

\section{Another significance of constraint polynomials}
\label{sec:qtpol}

The natural setting for the constraint polynomials is to describe Juddian solutions and shared-parity degenerate eigenvalues. Indeed, as we have seen in the previous sections, the existence of these solutions are in correspondence with algebraic properties of the constraint polynomials. It is more surprising that their role seems to extend beyond Juddian solutions and the fact that they appear to contain, at least approximately, information about the complete spectrum of these models.

There is a precedent for this in the generalized adiabatic approximation (GAA) for the AQRM proposed in \cite{LB2021b}. The authors obtain an improved agreement with all eigenvalues for a wide range of parameters $g,\Delta>0$ by replacing the Laguerre polynomials, used in the usual adiabatic approximation and coming from the case $\Delta = 0$ (see e.g. \S3 in \cite{KRW2017}), with constraint polynomials, in particular, there is agreement on the degenerate point. To the best knowledge of the authors, such a GAA has not been developed for 2pAQRM, but we expect a similar agreement with the spectral curves, mainly because of the similarities between the polynomials in the two models and the existence of a covering relation \cite{RBW2023, N2024}.

The GAA applies only for eigenvalue curves containing crossings. Therefore, it does not apply to a number of lower
energy eigenvalue curves. For the AQRM, in \cite{RBW2022} the authors observed that the inter-constraint polynomial,
employing the naming in this paper, gives an excellent approximation of these eigenvalue curves for large coupling. 

In this section, we develop a similar point of view for the 2pAQRM based on the spectral approximation properties
of the constraint polynomials. In particular, we attempt to give a more concrete approach to the phenomenon based
on the algebraic and geometric properties of the curves.

We remark here that there appears to be an explicit connection between the inter-constraint polynomial and the
symmetric operators (cf. Section \ref{sec:symdeg}), giving more evidence to significant position of the constraint polynomials in the theory of the QRM and its generalizations.

\subsection{Properties of the inter-constraint polynomial}
\label{sec:propicp}

A consequence of the explicit expression for the inter-constraint polynomial in Theorem \ref{thm:div} is that
the parameter $N$ of the degenerate Juddian eigenvalue $\lambda = \tau (2N + \rho + \tfrac12 + \ell)$ may be treated as another variable. This idea is the key to unveil various mathematical structures in the 2pAQRM. In this section, we start by describing the properties of the polynomial in the new variable.

Let us write $A^{(\ell,\rho)}(N,x,y) := A_{N}^{(\ell,\rho)}(x,y)$. We first see that by applying a shift on this variable, we can consider both the even and odd cases (i.e. $\bar{\rho}=0,1$) in a unified way.

\begin{lem}
  \label{lem:A01}
  For $\bar{\rho}=0,1$, we have
  \begin{align}
    \label{eq:commonPolyA}
        A^{(\ell,0)}(\tfrac{\lambda}{2\tau} - \tfrac14 -\tfrac{\ell}2,(2g)^2,\Delta^2) = A^{(\ell,1)}(\tfrac{\lambda}{2\tau} - \tfrac{1}2- \tfrac14 -\tfrac{\ell}2,(2g)^2,\Delta^2),
  \end{align}
  that is, the polynomial $A^{(\ell,\bar{\rho})}(\tfrac{\lambda}{2\tau} - \tfrac{\bar{\rho}}2- \tfrac14 -\tfrac{\ell}2,x,y)$ in $\lambda,x,y$ does not depend on $\bar{\rho}$.
\end{lem}

\begin{proof}
  We verify the result from the determinant expressions for both $\rho=0$ and $\rho=1$ and the respective 
  change of variables of $N$.

  From the determinant expression of Theorem \ref{thm:div}, since the factor
  \[
    (N+1)_\ell(N+\rho+\tfrac12)_\ell
  \]
  is equal to
  \[
    (\tfrac{\lambda}{2\tau} + \tfrac{1}4 - \tfrac{\ell}{2})_\ell (\tfrac{\lambda}{2\tau} + \tfrac{3}4 - \tfrac{\ell}{2})_\ell
  \]
  in both cases. Therefore, it is enough to verify that in both cases the determinant
  \begin{align}
    \label{eq:detexp5}
    \det \left(2 \bm{I}_\ell (2x-1) + \bm{D}_\ell^{(N,\rho)} y + \bm{V}_\ell^{(\rho)} \right)
  \end{align}
  is equal.
  Define
  \[
    d_i^{(\rho)} := (\tfrac{\lambda}{2\tau} - \tfrac{\rho}2- \tfrac14 -\tfrac{\ell}2+i)(\tfrac{\lambda}{2\tau} - \tfrac{\rho}2 + \tfrac{\ell}2-i+\rho+\tfrac14),
  \]
  and notice that $d_i^{(\rho)}$ corresponds to the denominators in $\bm{D}_\ell^{(N,\rho)}$.
  Then we verify directly that
  \[
    d_{\ell-i+1}^{(1)} = d_{i}^{(0)}.
  \]
  
  Similarly, defining
  \begin{align*}
    c_i^{(\rho)} &:= ((-2i + \ell) \rho - \tfrac{1}2)(-2 i + \ell + \rho + \tfrac32), \\
    u_i^{(\rho)} &:= \frac{(-2 i + \rho)(-2 i + \rho - 1)}{(\ell-2 i + \rho - \tfrac12)(\ell - 2 i + \rho + \tfrac12)}, \\
    l_{i}^{(\rho)} &:= \frac{(2 (\ell-i) + \rho)(2 (\ell- i) + \rho - 1)}{(\ell-2 i + \rho - \tfrac12)(\ell - 2 i + \rho - \tfrac32)}
  \end{align*}
  corresponding to the diagonal, upper diagonal and lower diagonal entries of $\bm{V}_\ell^{(\rho)}$.
  We verify directly then that
  \[
    c_{\ell-i+1}^{(1)} = c_{i}^{(0)},
  \]
  and
  \[
    u_{\ell-i}^{(1)} = l_{i}^{(0)}, \qquad l_{\ell-i}^{(1)} = u_{i}^{(0)}.
  \]
  Therefore, when we expand the determinant as a continuant, the expansion of \eqref{eq:detexp5} for $\rho=0$ starting at the top-left corner is equal to the expansion of \eqref{eq:detexp5} for $\rho=1$ starting at the bottom-right corner, proving the result.
\end{proof}

Another interpretation of the lemma above is that the matrices in the determinant expression
of both sides of \eqref{eq:commonPolyA} (that is, for $\bar{\rho}=0$ and $\bar{\rho}=1$) differ
only by a permutation of rows and columns, and therefore the determinants are equal.

\begin{dfn}
  We define $A^{(\ell)}(z)$ as the common value of \eqref{eq:commonPolyA} for $x=(2g)^2,y=\Delta^2$, that is,
  \[
    A^{(\ell)}(z) := A^{(\ell,0)}(\tfrac{\lambda}{2\tau} - \tfrac14 -\tfrac{\ell}2,(2g)^2,\Delta^2)
  \]
  where $z = \tfrac{\lambda}{\tau}$. In the following, for instance, when $\Delta$ is fixed and $A^{(\ell)}(z)$ is regarded as a
  function of $z$ and $g$, we will denote it as $A^{(\ell)}(z;g)$.
\end{dfn}

From the determinant expression it is easy to verify the following property.

\begin{prop}
  \label{prop:coefA}
  The polynomial $A^{(\ell)}(z)$ is a polynomial with integer coefficients, that is,  $A^{(\ell)}(z) \in \Z[z,g,\Delta]$. \qed
\end{prop}

\begin{ex}
  \label{ex:Apoly1}
  We give the polynomial $A^{(\ell)}(z)$ for small values of $\ell$, 
  \begin{align*}
    A^{(1)}(z) &=  4 g^2 z^2 + \Delta ^2  -g^2, \\
    A^{(2)}(z) &= 16 g^4 z^4 + 4(2 g^2 \Delta ^2 -10 g^4) z^2 +  \Delta ^4 +4 \Delta ^2 +9 g^4 -10 \Delta ^2 g^2, \\
    A^{(3)}(z) &= 64 g^6 z^6 + (48 \Delta ^2 g^4-560 g^6) z^4 + (1036 g^6-280 \Delta ^2 g^4+12 \Delta ^4 g^2+64 \Delta ^2 g^2) z^2 \\
               & \qquad + \Delta ^6+16 \Delta ^4+64 \Delta ^2-225 g^6+259 \Delta ^2 g^4-35 \Delta ^4 g^2-272 \Delta ^2 g^2,
  \end{align*}
  where we recall that $z = \lambda/\tau$.
\end{ex}

Note that in the examples, the variable $z$ always appear in even powers. 

\begin{prop}
  \label{prop:evenpow}
  Let $\ell\geq 0$. Then $A^{(\ell)}(z)$ is a polynomial on $z^2$.
\end{prop}

\begin{proof}
  By Lemma \ref{lem:A01}, it is enough to verify the result $\rho=0$. First, we replace
  $N = \tfrac{z}2 - \tfrac14 - \tfrac{\ell}2 $ in
  \begin{align}
    \label{eq:lemprod1}
    (N+i)(N+\ell-i+\tfrac12),
  \end{align}
  to obtain
  \[
    \frac{1}4 (z^2 - (\ell-2 i + \tfrac12)^2).
  \]
  Next, we notice that $N$ appears only in products of type \eqref{eq:lemprod1} inside the determinant of the
  expression of $A^{(\ell,0)}(N,x,y)$ in Theorem \ref{thm:div}. Similarly, by rearranging the factors of
  \[
    (N+1)_\ell (N+\rho+\tfrac12)_\ell,
  \]
  into products of type \eqref{eq:lemprod1} and replacing $N = \tfrac{z}2 - \tfrac14 - \tfrac{\ell}2 $ we see
  that $z$ appears only in even powers in the determinant expression of $A^{(\ell,0)}(N,x,y)$ and the result follows.
\end{proof}

The fact that the variable $z$ appears in even powers is relevant for the explicit relation between the
inter-constraint polynomial and the hidden symmetry operator in Section \ref{sec:symdeg}. 

\subsection{Excellent approximation by algebraic curves}
\label{sec:excapprox1}

While the natural setting for the inter-constraint polynomial is that of Juddian solutions, using the
variable $z$ we can connect these polynomials with general eigenvalues. For instance, we might ask if there is any relation between the curves described by the polynomial $A^{(\ell)}(z)$ and the spectrum of the 2pAQRM.

For instance, for $\Delta>0$ fixed, in the $(g,E)$-plane, the curves described by
\begin{align}
  \label{eq:pcurves}
  A^{(\ell)}(E)&=0 \qquad  \text{(normalized)} \\
  A^{(\ell)}(\tau E)&=0, \qquad  \text{(non-normalized)}  \nonumber 
\end{align}
give an approximation of the $\ell$ lower energy curves of each parity for the 2pAQRM, as shown in Figure \ref{fig:EA1}.
This is 2pAQRM version of the excellent approximation conjecture of the AQRM (see e.g. \cite{RBW2022}).

\begin{figure}[ht]
  \centering
  \subfloat[$\ell=1$]{
    \includegraphics[height=4cm]{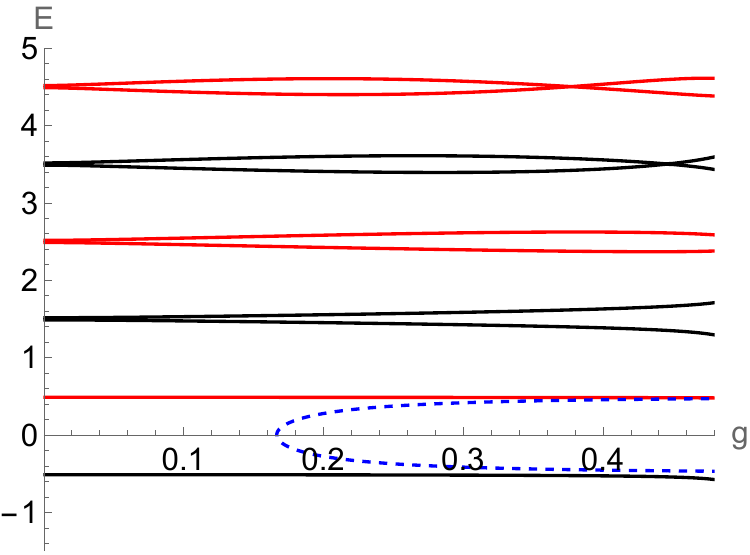}}
  ~ \qquad
  \subfloat[$\ell=2$]{
    \includegraphics[height=4cm]{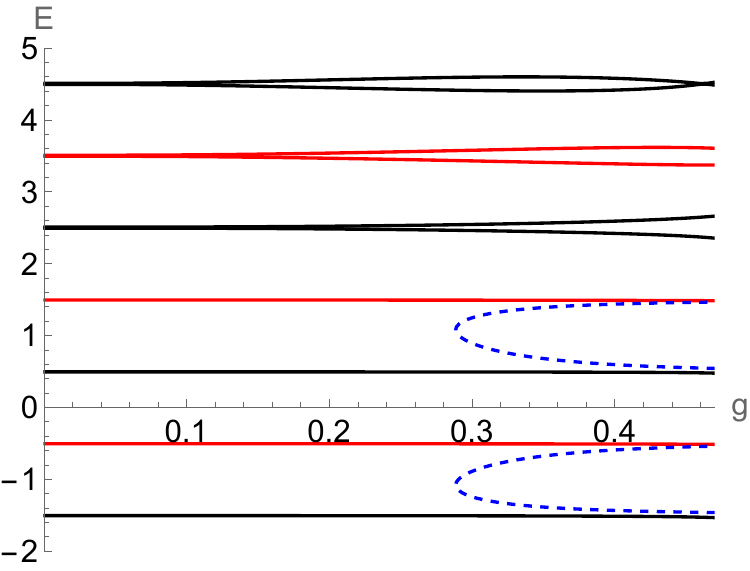}} \\
  \subfloat[$\ell=5$]{
    \includegraphics[height=4cm]{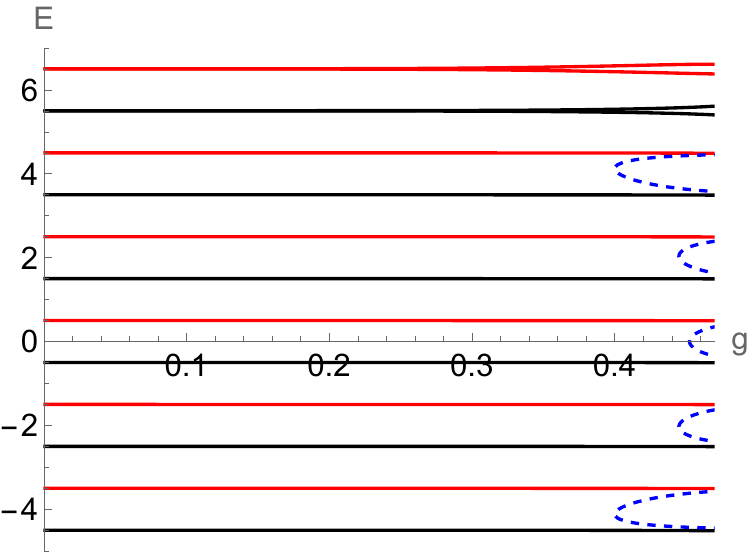}}  
  \caption{Normalized spectral curves  \eqref{eq:pcurves} (blue) for $\Delta=\frac16$ for $\ell \in \{1,2,5\}$.}.
  \label{fig:EA1}
\end{figure}

In our previous work \cite{RBW2022} we did not state explicitly a conjecture for this phenomenon for the AQRM. It is not easy to prove that the curves \eqref{eq:pcurves} do not intersect the spectral curves or to give a quantitative expression for the approximation. In fact, we can see that in the limit cases $\Delta \to 0$ and $g\to \infty$ (see  \cite{RBW2024} or \cite{HS2025}) these curves actually coincide, while for small $g$ there is no approximation at all. Moreover, the locus of the curves \eqref{eq:pcurves} may be empty for $\Delta$ large enough. The significance of the approximation appears to be related to another polynomial associated to the hidden symmetry, as we explain in Section \ref{sec:symdeg} below.

We also note that the curves \eqref{eq:pcurves} are symmetric with respect the $g$-axis due to fact that the polynomial $A^{(\ell)}(z)$ is a polynomial in $z^2$ (cf. Proposition \ref{prop:evenpow}). If the approximation holds, then it also implies that there is a non-trivial symmetry (with respect to $0$) for low energy curves, at least for $g$ large enough.

In this section we focus on the algebraic properties of the curves \eqref{eq:pcurves}. For instance, in the graphs in Figures \ref{fig:EA1}, it is clear that for large $g$, the polynomial $A^{(\ell)}(z)$ has $\ell$ real roots. The hyperellipitic curve appearing in this conjecture arises from the discussion of the hidden symmetry (see Section \ref{sec:symdeg}). In fact, a hyperelliptic curve appears naturally in that setting and we conjecture that it is equivalent to the one defined by \eqref{eq:Helliptic} (see Conjecture \ref{conj:symDeg}).

\begin{conject}
  For sufficiently large $g$, there are exactly $\ell$ real roots of $A^{(\ell)}(x;g)=0$. Particularly, the algebraic curve
  \begin{align}
    \label{eq:Helliptic}
    y^2= A^{(\ell)}(x;g)
  \end{align}
  defines a hyperelliptic curve. Moreover, there are exactly $[\ell/2]$-points, where the polynomial
  $A^{(\ell)}(x,g)=0$ has multiple real roots.
\end{conject}

As an illustration of the conjecture, let us consider the case of $\ell=3$. In this case, the equation
\begin{align}
  \label{eq:elliptic}
  y^2 = A^{(3)}(z),
\end{align}
describes an elliptic curve for generic values of parameters $g$ and $\Delta$. Since $A^{(3)}(z)$ is a polynomial in
$z^2$ let us take
\[
  x = \left(2g z\right)^2,
\]
so that \eqref{eq:elliptic} is an elliptic curve $\mathcal{E}$ in Weierstrass form with coefficients
\begin{align*}
  a_2 &= 3 \Delta ^2-35 g^2, \\
  a_4 &= 3 \Delta ^4+16 \Delta ^2+259 g^4-70 \Delta ^2 g^2, \\
        a_6 &= \Delta ^6+16 \Delta ^4+64 \Delta ^2-225 g^6+259 \Delta ^2 g^4-35 \Delta ^4 g^2-272 \Delta ^2 g^2.
\end{align*}
After the customary transformations (see e.g. \cite{K1992}),  \eqref{eq:elliptic} is reduced to
\[
  y^2 = x^3 - 27 c_4 x - 54 c_6,
\]
where 
\begin{align*}
  c_4 &= 256 \left(28 g^4-3 \Delta ^2\right) \\
  c_6 &= 2048 \left(160 g^6+9 \Delta ^2 \left(4 g^2-3\right)\right),
\end{align*}
with the discriminant
\[
  {\rm Disc}(\mathcal{E}) = 2^{16} \left(-4 \Delta ^6+2304 g^{12}+64 \Delta ^2 \left(5-23 g^2\right) g^6+\Delta ^4 \left(64 g^4+72 g^2-27\right)\right).
\]

For fixed $\Delta$, the discriminant equation ${\rm Disc}(\mathcal{E})=0$ is a degree $12$ equation with respect to $g$ (degree $6$ on $g^2$). For $\Delta = \frac16$, the discriminant has only $2$ real roots, the positive root being \( \bar{g} \approx 0.34578 \). In Figure \ref{fig:EA2} we show the approximation of elliptic curves by the inter-constraint polynomial (the line in blue
corresponds to $g=\bar{g}$) and the singular elliptic curve corresponding to this parameter.

\begin{figure}[ht]
  \centering
  \subfloat[$\ell=3$]{
    \includegraphics[height=4cm]{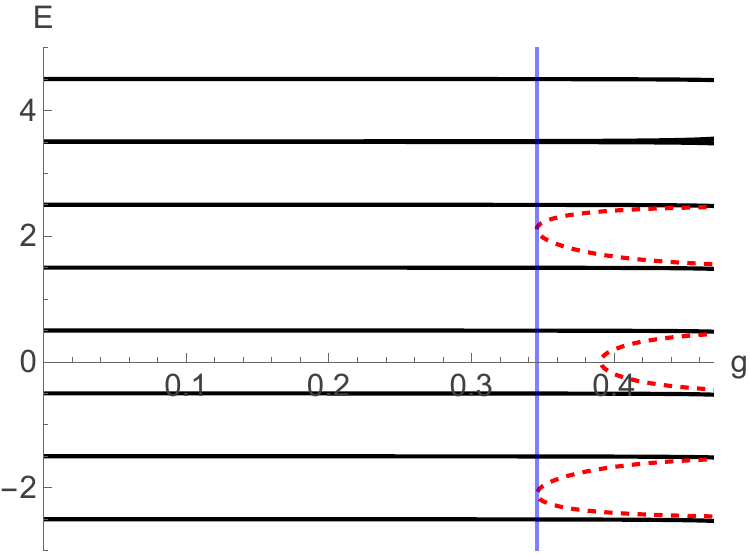}} \\
  \subfloat[$g<\bar{g}$ ]{
    \includegraphics[height=4cm]{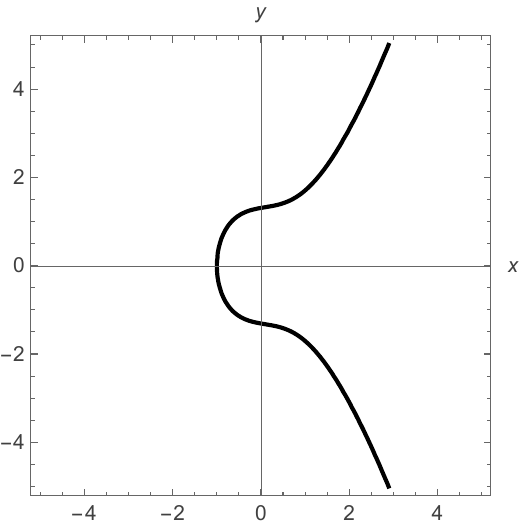}}
  ~ \quad
  \subfloat[$g=\bar{g}$ ]{
    \includegraphics[height=4cm]{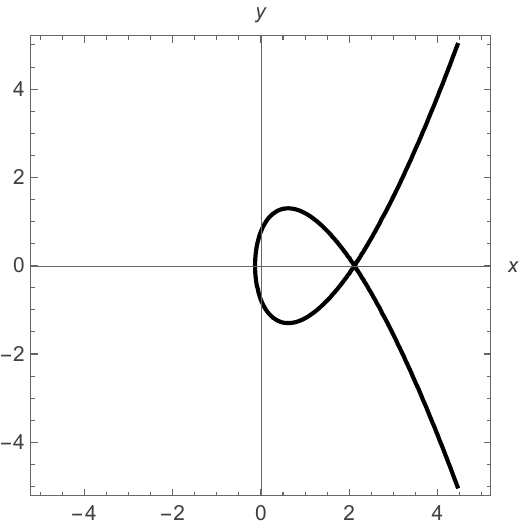}}
  ~ \quad
    \subfloat[$g>\bar{g}$ ]{
      \includegraphics[height=4cm]{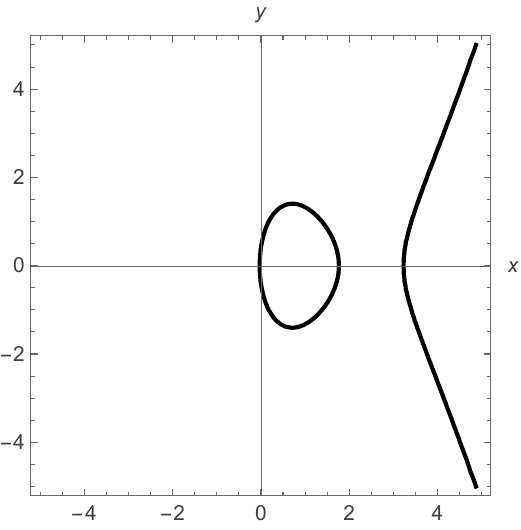}}
    \caption{Approximation of spectral curves by the curves \eqref{eq:pcurves}(red) with vertical tangent ($g=\bar{g}$) shown in blue. The corresponding elliptic curves \eqref{eq:elliptic} for different values of parameter $g$ ($g=0.1$, $g=\bar{g}$, $g=0.38$, respectively).}
    \label{fig:EA2}
\end{figure}

We note that for any real value of $g$ with the exception of $\bar{g}$, the equation \eqref{eq:elliptic} defines an
elliptic curve in the usual sense (see Figure \ref{fig:EA2}(b-d)). In particular, while there is another value $g>\bar{g}$  point where the polynomial $A^{(\ell)}(z)$ appears to have a double root (corresponding to another vertical tangent to the approximation curve in Figure \ref{fig:EA2}), it does not correspond to a singular elliptic curve. We might ask about the situation in general.

\begin{prob}
  For $\ell>0$, describe the structure of algebraic curve $y^2= A^{(\ell)}(x;\hat{g})$ where the vertical line $g=\hat{g}$ is tangent to the curve $A^{(\ell)}(x,g)=0$?
\end{prob}

While this problem is algebro-geometric in nature, it is also interesting to see if the points of tangency of the curves $A^{(\ell)}(x;g)=0$ with vertical lines, that is, the leftmost points of each of the connected components, have a physical meaning. For instance, we may ask their relation to parameter regime classifications that regions in the $(g,E)$-plane
for the AQRM, as in \cite{EVBSS2017} for the QRM, based on the study for the validity of perturbative criteria. Another important question is whether the behavior of the curve \eqref{eq:pcurves} at points $g$ outside of the interval $(0,1/2)$ can give information for the corresponding 2pAQRM, specifically the critical point $g_c=1/2$.

One of the advantages of considering the hyperelliptic curve \eqref{eq:Helliptic}, instead of the plane curves
\eqref{eq:pcurves} is that we might consider the arithmetical properties. For instance, for the elliptic curve case it is an open direction to study the associated $L$-functions and integral points, and their possible relation with the spectrum of the 2pAQRM and in particular with the special values of its spectral zeta function, i.e., Dirichlet series.  In fact, we recall that the spectral zeta function of the NCHO is known to
have a rich number theoretical structure \cite{KW2023} and we expect a similar structure also for the QRM and 2pQRM.

\begin{rem}
  The paper \cite{BNZ2026}, mentioned in the introduction,  is dedicated to obtaining exact solutions of the 2pAQRM by the Bethe ansatz approach. Concretely, the authors apply this technique to the meromorphic structure of the resulting fourth-order differential equation, arbitrary-state analytical solutions, particularly the Juddian solutions, of the eigenvalue problem of the 2pAQRM for integral value of certain parameters. This approach is
  different that the one we used in Section \ref{sec:deg} or the Heun picture of the 2pAQRM, but it is easy to see it coincides for Juddian solutions.

  In the setting of \cite{BNZ2026}, the exceptional Juddian points, giving degeneracy in the spectrum, are considered to be of great experimental importance and directly associated with enhanced sensitivity and nontrivial dynamical behavior of the system. The reason being that in the vicinity of the special parameter values giving Juddian solutions, eigenstates become highly structured, thus making observables extremely sensitive to small changes in coupling or bias.

  This explanation is consistent with the discussion of this section regarding the significance of Juddian solutions and their associated polynomials. For instance, the excellent approximation conjectures of Section \ref{sec:excapprox1} and the generalized adiabatic approximation discussed there.
\end{rem}

\section{Characterization of degeneracies for the 2pAQRM via monodromy}
\label{sec:degstruct}

In this section we use the theory developed in the previous section to characterize the spectral structure of the 2pAQRM. To accomplish this, using the monodromy data from the Heun picture we first show that the solutions whose existence is ensured by the divisibility relation of Theorem \ref{thm:div} are linearly independent and therefore constitute actual spectral degeneracies. We also use the monodromy data to describe the opposite-parity degeneracies, previously studied in physics \cite{D2009,MS2019,X2025,X2020}.

\subsection{Multiplicity of Juddian solutions}
\label{sec:linindep}

The simultaneous vanishing of the constraint polynomials does not immediately guarantee that the corresponding Juddian solutions are linearly independent since they are not comparable. Indeed, they correspond to squeezing transformations with different parameters, that is, one corresponds to the Hamiltonian $H_A$ and the other to $H_B$ (cf. Section \ref{sec:deg}).

For the AQRM, there are different ways to prove the linear independence of the degenerate eigenfunctions, including the comparison of the degree of associated irreducible representations \cite{W2017} or even direct comparison of the degree of quasi-exact eigenfunctions. For the 2pAQRM, using the Heun picture and its associated monodromy representation allows to obtain a description of both eigenfunctions at degenerate eigenvalues. Since the structure of singularities is the same, we refer the reader to \cite{RBW2023} for the details of the proofs.

Let us fix an exceptional eigenvalue $\lambda = \tau (2 N + \e + \rho + \frac12) $, then the Riemann scheme of the Heun picture is given by
\begin{align}
  \label{eq:JuddianRS}
  \begin{bmatrix}
    0 & 1 & t & \infty &; \omega \\
    0 & 0 & 0 & a-\frac{1}{2}&; \bar{q}_a \\
    N &  N + 1 + \e  & - N - \e + \frac32 - a & - N
  \end{bmatrix}.
\end{align}
where $a=1$ for the even case and $a=2$ for the odd case. Since one of exponents of the singularity at infinity is a
non-negative integer, a polynomial solution may appear. To describe the solutions in detail we use the monodromy
representation.

Note that the structure is the same as in the case of the aNCHO, so we can proceed in the same way (cf. Theorem 6.3 of \cite{RBW2023}, note the change of sign of $\eta$, resp. $\e$ due to the choice of Heun equation). In particular, for the odd case let us denote the monodromy matrices (i.e. associated with the Riemann scheme \eqref{eq:JuddianRS}) by $\mathbf{A}_i$ ($i=0,1,2,3$) corresponding to singularities at $\omega=0,1,t, \infty$ in that order. Then, we have
\[
  \mathbf{A}_0 \mathbf{A}_1 \mathbf{A}_2 \mathbf{A}_3 = \mathbf{I}_2 
\]
along with a complete analogous situation for the even case. The eigenvalues of the monodromy matrices are computed from the critical exponents of the ODE (i.e. from the Riemann scheme), and are given in Table \ref{tab:eigenMon2}.

\begin{table}[h]
  \centering
  \begin{tabular}{ r | c c c c }
    & $\mathbf{A}_0$ & $\mathbf{A}_1$ & $\mathbf{A}_2$ & $\mathbf{A}_3$  \\
    \hline
    $\epsilon_1$ & 1 & 1 & 1 & -1  \\
    $\epsilon_2$ & 1 & $e^{-2 \pi i \e }$ & $e^{-2 \pi i (\frac12 + \e)}$ & 1
  \end{tabular}
  \caption{Eigenvalues of monodromy matrices for the case of odd Juddian eigenvalues}.
  \label{tab:eigenMon2}
\end{table}

\begin{thm}
  \label{thm:mult2}
  Let $N \in \Z_{\geq 0}$, $\rho \in \{0,1\}$ and \( \lambda =  \tau (2 N + \rho + \frac12 + \e) \)
  with $\e = \ell \in  \Z$. If the parameters satisfy
  \[
    P_{N+\ell}^{(N+\ell,\rho,-\ell)}((2g)^2,\Delta^2) = 0,
  \]
  so that  $\lambda = \tau (2 N + \ell+ \rho + \frac12) $ is a Juddian solution, then 
  \[
    {\rm dim}_{\C}\{ f \in \mathcal{O}(\Omega) \, | \, \Lambda_{1+\rho} f = 0 \} = 2. \qed
  \]
\end{thm}

We note that directly from the monodromy data we also see that there cannot be degenerate exceptional eigenvalues for the case $\e \notin \Z$.
We also obtain the general form of the solutions, one of the solutions is a polynomial (quasi-exact) while the other is a Heun polynomial\footnote{In general, a Heun polynomial is a function of the form $\omega^{\sigma_1}(\omega-1)^{\sigma_2}(\omega-t)^{\sigma_3} p(\omega)$ where $\sigma_i$ are the critical exponent at the singularities and $p(\omega)$ is a polynomial \cite{R1995}.}.

\begin{cor}
  \label{cor:gensolFin}
  Let $\e = \ell \in \Z$, $N \in \Z_{\geq0}$ and $\rho \in \{0,1\}$. Suppose that
  \[
    P_{N+\ell}^{(N+\ell,\rho,-\ell)}((2g)^2,\Delta^2) = 0,
  \]
  so that $\lambda = \tau (2 N + \ell+ \rho + \frac12) $ is a Juddian solution of the 2pAQRM.

  Then, there exist Heun polynomials $H p_1(\omega)$ and $H p_{2}(\omega)$ spanning the space of solutions of $\Lambda_{1+\rho} f = 0$
  holomorphic at $\Omega$. We have
  \begin{enumerate}
  \item for the even case ($\rho=0$), there are polynomials $f_1(\omega)$ of degree $N + \ell$ and $g_2(\omega)$ of degree at
    most $N + 1$, unique up
    to constant multiples , such that $H p_1(\omega) = f_1(\omega)$ and
    \[
      H p_2(\omega) = \frac{g_2(\omega)\sqrt{\omega-t}}{(\omega- t)^{N+1}},
    \]
  \item for the odd case ($\rho=1$), there are polynomials $f_1(\omega)$ of degree $N+\ell$ and $g_2(\omega)$ of degree at most $N$,
    unique up to constant multiples , such that $H p_1(\omega) = f_1(\omega)$ and
    \[
      H p_2(\omega) = \frac{g_2(\omega)\sqrt{\omega-t}}{(\omega- t)^{N+2}}.
    \]
  \end{enumerate}

  Conversely, suppose there is a solution of  $\Lambda_{1+\rho} f = 0$ that is either
  \begin{itemize}
  \item a rational function in $\omega$ at the origin, or
  \item of the form of the form $(\omega-t)^{\frac12} q(\omega)$ at the origin for a rational function $q(\omega)$,
  \end{itemize}
  then, if
   \[
    P_{N+\ell}^{(N+\ell,\rho,-\ell)}((2g)^2,\Delta^2) = 0,
  \]
  we have
  \[
    \dim_\C \{ f \in \mathcal{O}(\Omega) | \Lambda_{1+\rho} f = 0 \}  = 2,
  \]
  and the solutions are of the form given above.
\end{cor}

Summarizing, since the multiplicity is bounded above by $2$ we automatically obtain the desired result.

\begin{thm}
  For $N \geq0 $ and $\ell \geq 1$, if the system parameters $g,\Delta$ satisfy
  \[
    P_{N+\ell}^{(N+\ell,\rho,-\ell)}((2g)^2,\Delta^2) = P_{N}^{(N,\rho,\ell)}((2g)^2,\Delta^2) = 0,
  \]
  then the corresponding Juddian solutions are linearly independent.
\end{thm}

Corollary \ref{cor:gensolFin} implies that for the shared-parity degenerate solutions one of the solutions is  quasi-exact while the other is not quasi exact (see e.g.  \cite{B2023} for the case $\e=0$). This situation may already by anticipated by the relation between the two models given in Theorem \ref{them:N2024}, since the transformation between eigenfunctions for the two models cannot change from a quasi-exact solution to a non quasi-exact solution or vice versa (see also Remark \ref{rem:juddianHeun}).

\subsection{Opposite-parity degeneracies}
\label{sec:diffparity}

For the 2pAQRM, in addition to the shared-parity degenerate solutions there are also opposite-parity degenerate solutions (cf. \cite{MS2019}), also shown in Figure \ref{fig:tpQRMeigencurves25} (see also Figure \ref{fig:tpQRMeigencurves}). In the figures, the blue lines correspond to the normalized baselines $E = N + \e + \frac12$ and in grey the baselines $E = N + \e + \frac12$ for positive integer $N$.

\begin{figure}[ht]
  \centering
  \subfloat[$\e=1/4, \Delta=1$]{
    \includegraphics[height=4.5cm]{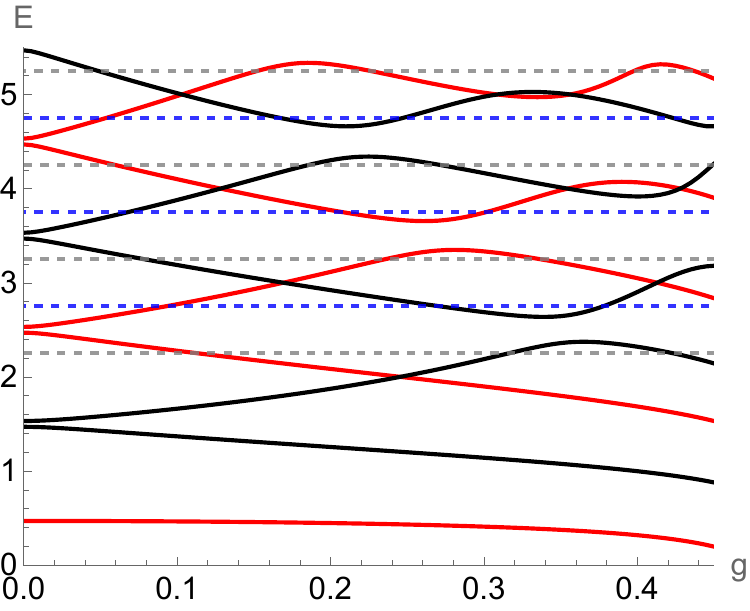}}
  ~ \qquad
  \subfloat[$\e=1/7, \Delta=0.7$]{
    \includegraphics[height=4.5cm]{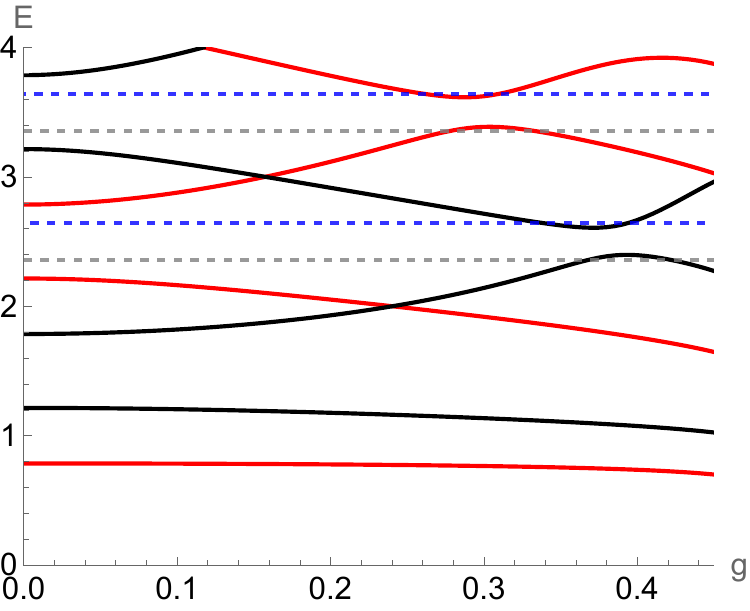}}
  \caption{Normalized spectral curves of the 2pAQRM for $\e \notin \Z$ with baselines.}
  \label{fig:tpQRMeigencurves25}
\end{figure}

For the case of 2pAQRM, the opposite-parity degenerate solutions have been considered by Q. T. Xie \cite{X2020} (see also  \cite{D2009,MS2019,X2025} for other approaches) who showed that these solutions must be of the form \( \lambda = \tau N \) for $N \in \Z_{\geq 0}$. In particular, these do not depend on the bias parameter $\e$. In addition, it is shown that these solutions can be written as finite sums of confluent hypergeometric functions.  The opposite-parity degenerate eigenvalues correspond to the parity symmetry and thus do not give rise to a new type of symmetry operator similar to the case of $\e \in \Z$ (cf. Section \ref{sec:hiddenS} below).

The case $\e = n + \tfrac12$ with $n \in \Z_{\geq 0}$ is of particular interest. It follows from the results in this paper that there cannot be shared-parity spectral degeneracies (necessarily of Juddian type). In Figure \ref{fig:tpQRMeigencurves3} we show the normalized spectral curves where we observe that the different spectral crossings lie on the normalized baselines
\[
  E = n 
\]
for $n \in \Z_{\geq 0}$. In particular, we note that for $\e \in \frac12 + \Z_{\geq0}$, any degeneracy is of exceptional type, similar to the case of AQRM. 

\begin{rem}
A similar situation was observed in \cite{NRBBW2024}, where the spectral spacing statistics of the AQRM resemble that of a single parity of the QRM for bias of the form $\e = \frac{n}2 + \tfrac14$ (corresponding to $\e = n + \tfrac12$ with in the two photon case). 
\end{rem}

\begin{figure}[ht]
  \centering
  \subfloat[$\e=0.5, \Delta=1$]{
    \includegraphics[height=4.5cm]{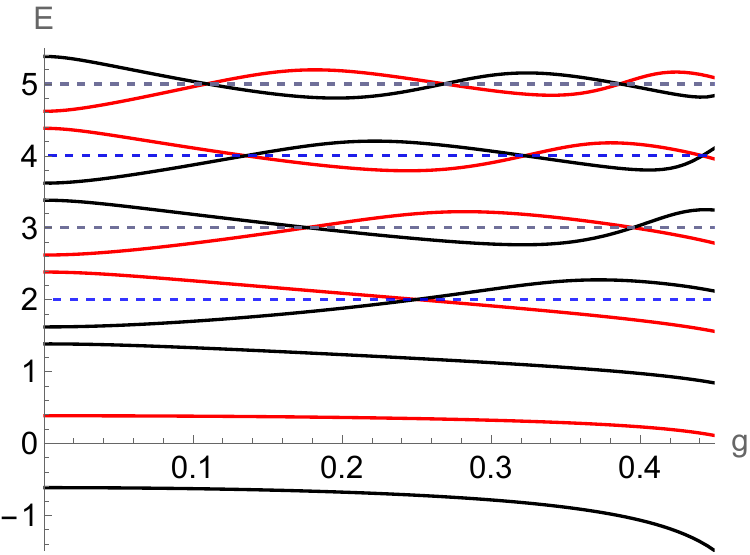}}
  ~ \qquad
  \subfloat[$\e=1.5, \Delta=1.2$]{
    \includegraphics[height=4.5cm]{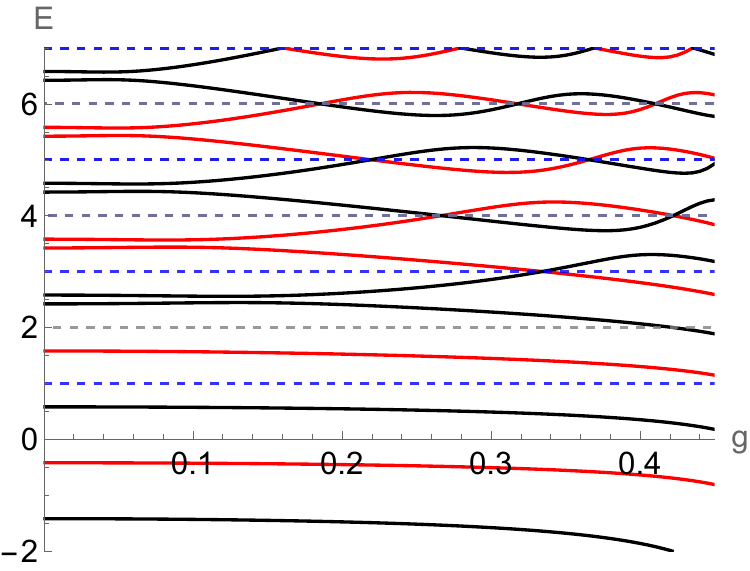}}
  \caption{Normalized spectral curves of the 2pAQRM. Black (resp. red) curves correspond to even (resp. odd) eigenvalue curves.}
  \label{fig:tpQRMeigencurves3}
\end{figure}

According to the results of \cite{X2020}, these solutions should consist of finite sums of confluent hypergeometric functions, that is, in general they are of non-Juddian type. Numerical verifications show that in the cases shown in Figure \ref{fig:tpQRMeigencurves3} one of the solutions appears to be of Juddian type.  Using the singularity data from the Heun ODE picture we can prove that this is the case in general.

\begin{prop}
  For $\e \in \frac12 + \Z_{\geq0}$ any degenerate eigenvalue of the 2pAQRM is exceptional of the form
  \[
    \lambda = \tau N,
  \]
  for $N \in \Z \geq 0$, and is of opposite-parity type. Moreover, one of the eigenfunctions is a Juddian solution while the
  other is non-Juddian.
\end{prop}

\begin{proof}
  First, let us note that for $\e =  M + \frac12$ the exceptional eigenvalues correspond exactly to eigenvalues of
  the form $\lambda = \tau N$, moreover, since $\e \notin \Z$, there cannot be shared-parity degeneracies. Suppose $\lambda = \tau N$
  is a degenerate eigenvalue, then the Riemann scheme of the Heun ODE picture is given by
  \begin{align}
  \label{eq:JuddianRS2a}
  \begin{bmatrix}
    0 & 1 & t & \infty &; \omega \\
    0 & 0 & 0 & \frac{1}{2}&; \bar{q}_a \\
    \frac12(N-M)-\frac{a}2  &  \frac12(N+M) + \frac32-\frac{a}{2}  &  -\frac12(M+N) + 1 -\frac{a}2 &  -\frac12(N-M)+\frac{a}2
  \end{bmatrix}
  \end{align}
  with $a=1,2$ for the even or odd case.
  For $N-M \equiv 0 \pmod{2}$, we note that the Riemann scheme \eqref{eq:JuddianRS2a} with $a=1$ corresponds to the case
  of the Juddian solution
  \[
    \lambda = \tau(N  + \frac12 + \e),
  \]
  where $N= 2n + \rho$ and $\e = M + \frac12$, therefore it corresponds to a single Juddian solution of even or odd
  parity according to the parity of $N$. Similarly, For the case  $N-M \equiv 1 \pmod{2}$, the Juddian solution
  corresponds to the Riemann scheme \eqref{eq:JuddianRS2a} with $a=2$. In both cases, the remaining eigenfunction must
  be a non-Juddian solution, since there are no degeneration between Juddian solutions of the opposite-parity type.
\end{proof}

We note that a full study of the opposite-parity spectral degeneration requires a detailed study of the constraint conditions for the degenerate solutions described in \cite{X2020}, we leave its details for a future study.

\subsection{Degeneracy structure of the 2pAQRM}
\label{sec:degstruct2}

In this section, we summarize the results on spectral degeneracy of the 2pAQRM to characterize its degeneracy picture.
First, we begin by describing the degeneracy occurring in shared-parity eigenvalues.

\begin{thm}
  \label{thm:judddeg}
  For fixed parameters $g,\Delta>0$, $\rho \in \{0,1 \}$ and $\e = \ell \in \Z_{\geq0}$, any shared-parity degenerate eigenvalue of the 2pAQRM is an exceptional eigenvalue of the form
  \[
    \lambda = \tau(2 N + \rho +\tfrac12 + \ell),
  \]
  and consist of a pair of Juddian solutions (described in Section \ref{sec:juddian}) and the constraint condition is
  given by
  \begin{align}
    \label{eq:thmconst}
     P_{N}^{(N,\rho,\e)}((2g)^2,\Delta^2)=0,
  \end{align}
  with $\rho\in \{0,1\}$ according to the parity. In particular, there are no shared-parity degenerate solutions that are not Juddian.
\end{thm}

\begin{proof}
  The first part of the result was proved in Section \ref{sec:linindep}. To complete the result, it remains to show that non-Juddiann solutions  do not constitute shared-parity spectral degeneracies. In fact, this can be done for eigenvalues not of the form $\lambda = \tau N$ for an integer $N \geq 0$ using spectral determinant methods as in the case of the  2pQRM (e.g. \cite{B2023} and \cite{MS2019}).  On the other hand, when eigenvalues of the form $\lambda = \tau N$ are degenerate, they correspond to solutions of different
  parity as shown in \cite{X2020}.
\end{proof}

\begin{cor}
  For $\e \notin \Z_{\geq 0}$, there are no degeneracies on the spectrum of the 2pAQRM consisting
  of two Juddian solutions. Moreover, there are no degeneracies of any type between eigenvalues
  of the same parity. \qed
\end{cor}

To conclude this section, in Table \ref{tab:SS2pAQRM} we present a summary of the spectral structure of the 2pAQRM described by the results above. Note that in this table, we consider any solution given by Heun polynomials (including usual polynomials) as Juddian solutions.

\ctable[
caption = {Spectral structure of the 2pAQRM},
label = tab:SS2pAQRM,
pos     = ht,
width   = \hsize,
left
]{l>{\raggedright}X>{\raggedright}X>{\raggedright}X }{
  \tnote[a]{Here, we say ``non-exceptional'' instead of ``regular'' since regular spectrum usually is typically understood as simple. Due to the form of the eigenvalues, these may be considered a different type of exceptional solutions.}
}{
\NN
 & $\e \in  \Z$ & $\e \in \frac12 + \Z$ & $\e \notin  \frac12 \Z$  
\NN
\cmidrule(r){2-2}\cmidrule(rl){3-3}\cmidrule(l){4-4}
Type  & exceptional or
non-exceptional\tmark[a] & exceptional  & non-exceptional\tmark[a]
\ML
Form  & $\lambda = \tau(N + \frac12)$ or
        $\lambda = \tau  N $        &  $\lambda = \tau  N $ & $\lambda = \tau  N $
\ML
Eigenfunction & two Juddian or
 two non-exceptional & Juddian and non-Juddian & non-exceptional  \ML
 Parity & shared  or opposite & opposite  &  opposite
\ML
Symmetry & hidden or $\Z_2$ & $\Z_2$  & $\Z_2$
}

As it may be expected due to the relation between the two models described in Section \ref{sec:NCHOp}, the spectral structure for the solutions of the aNCHO (see \cite{RBW2023}) is essentially equivalent to that of the 2pAQRM, we leave the detailed discussion to \cite{RBW2027}.

\begin{rem}
  \label{rem:juddianHeun}
  From the discussion on this section and in light of the picture given in Table \ref{tab:SS2pAQRM}, it seems natural to consider as Juddian solutions not only the polynomial solutions of the corresponding Heun ODE, but also Heun polynomials in general.

  From this point of view the curious situation described in \cite{RBW2023}, where a non-finite solution of the aNCHO under the confluence process (covering of models) may result in a Juddian solution of the AQRM, could be
  clarified by describing the degenerate solutions in a similar way as Corollary \ref{cor:gensolFin} above (using
  the corresponding monodromy data of the confluent Heun ODE picture of the AQRM).
\end{rem}

\begin{rem}
  We can give a simpler proof of the second part of the Theorem \ref{thm:judddeg} above based on the relation between the 2pAQRM and the aNCHO of Theorem \ref{them:N2024}, valid for eigenvalues $\lambda$ such that $\lambda > \Delta$.

  Consider a shared-parity degenerate solution $\lambda$ of 2pAQRM for $\e \in \Z_{\geq0}$, then by the relation with
  the aNCHO given by Theorem \ref{them:N2024}, we see that the corresponding solution must be a shared-parity
  degenerate eigenvalue $\tilde{\lambda}$ of the aNCHO of the form
  \[
    \tilde{\lambda}  = \frac{2\sqrt{\alpha \beta (\alpha \beta - 1)}}{\alpha+\beta} \left( 2 L + \rho  + \tfrac12 + 2\eta \right),
  \]
  with $\eta \in \Z$ and $L \in \Z_{\geq0}$, which is easy to see that it must correspond to an
  exceptional eigenvalue $\lambda$ of the 2pAQRM.
  Moreover, since at least one of the eigenfunctions of $\tilde{\lambda}$ must be finite type, then $\lambda$ has a Juddian solution, that is, it satisfies the constraint equation \eqref{eq:thmconst} and therefore it must be a degenerate eigenvalue
  corresponding to two Juddian solutions by the divisibility relation of Theorem \ref{thm:div}, since the multiplicity is bounded above by two, the result follows.
\end{rem}

\section{Remarks on the hidden symmetry for the 2pAQRM}
\label{sec:hiddenS}

As mention in the introduction,  associated with the shared-parity spectral degeneracies of the 2pQRM there is a corresponding symmetry operator. The relation between the two is clearly described by the invariant decomposition of the Hilbert space induced by said operator. The situation for 2pAQRM with $\e \in \R$ is more complicated. As we shown in Section \ref{sec:degstruct}, for $\e \in \Z$, shared-parity degeneracies may exist, according to the roots of the constraint polynomials, but there is not an apparent symmetry in the Hamiltonian, and thus such symmetries came to be called hidden symmetries. 

The existence of the hidden symmetry operators was first systematically studied in \cite{XC2022}. In that paper, the authors show explicitly the operator for small values of $\e  \in \Z$ in \cite{XC2022} and give a heuristic method for the general case, along with certain properties observed from the numerical computations. 

The main result of \cite{XC2022} is the existence of the ``hidden symmetry operator'' $J_\ell$ satisfying
\[
  [J_\ell, H_{\ell}] = 0,
\]
with $J_\ell$ of the form
\[
  J_\ell = e^{\tfrac{\pi i a^\dag a}2} Q_\ell,
\]
where $Q_\ell \in {\rm Mat}_2[\C[a,a^\dag]]$. While the method outlined in \cite{XC2022} may be used to compute the components of the operator $Q_\ell$ for general $\ell$, the authors did not prove the existence of the operator in the general case.

Another relevant observation made in \cite{XC2022} is that the second power of $J_\ell$ can be partially written as polynomial in the Hamiltonian $H_{\ell}$. Concretely\footnote{We note that in \cite{XC2022}, in the expression for $J_1^2$ the factor $e^{\pi i a^\dag a}$ is missing}, the observed that for the cases they computed, the identity 
\[
  J_\ell^2 = \exp(\pi i a^\dag a) p_\ell(H_{\ell}),
\]
holds, where $p_\ell(x) = p_\ell(x;g,\Delta)$ is a polynomial of degree $2\ell$. Note that in particular, we $J_\ell^2$ commutes trivially with $H_{\ell}$. Since the polynomial $p_\ell(x)$ is not constant, there is no canonical decomposition of Hibert space induced by $J_\ell$ as in the symmetric case and thus the relation between the existence of degeneracies and hidden symmetries is not straightforward.

In this section, we prove the existence and  other basic properties of the symmetry operator $J_\ell$, including the existence of the polynomial $p_\ell(x)$. Currently, the only way to compute the polynomial $p_\ell(x)$ is to obtain it from the quadratic equation of $J_\ell$, however in Section \ref{sec:symdeg} we give a conjectural relation with the inter-constraint polynomial.

As we mentioned, the operator $J_\ell$ is not of finite (multiplicative) order and therefore it does not give a $\Z_4$-decomposition of the Hilbert space. An alternative way to study the $\Z_4$-symmetry of the Hamiltonian is to consider the commutant algebra, as it has been suggested in relation to the integrability. An obvious, but important point is that the operator $J_\ell$ is not an element of the matrix Weyl algebra ${\rm Mat}_2[\C[a,a^\dag]]$ and thus some technical care should be taken when considering the commutant algebra. We consider this point of view in Section \ref{sec:SymtpQRM}.

\subsection{Existence and uniqueness of the symmetry operator}
\label{sec:symmetryproof}

In this section we prove the existence of the hidden symmetry operator and its basic properties. The basic scheme of the proof that of the asymmetric QRM \cite{RBBW2021} and is based on a careful analysis of the recurrence relations for the operator $Q_\ell$ determined by the commutation condition and
the most delicate point is the existence of the operator.

First, let us consider  $Q\in  {\rm Mat}_2[\C[a,a^\dag]]$ given by
\begin{align}
  \label{eq:Qdef}
  Q =
  \begin{bmatrix}
    A & B \\
    C & D
  \end{bmatrix}
\end{align}
where $A,B,C,D \in \C[a,a^\dag]$.

For simplicity, we write the Hamiltonian of the 2pAQRM as $H^{\e}$, is given by
\[
  H^{\e} =
  \begin{bmatrix}
      a^\dag a + \tfrac12 + g ((a^\dag)^2+a^2) + \tau \e  & \Delta \\
    \Delta &  a^\dag a + \tfrac12 - g ((a^\dag)^2+a^2) - \tau \e
  \end{bmatrix},
\]
and assuming that $J = e^{\tfrac{\pi i a^\dag a}2} Q$ commutes with $H^{\e}$, we see that
\begin{align}
  \label{eq:symcomd}
  Q H^{e} = \widetilde{H^{\e}} Q,
\end{align}
with
\[
   \widetilde{H^{\e}} =   \begin{bmatrix}
    a^\dag a + \tfrac12 - g ((a^\dag)^2+a^2) + \tau \e & \Delta \\
    \Delta &  a^\dag a + \tfrac12 + g ((a^\dag)^2+a^2) - \tau \e
  \end{bmatrix}.
\]

The condition \eqref{eq:symcomd} is written as
\begin{align*}
  \begin{bmatrix}
    [A,a^\dag a] & [B,a^\dag a] \\
    [C, a^\dag a] & [D, a^\dag a]
  \end{bmatrix}
  + g
  \begin{bmatrix}
    A((a^\dag)^2+a^2) + ((a^\dag)^2+a^2)A & - [B,(a^\dag)^2+a^2] \\
    [C,(a^\dag)^2+a^2] & -D((a^\dag)^2+a^2) - ((a^\dag)^2+a^2)D
  \end{bmatrix} \\
  + 2 \tau 
  \begin{bmatrix}
    0 & -\e \\
    \e & 0
  \end{bmatrix}
  + \Delta
  \begin{bmatrix}
    B-C & A -D \\
    D-A & C-B
  \end{bmatrix}
  = 0
\end{align*}
and determines a system of recurrence relations for the coefficients of $Q_\e$.

Recall that any element $f \in \C[a,a^\dag]$ can be written in the standard form
\[
  f = \sum_{m,n \geq 0} c_{m,n} (a^\dag)^m a^n,
\]
with $c_{m,n} \in \C$ and with $c_{m,n}=0$ for $m+n$ large enough. The total degree of $f$ is the largest value of $m+n$ such that $c_{m,n}\neq 0$. Let us denote for $A,B,C,D$ standard form by the corresponding greek letter, for instance $A = \sum_{m,n} \alpha_{m,n} (a^\dag)^m a^n $.

From this, the recurrence relations, for $m,n\geq 0$ are given by
\begin{align}
  \alpha_{m,n}(n-m) &+ 2 g \alpha_{m,n-2} + 2 g \alpha_{m-2,n} + (n+2)(n+1) g \alpha_{m,n+2} + (m+2)(m+1) g \alpha_{m+2,n}  \nonumber \\
               &+ 2 (n+1) g \alpha_{m-1,n+1} + 2 (m+1) g \alpha_{m+1,n-1} + \Delta(\beta_{m,n} - \gamma_{m,n}) = 0,  \nonumber   \\
  \beta_{m,n}(n-m-2 \tau \e) &- (n+2)(n+1) g \beta_{m,n+2} + (m+2)(m+1) g \beta_{m+2,n} \nonumber \\
               &- 2 (n+1) g \beta_{m-1,n+1} + 2 (m+1) g \beta_{m+1,n-1} + \Delta(\alpha_{m,n} - \delta_{m,n}) = 0,  \nonumber   \\
  \gamma_{m,n}(n-m+ 2 \tau \e) &+ (n+2)(n+1) g \gamma_{m,n+2} - (m+2)(m+1) g \gamma_{m+2,n}  \nonumber \\
               &+ 2 (n+1) g \gamma_{m-1,n+1} - 2 (m+1) g \gamma_{m+1,n-1} + \Delta(\delta_{m,n} - \alpha_{m,n}) = 0,    \nonumber  \\
  \delta_{m,n}(n-m) & - 2 g \delta_{m,n-2} - 2 g \delta_{m-2,n} - (n+2)(n+1) g \delta_{m,n+2} - (m+2)(m+1) g \delta_{m+2,n} \nonumber \\
               &- 2 (n+1) g \delta_{m-1,n+1} - 2 (m+1) g \delta_{m+1,n-1} + \Delta(\gamma_{m,n} - \beta_{m,n}) = 0.   \label{eq:recQ} 
\end{align}

With these preparations, we present the main result of this section. 

\begin{prop}
  \label{prop:sym}
  For $\e= \ell$, there is an operator $J_\ell = e^{\frac{\pi i a^\dag a}2} Q_\ell$ with $Q_\ell \in {\rm Mat}_2(\C[a,a^\dag])$, such that
  \begin{align}
    \label{eq:comm}
    [H_{\ell},J_\ell]=0,
  \end{align}
  and with the properties
  \begin{enumerate}
  \item  the entries of $Q_\ell$ have total degree $2 \ell$ in $\C[a,a^\dag]$,
  \item  the operator $J_\ell$ is self-adjoint,
  \item  $J_\ell^2 = \exp(\pi i a^\dag a) p_\ell(H_{\ell})$ for a polynomial $p_\ell(x)$ of degree $2\ell$.
  \end{enumerate}
\end{prop}

To prove the existence of the hidden symmetry operator it is necessary to show that the system of simultaneous linear recurrences given in \eqref{eq:recQ} has non-zero solutions. This reduces to solving certain systems of linear recurrences, depending on the system parameters, and checking certain compatibility conditions.

Concretely, we require to consider solutions of the homogeneous system of defined by the
$(N+1)$-dimensional matrices
\[
  \Omega_N = \Tridiag{(N-2i-2\tau \e)}{2g(i+1)}{-2g(N-i+1)}{0\le i\le N},
\]
and
\[
  \widehat{\Omega}_N = \Tridiag{(N-2i+2\tau \e)}{-2g(i+1)}{2g(N-i+1)}{0\le i\le N}.
\]

\begin{lem}
  \label{lem:SolOmega}
  Let $\e=\ell \in\Z_{\geq0}$, or $2 \e=\ell \in\Z_{\geq0}$ with $\ell \equiv 1 \pmod{2}$, and set $N = 2 \ell$, then eigenfunctions of the zero eigenvalues for $\Omega_N$ and
  $\widehat{\Omega}_N$ are given, respectively, by $B_N = \{\beta_0,\beta_1,\cdots,\beta_N \}$ and $C_N = \{\gamma_0,\gamma_1,\cdots,\gamma_N \}$
  with
  \[
    \beta_i = (-1)^i \binom{N}{i} (1+\tau)^{\ell-\frac{i}{2}}(1-\tau)^{\frac{i}2}, \qquad \gamma_i = (-1)^\ell \binom{N}{i} (1-\tau)^{\ell-\frac{i}{2}}(1+\tau)^{\frac{i}2},
  \]
  for $ 0 \leq i \leq N$. Moreover, the dimension of the zero eigenvalue for both matrices is $1$.
\end{lem}

\begin{proof}
  The result can be verified directly. For instance, for the case $\e=\ell \in\Z_{\geq0}$ and $\Omega_N$ we verify compute
  \begin{align*}
    & - 2 g (N-i+1) \beta_{i-1} + (N-2 i - 2 \tau \e)\beta_i + 2 g (i+1) \beta_{i+1} \\
    \qquad &= (-1)^i (1-\tau)^{\ell-\frac{i+1}{2}} (1-\tau)^{\frac{i-1}{2}} \binom{N}{i}\Bigg( - 2 g i(1+\tau)  \\
    &\qquad \qquad \qquad \qquad \qquad \qquad \qquad - (N-2i - 2\e \tau)(1-\tau^2) + 2g (N-i)(1-\tau)  \Bigg),
  \end{align*}
  and this expression vanishes since $\tau^2 = 1-4g^2$ and $2\e = N$. The dimension of the eigenspaces follows since the
  off-diagonal entries of the tridiagonal matrices $\Omega_N$ and $\widehat{\Omega}_N$ are nonzero (see e.g. \cite{HR2013}).
\end{proof}

Let us also consider the expansion of the entries of the operator $Q$ with respect to Bogoliubov transformed operators
$a_\pm, a_\pm^{\dag}$ given by
\[
  a_{\pm} = u a \pm v a^{\dag}, \qquad a_{\pm}^\dag = u a^{\dag} \pm v a,
\]
satisfying $[a_{\pm},a_{\pm}^\dag]=1$. Here, $u$ and $v$ are defined by \( u = \sqrt{\frac{1+\tau}\tau } \) and  \(v = \sqrt{\frac{1-\tau}\tau }\). Here, the expansions of entries are of the form
\[
  A = \sum_{m,n} \alpha_{m,n}^{(\pm)} (a_{\pm}^\dag)^m a_{\pm}^n,
\]
and similar for other entries. We remark that $\alpha_{m,n}^{(+)}$ with $m+n$ is written in terms of $\alpha_{m',n'}$ with $m'+n' \leq m+n$.
We note that this is equivalent to the approach considered in \cite{XC2022}, however we note that the expansion and therefore the recurrence relations that appear in the aforementioned paper are different than the ones presented here. We also note that the transforms above are equivalent to the Bogoliubov transforms used in Section \ref{sec:deg}.

For our purposes, it is enough to consider the recurrence relation for the components $C$ and $B$ of
$Q$, which is given by
\begin{align}
  \label{eq:recCp}
   \gamma^{(+)}_{m,n} & (n-m+ 2 \e)  + \frac{\Delta}{\tau} (\delta^{(+)}_{m,n} - \alpha^{(+)}_{m,n}) = 0,
\end{align}
and
\begin{align}
  \label{eq:recBm}
  \beta^{(-)}_{m,n}  (n-m - 2 \e)  + \frac{\Delta}{\tau} (\alpha^{(-)}_{m,n} - \delta^{(-)}_{m,n}) = 0.
\end{align}

With these preparations, we give the proof of Proposition \ref{prop:sym}.

\begin{proof}[Proof of Proposition \ref{prop:sym}]

  Let $Q \in {\rm Mat}_2(\C[a,a^\dag])$ have entries of total degree at most $2 \ell$ satisfying
  \eqref{eq:symcomd} with entries $A,B,C,D$ as in the foregoing discussion.

  To simplify the notation, for $n\geq 0$ let us denote by $A_n = \{ \alpha_{0,n}, \alpha_{1,n-1}, \cdots, \alpha_{n,0}\}$ and
  we write $A_n \equiv 0$ when all the coefficients are equal to zero, we follow the same notation for
  $B,C$ and $D$.

  First, by considering \eqref{eq:recQ} with $n+m=2\ell+2$ (resp. $n+m=2\ell+1$) we see that
  \[
     \alpha_{m,n-2} +  \alpha_{m-2,n} = 0.
  \]
  By taking $m=0, n=2 \ell+2$ and $m=1, n=2\ell+1$ (resp. $m=0, n=2 \ell+1$ and $m=1, n=2\ell$)
  we see that $A_{2 \ell} \equiv 0$ and $A_{2\ell -1}\equiv 0$. The same argument shows that $D_{2 \ell} \equiv 0$
  and $D_{2\ell -1}\equiv 0$.
  
  Next, we consider the case of $B$ of \eqref{eq:recQ} for $n+m= 2\ell$ (resp. $n+m=2\ell -1$) and note that
  the existence of the coefficients is equivalence to the equation
  \[
    \Omega_{2\ell} B_{2\ell} = 0,
  \]
  and such a vector exists by virtue of Lemma \ref{lem:SolOmega} and the situation for $C_N$ is
  identical (with $\Omega_{2\ell} $ replaced with $\widehat{\Omega}_{2\ell}$ ).

  From \eqref{eq:recCp} and \eqref{eq:recCp}, in the transformed variables $a_{\pm},a_{\pm}^\dag$ we see that
  there are no solutions of degree $2\ell-1$ for $B$ and $C$. Due to this and the relation $[a,a^\dag]=1$ (any
  transformation of expressions involving $a$ and $a^\dag$ preserves the parity of the degrees), we see that
  $B_{2\ell-1} \equiv 0$ and $C_{2\ell-1}\equiv 0$.
  
  Since the two solutions $B_{2\ell}$ and $C_{2\ell}$ were obtained independently, we need to fix the constant for the solutions, in order to do this consider the recurrence relation \eqref{eq:recQ}
  for $n+m=2\ell$ giving the equations
  \begin{align}
    \label{eq:AN2}
    2 g \alpha_{m,n-2} + 2 g \alpha_{m-2,n} + \Delta(\beta_{m,n} - \gamma_{m,n}) = 0,
  \end{align}
  which imposes the conditions
  \begin{align*}
    \sum_{i=0}^{\ell}(-1)^i (\beta_{2i,2\ell-2i} - \gamma_{2i,2\ell-2i}) &=0, \\
    \sum_{i=0}^{\ell-1}(-1)^i (\beta_{2i+1,2\ell-2i-1} - \gamma_{2i+1,2\ell-2i-1}) &=0.
  \end{align*}
  We directly verify that with the normalization given in Lemma \ref{lem:SolOmega}, this condition
  is satisfied directly. Indeed, we have
  \begin{align*}
    (-1)^i (\beta_{i,2\ell-i} - \gamma_{i,2\ell-i}) &= (-1)^{\ell+1+i} \left[(-1)^{2\ell-i} (\beta_{2\ell -i,i} - \gamma_{2\ell-i,i})\right] 
  \end{align*}
  and the condition is satisfied, since the summands cancel in pairs. Once the condition is verified,
  the recurrence relation system \eqref{eq:recQ} is completely determined and the coefficients of $A_{2\ell-2}$ are
  computed using \eqref{eq:AN2} and similarly for $D_{2\ell-2}$ and the remaining coefficients, finishing the
  proof of (1).
  The proof of (2) and (3) follow exactly as in \cite{RBBW2021}, so we omit them. 
\end{proof}


Note that the proof of Proposition \ref{prop:sym} also shows that for $\e=\ell$, the minimal solution of the form \eqref{eq:symcomd} must be of degree $2 \ell$. Regarding the uniqueness of the operator, it is straightforward to follow the proof of \cite{RBBW2021} to see that if $\exp(\frac{\pi i a^\dag a}2) Q$ with $Q \in {\rm Mat}_2(\C[a,a^\dag])$ and entries of total degree $2 \ell$ satisfies
\[
  [H_{\ell},\exp(\frac{\pi i a^\dag a}2) Q]
\]
then $\e \in \Z_{\geq 0}$ with $\e \leq \ell$. Furthermore, if $\e=\ell$ and $Q$ is a constant multiple of $Q_\ell$, or $\e = \ell' < \ell$ and then $Q = Q_{\ell'} p (H_{\ell'})$ for a polynomial $p$ of degree $\ell-\ell'$.  Moreover, if $\e \notin \tfrac12 \Z$, then we can show that there is no $Q \in {\rm Mat}_2(\C[a,a^\dag])$ satisfying \eqref{eq:symcomd}.

From Proposition \ref{prop:sym}, we see that solutions $B_{N}$ and $C_{N}$ of the recurrence relation exist only for $\e = \frac{N-2i}{2}$ for $i=0,N$. This shows that for $2\e \notin \Z $ there are no solutions $Q$ of \eqref{eq:recQ} with polynomial entries. Moreover, it also follows that if $\e = \pm N/2$, then $\det(\Omega_N)=0$ (resp. $\det(\widehat{\Omega}_N)=0$ ) and $\det(\Omega_{N-1}) \neq 0$ (resp. $\det(\widehat{\Omega}_N) \neq 0$).

\begin{ex}
  \label{ex:HSO}
  We give the explicit expression of the hidden symmetry operator for small values of $\ell$.
  Concretely, for $\ell=1$, we have
  \[
    Q_1 =
    \begin{bmatrix}
     -\frac{\Delta}{g} & \left(1+\tau \right)a^2 - 4 g a^\dag a  + \left(1-\tau\right) (a^{\dag})^2 - 2 g  \\
      - \left(1-\tau\right)a^2 - 4 g a^\dag a - \left(1+\tau\right)(a^{\dag})^2 -2 g & -\frac{\Delta}{g}.
    \end{bmatrix}
  \]
  For $\ell=2$, we have
  \(
    Q_2 =
    \begin{bmatrix}
      A_2 & B_2 \\
      C_2 & D_2
    \end{bmatrix},
    \)
    with
    \begin{align*}
      A_2 &= - \frac{2\tau  \Delta }{g} a^2 + 8 \Delta a^\dag a +  \frac{2\tau  \Delta }{g} (a^\dag)^2
            +    \left(4+\frac{2 \tau}{g^2}\right)\Delta \\
      B_2 &= (1+\tau)^2   a^4  - 8 g (1+\tau) a^\dag a^3 + 24 g^2(a^\dag)^2 a^2
            - 8 g (1-\tau) (a^\dag)^3 a + (1 -\tau)^2 (a^\dag)^5  \\
          &\quad  - 12 g(1+ \tau) a^2 +  48 g^2 a^\dag a - 12 g ( 1-\tau) (a^\dag)^2
            + \frac{\Delta ^2}{g^2}+ 12 g^2 , \\
      C_2 &=  (1-\tau)^2  a^4 + 8 g(1-\tau) a^\dag a^3 + 24 g^2(a^\dag)^2 a^2
            + 8 g (1+\tau) (a^\dag)^3 a + (1 + \tau)^2 (a^\dag)^5  \\
          &\quad 12 g(1-\tau) a^2 +  48 g^2 a^\dag a + 12 g (1+\tau) (a^\dag)^2 + \frac{\Delta ^2}{g^2}+12 g^2, \\
      D_2 &=  -\frac{ 2\tau  \Delta }{g} a^2 + 8 \Delta a^\dag a + \frac{ 2 \tau  \Delta }{ g} (a^\dag)^2
            +  \left(4-\frac{2 \tau}{g^2}\right)\Delta.
    \end{align*}
    We give the example of $\ell=3$ in Appendix \ref{sec:compl3}.
\end{ex}

We conclude this section by considering the case of $\e= n+ \frac{1}{2}$ with $n \in \Z_{\geq 0}$. As we showed in Section \ref{sec:diffparity}, the degenerate eigenvalues in this case are of opposite-parity type and therefore are associated with the $\Z_2$-symmetry. On the other hand, existence of solutions of $Q \in {\rm Mat}_2(\C[a,a^\dag])$ cannot be rule out on principle with the theory developed in this section. 

The following lemma is obtained from the proof of Proposition \ref{prop:sym}.

\begin{cor}
  \label{cor:necCond}
  A necessary condition for the existence of solutions of $Q \in {\rm Mat}_2(\C[a,a^\dag])$ with entries
  of total degree at most $N$ are
  \begin{align}
    \label{eq:compCond}
    \sum_{i=0}^{\lfloor N/2\rfloor}(-1)^i (\beta_{2i,N-2i} - \gamma_{2i,N-2i}) &=0, \\
    \sum_{0=1}^{\lfloor N/2\rfloor-1}(-1)^i (\beta_{2i+1,N-2i-1} - \gamma_{2i+1,N-2i-1}) &=0.
  \end{align}
\end{cor}

The next proposition confirms that the opposite-parity degeneracies, discussed in Section \ref{sec:degstruct}, do not correspond to a hidden symmetry of the type discussed in this section.

\begin{prop}
  \label{eq:SimOpHalf}
  For $\e= n+ \frac{1}{2}$ with $n \in \Z_{\geq 0}$, there are no operators $J = e^{\tfrac{\pi i a^\dag a}2} Q$, with entries $Q$ of finite degree, such that $[H_{\frac{\ell}{2}},J]= 0$.
\end{prop}

\begin{proof}
  Let  $\e = \frac{\ell}2$ with $\ell \equiv 1 \pmod{2}$ and $N=\ell$. First, we set the solutions of degree
  $N$ for $B$ and $C$ according to Lemma \ref{lem:SolOmega}. We show that there is no constant $c \in \C$ such that the equation
  \begin{align}
    \label{eq:sumCond1}
    \sum_{i=0}^{\frac{\ell-1}2}(-1)^i (\beta_{2i,\ell-2i} - c \gamma_{2i,\ell-2i}) =0
  \end{align}
  is satisfied.

  Note that we have
  \begin{align*}
    \beta_{2i,\ell-2i} &= (-1)^i \binom{\ell}{i} \left( 1 + (-1)^{i} \tau^N + q_{2i}(\tau)\right) \\
    \gamma_{2i,\ell-2i} &= (-1)^N \binom{\ell}{i} \left( 1 + (-1)^{\ell-i} \tau^N + p_{2i}(\tau)\right) 
  \end{align*}
  for  $q_{2i},p_{2i} \in \C[\tau]$ of degree $\ell-1$ with no constant term.

  Then, as a polynomial in $\tau$, the coefficient of degree $\ell$ on the left-hand side of \eqref{eq:sumCond1}
  is given by
  \begin{align*}
    &\sum_{i=0}^{\frac{\ell-1}2} (-1)^i \left[(-1)^i \binom{\ell}{2i} (1 - c)\right] = (1 - c)  \sum_{i=0}^{\frac{\ell-1}2}\binom{\ell}{2i}
  \end{align*}
  and  the constant term of the left-hand side of \eqref{eq:sumCond1} is
  \begin{align*}
    &\sum_{i=0}^{\frac{\ell-1}2} (-1)^i \binom{\ell}{2i} (1 - (-1)^{\ell} c) = (1 + c) \sum_{i=0}^{\frac{\ell-1}2} (-1)^i \binom{\ell}{2i}.
  \end{align*}
  
  Since
  \[
    \sum_{i=0}^{\frac{\ell-1}2}\binom{\ell}{2i} = 2^{\ell-1}, \qquad \sum_{i=0}^{\frac{\ell-1}2} (-1)^i \binom{\ell}{2i} = 2^{\ell/2} \cos\left(\frac{\ell \pi}{4}\right),
  \]
  (see, e.g. \cite{PBM1998} 4.2.1 (8)-(9)) are both nonzero and $\beta\neq0$, there is no constant $c$ that satisfies the condition \eqref{eq:sumCond1}.
\end{proof}

\subsection{Commutant algebra of the Hamiltonian}
\label{sec:SymtpQRM}

Once the existence of the hidden symmetry $J_\ell$ is settled, the natural question is how is this related
to the $\Z_4$-symmetry of the usual 2pQRM. Since 
\begin{align}
  \label{eq:J4}
  J_\ell^4 = p_\ell(H_{\ell};g,\Delta)^2,
\end{align}
except for the case $\ell=0$, the hidden symmetry operator is not an operator of order $4$, that is, it does not
determine a $\Z_4$ symmetry. Equivalently, there is no canonical way to define invariant subspaces of the Hilbert
space using the operator $J_\ell$ to define (generalized) ``parities'' as in the 2pQRM or other similar models. 

In \cite{XC2022} it has been suggested to consider the sign (in a general sense, that is, including roots of unity) of the eigenvalue $\mu$ of $J_\ell$
in the equation
\[
  p_\ell(\lambda;g,\Delta)^2 = \mu^4
\]
for an eigenvalue $\lambda$ of $H_{\ell}$. This approach may be considered as a type of normalization  of the operator $J_\ell$. In practice, it is rather complicated to verify the sign of $\mu$ and more importantly, it has not been proved that $p_\ell(\lambda;g,\Delta) \neq 0$ for all eigenvalues $\lambda$ of $H_{\ell}$, or equivalently, that $J_\ell$ has a trivial kernel. This can easily be reduced to the verification a finite number of eigenvalues, however it a proof in general seems to be complicated and we the issue open. The equivalent question for the AQRM is still not settled \cite{RBW2022}.

The structure of the commutant algebra has been suggested as an alternative approach to the analysis of the symmetry  for quantum interaction models \cite{RBBW2021}, specifically in relation to quantum integrability.

Denote by $\mathcal{C}_a(H_\ell)$ the commutant algebra of $H_\ell$ in $\rm{Mat_2(\overline{\C[a,a^\dag])}}$.
We need to consider a restriction of this commutant algebra since in $\rm{Mat_2(\overline{\C[a,a^\dag])}}$ it is possible to define ``trivial'' elements that commute with the Hamiltonian (see the discussion in \cite{RBBW2021}). On the  other hand, restricting to matrices with entries in $\C[a,a^\dag]$ does 
not even allow the presence of the $\Z_2$-symmetry operator in the commutant algebra.
Let $\mathcal{A}$ be the matrix algebra generated by self-adjoint elements of $\rm{Mat_2(\overline{\C[a,a^\dag])}}$ and $\exp({\tfrac{\pi i a^\dag a}2})$. Then, the restricted commutant algebra $\mathcal{C}_a'(H_\ell)$ is defined by $\mathcal{C}_a'(H_\ell):=\mathcal{C}_a(H_\ell) \cap \mathcal{A}$. 

The following result summarizes the discussion of this section.

\begin{prop}
  For $\ell \in \Z$, the commutant algebra $\mathcal{C}_a'(H_\ell)$ of $H_{\ell}$ is
  generated by the Hamiltonian $H_{\ell}$, the parity operator $\mathcal{P}$ and the hidden symmetry operator $J_{\ell}$. In particular, we have 
  \[
      \mathcal{C}_a'(H_\ell) \simeq \R[x,y,z]/\langle x^2-1,y^2-x p_\ell(z;g,\Delta)\rangle,
  \]
  where the isomorphism is given by
  \begin{align*}
      x &\mapsto \mathcal{P}_x^2 = \exp(\pi i a^\dag a), \qquad  y \mapsto J_\ell, \qquad  z \mapsto H_\ell.
  \end{align*}
\end{prop}

This setting provides another geometric perspective by considering the commutant $\mathcal{C}_a'(H_\ell)$ as a ring of functions of a certain variety. For instance, for the symmetric case we have $p_0(z;g,\Delta)=1$, and it is easy to verify that a primary decomposition (see, e.g., \cite{AM1969,CLS2015}) of the ideal $\langle x^2-1,y^2-x p_0(z;g,\Delta)\rangle$ is given by
\[
    \langle x^2-1,y^2-x)\rangle = \langle x+1, y-i\rangle \cap \langle x+1,y+i\rangle \cap \langle x-1,y-1\rangle \cap \langle x-1,y+ 1\rangle,
\]  
reflecting the $\Z_4$-symmetry of the 2pQRM Hamiltonian. On the other hand, for $\ell>0$, a primary decomposition is given by
\[
   \langle x^2-1,y^2-x p_\ell(z;g,\Delta)\rangle =  \langle x+1,y^2+ p_\ell(z;g,\Delta)\rangle \cap \langle x-1,y^2- p_\ell(z;g,\Delta)\rangle,
\]
thus we may interpret the polynomial $p_\ell(z;g,\Delta)$ as an obstruction of the $\Z_4$-symmetry.

A study of the properties of the commutant algebra $\mathcal{C}_a'(H_\ell)$, as well as the applications to quantum integrability, is outside the scope of this paper, but we note that it further shows the importance of the algebro-geometric picture of the 2pQRM.

\section{Conjectural explicit relation between hidden symmetry and degeneracy}
\label{sec:symdeg}

The quartic relation \eqref{eq:J4} between the Hamiltonian of the 2pAQRM and its hidden symmetry operator $J_{\ell}$ gives a fascinating but reasonable description of a geometric picture for the spectrum of operators. Concretely, for fixed $\Delta>0$ the joint eigenvalues $(\lambda,\mu)$ of $H_{\ell}$ and $J_\ell$ lie on the surface described by \eqref{eq:J4}. It is clear that the said tuple must satisfy either of the equations
\begin{align*}
  \mu^2 = p_\ell(\lambda;g,\Delta), \qquad \mu^2 = - p_\ell(\lambda;g,\Delta),
\end{align*}
and thus the geometric picture is determined by hypergeometric equations. These hyperellipic curves also appear in the decomposition of the primary decomposition of the defining ideal of the of the commutant  algebra $\mathcal{C}_a'(H_\ell)$.

Currently, there is no method to determine the polynomial $p_\ell(x;g,\Delta)$ outside of the explicit computation for small $\ell$. Thus, in practice, it is difficult to explore this geometric picture, the 2pAQRM. In this section, we discuss an approach to this problem that also gives a relation with the other geometric picture given in Section \ref{sec:qtpol} based on the excellent approximation using the inter-constraint polynomial.

First, let us give the explicit forms of the polynomial $p_\ell(x;g,\Delta)$ via direct computation for small
values of $\ell$.

\begin{ex}
  Using the expressions given in Example \ref{ex:HSO}, we obtain by direct computation
  \begin{align*}
    J_1^2 &=  \left(4 (H_{1})^2 +  \frac{\tau^2 (\Delta^2 - g^2)}{g^2} \right)e^{\pi i a^\dag a}
  \end{align*}
  and
  \begin{align*}
    J_2^2 &= \Bigg( 16 (H_{2})^4 + \frac{\tau^2 (8 \Delta ^2 g^2-40 g^4)}{g^4} (H_{2})^2
    + \frac{\tau^4 \left(\Delta ^4+4 \Delta ^2+9 g^4-10 \Delta ^2 g^2\right)}{g^4} \Bigg) e^{\pi i a^\dag a}. 
  \end{align*}
  In particular, we have
  \begin{align*}
    p_1(x;g,\Delta) &= 4 x^2 + \frac{\tau^2}{g^2}(\Delta^2 - g^2), \\
    p_2(x;g,\Delta) &= 16 x^4 + \frac{\tau^2}{g^4} (8 \Delta ^2 g^2-40 g^4) x^2  + \frac{\tau^2}{g^2}(\Delta^2 - g^2).
  \end{align*}  
\end{ex}

We remark here that the case $\ell=1$ appears in \cite{XC2021} using a different normalization for $J_\ell$. 
Since the Hamiltonian for the 2pAQRM used in \cite{XC2021} differs from \eqref{eq:tpHamilt} by $\tfrac{1}{2} \bf{I_2}$ their version of the polynomial $p_1(x;g,\Delta)$ does not contain only even powers of the Hamiltonian.

One of the most intriguing open problems of the AQRM is the fact that the corresponding polynomial is essentially equal to inter-constraint polynomial. Moreover, since the variable $x$ appears square in the polynomial $p_\ell(x;g,\Delta)$ in the above examples and, in light of Proposition \ref{prop:evenpow}, it seems natural to compare with the inter-constraint
polynomial $A^{(3)}(z)$. We verify that
\begin{align*}
  p_1(x;g,\Delta) &=  \frac{\tau^2}{g^2} A^{(1)}(\tfrac{1}{\tau} x), \qquad
  p_2(x;g,\Delta) =  \frac{\tau^4}{g^4} A^{(2)}(\tfrac{1}{\tau} x)
\end{align*}
and similarly, for the case $\ell=3$ described in Appendix \ref{sec:compl3}, we also
verify the following. 
\[
  p_3(x;g,\Delta) =  \frac{\tau^6}{g^6} A^{(3)}(\tfrac{1}{\tau}x),
\]
and similarly for the cases $\ell \leq 6$. Based on the numerical evidence and the precedent for the AQRM, we formulate the general case as a conjecture.

\begin{conject}
  \label{conj:symDeg}
  For $\ell\geq 0$, the polynomials $p_\ell(x;g,\Delta)$ and $A^{(\ell)}(\tfrac{x}{\tau})$ are equal up to a constant, that is,
  \begin{align}
    \label{eq:conject2eq}
    p_\ell(x;g,\Delta) = c A^{(\ell)}\left(\frac{x}{\tau}\right).
  \end{align}
  for $c = c(g,\Delta)$ depending only on the system parameter and determined by the normalization chosen for $J_\ell.$
\end{conject}

Using a normalization such that the coefficient of $a^{2\ell}$ in $B$ (cf. \eqref{eq:Qdef}) is given by $(1+\tau)^2$
(used e.g. in Appendix \ref{sec:compl3} and Example \ref{ex:HSO}), we have
\[
  p_\ell(x;g,\Delta) = \left(\frac{\tau}{g}\right)^{2\ell} A^{(\ell)}\left(\frac{x}{\tau}\right).
\]
We also note that it is possible to normalize $J_\ell$ such that $c=1$ but in general it seems to lead to a more complicated expression for $J_\ell$ than the one given in Appendix \ref{sec:compl3}. 

Conjecture \ref{conj:symDeg} may be interpreted as an explicit relation between hidden symmetry and spectral degeneracy through the  polynomials arising in each of the settings, in the absence of a canonical decomposition of the Hilbert space for general $\ell>0$ (see the discussion in Section \ref{sec:SymtpQRM}). Moreover, it also provides a connection between the spectral picture determined by the operator valued relation \eqref{eq:J4} and the one determined by the excellent approximation of energy curves by the inter-constraint polynomial (cf. Section \ref{sec:excapprox1}).

In practical terms, since the inter-constraint polynomial can be computed from the determinant expression in Theorem \ref{thm:div}, the conjecture above provides an efficient way to compute the polynomial $p_\ell(x;g,\Delta)$. The equivalence of the two conjectures is easy to verify.

\begin{conject}
  \label{conj:symDeg2}
  For $\ell\geq 0$, the polynomial $p_\ell(x;g,\Delta)$ is, up to a factor depending only on the system parameters, given
  by the determinant expression
  \[
    (\tfrac{x}{2\tau} + \tfrac{1}4 - \tfrac{\ell}{2})_\ell (\tfrac{x}{2\tau} + \tfrac{3}4 - \tfrac{\ell}{2})_\ell
    \det \left(2 \bm{I}_\ell (8g^2-1) + \bm{D}_\ell \Delta^2 + \bm{V}_\ell \right),
  \]
where $\bm{D}_\ell  = \Diag\left\{ \frac{4 \tau ^2}{x^2- \tau ^2 (\ell+\frac{1}{2} -2 i)^2} \right\}_{1\leq i \leq \ell}$ and
\[
  \bm{V}_\ell =  \Tridiag{\frac{1-4 \ell^2}{2 \left(-2 i+\ell -\frac{1}{2}\right)\left(-2 i+\ell +\frac{3}{2}\right)}}{\frac{(-2 i )(-2 i-1)}{\left(\ell-2 i -\frac{1}{2}\right)\left(\ell-2 i +\frac{1}{2}\right)  }}{\frac{ (2 (\ell-i)) (2 (\ell-i)-1 )}{\left(\ell -2 i-\frac{1}{2}\right) \left(\ell -2 i -\frac{3}{2}\right) }}{1\le i\le \ell}.
\]
\end{conject}

At present, there is no reasonably considered approach to the proof of these conjectures. Although a direct attempt based on the recurrence relations \eqref{eq:recQ} of $J_\ell$ is possible in principle, in practice it seems excessively complicated due to the fact that it concerns a system of simultaneous recurrence relations. We expect that exploring the geometric pictures described in this paper (and \cite{RBW2022} for the AQRM) may lead to an elegant proof of the conjectures, and conversely, the proof and statement of the conjectures should give new insights on the geometric picture for these models. Moreover, it would be interesting to study if a similar structure appears in other quantum systems, e.g., the XXZ model \cite{BA2001, FM2001} or the Kondo effect \cite{K2012}.

Moreover, we also expect that due to the relations between models based on covering and non-unitary equivalences
described in Section \ref{sec:relationsModels}, the proof of the conjectures for certain models may imply the truth of the conjecture for other models. Concretely, we expect that the conjectures for the 2pAQRM imply those of the AQRM due to the covering relation. These relations will be explored in a sequel of the present paper \cite{RBW2027}.

\begin{rem}
  The variable $z = \frac{x}{\tau}$ appearing in the conjectural relation \eqref{eq:conject2eq} suggests that a more natural choice for the Hamiltonian of the 2pAQRM may be 
  \[
    \hat{H}_{\e} = \frac1{\tau} \left( (a^\dag a + \frac12) + g ( a^2 + (a^\dagger)^2 ) \sigma_z + \Delta \sigma_x + \tau \e \sigma_z \right),
  \]
  that is, the one used for the calculation of the normalized  spectral curves in this paper.

  Moreover, the non-standard choices made in the definition of the Hamiltonian of the 2pAQRM in Section
  \ref{sec:tpQRM} seem justified as they lead to simpler expressions for the polynomials $p_\ell(x;g,\Delta)$ and $A^{(\ell)}(z)$.
\end{rem}

\appendix

\section{Hidden symmetry operator for $\ell=3$}
\label{sec:compl3}

We use the notation of Example \ref{ex:HSO} in Section \ref{sec:hiddenS} for the symmetry operator. We write the coefficients of a given total degree in vector form, for instance, $A_n = \{ \alpha_{0,n}, \alpha_{1,n-1}, \cdots, \alpha_{n,0}\}$, and write $A_n \equiv 0$ to indicate that all the coefficients are zero, and similarly for the other components of $Q_3$.

We recall that in order to compute the coefficients, we start by setting $A_{N} \equiv D_{N}\equiv 0$ for $N \geq 2\ell-1$. Then, the values of $B_{2\ell}$ and $C_{2\ell-1}$ are given by Lemma \ref{lem:SolOmega} and by the proof of Proposition \ref{prop:sym} we have $B_{2\ell} \equiv 0$ and $C_{2\ell-1} \equiv 0$. With this choice of solutions, there is no need to normalize either of the solutions by a constant, since they satisfy the conditions \eqref{eq:compCond}. Then, the computation of the remaining coefficients is done using the recurrence relations in Section \ref{sec:symmetryproof}. At the end of the computation we further normalize to give the expression in the form described in Conjecture \ref{conj:symDeg}.

\begin{ex}
  Let us give the coefficients of the matrix $Q_3$ in vector form. We note that for $X \in \{A,B,C,D\}$ we have
  $X_5 \equiv X_3 \equiv A_1 \equiv 0$. For the diagonal components, we have $A_6\equiv 0$,
\begin{align*}
  A_4 &= \frac{4\Delta}{g} \left\{3 g^2- 1  ,6 g \tau ,-18 g^2 + 1,-6 g \tau, 3  g^2-1  \right\}, \\
  A_2 &= \frac{4\Delta}{g^2}\left\{3 g^2(3 \tau -2)+2,-2 g \left(3 \tau +18 g^2-1\right),-3 g^2 (3 \tau -2)-2  \right)\\
  A_0 &= \frac{\Delta}{g^3}\left\{-36 g^4 + (36-12 \tau)g^2 -\Delta ^2-8\right\}.
\end{align*} 
Similarly, we have $D_6 \equiv 0$,  and
\begin{align*}
  D_4 &= \frac{4\Delta}{g} \left\{3 g^2- 1  ,6 g \tau ,-18 g^2 + 1,-6 g \tau, 3  g^2-1  \right\}, \\
  D_2 &= \frac{4\Delta}{g^2}\left\{3 g^2(3 \tau +2)-2,2 g \left(3 \tau -18 g^2+1\right),-3 g^2 (3 \tau +2)+2  \right\}\\
  D_0 &= \frac{\Delta}{g^3}\left\{-36 g^4 + (36+12 \tau)g^2 -\Delta ^2-8\right\}.
\end{align*}
For the off diagonal components, we have
\begin{align*}
  B_6 &= \bigg\{ (1+\tau)^3, -12g (1+\tau)^2 ,15 (1-\tau) (1+ \tau)^2,-160 g^3,15 (1-\tau)^2 (1+\tau), \\
        &\qquad -12g (1-\tau)^2, (1-\tau)^3 \bigg\} \\
  B_4 &= 60 g \left\{-\tau +2 g^2-1, 4 (1+\tau) g, -12 g^2,4 (1-\tau) g, \tau +2 g^2-1  \right\}, \\
  B_2 &= \frac{1}{g^2} \left\{\Delta ^2 (3 \tau +1)+180 g^4 (\tau +1), -12 g \left(\Delta ^2+60 g^4\right), \Delta ^2 (1-3 \tau )-180 g^4 (\tau -1) \right\}  \\
  B_0 &= \frac{1}{g} \left\{ -6 \left(\Delta ^2+20 g^4\right)\right\},
\end{align*}
and
\begin{align*}
  C_6 &= \bigg\{ -(1-\tau)^3, -12g (1-\tau)^2 ,- 15 (1+\tau) (1- \tau)^2,-160 g^3,-15 (1-\tau) (1+\tau)^2, \\
        &\qquad -12g (1+\tau)^2, -(1+\tau)^3 \bigg\} \\
  C_4 &= 60 g \left\{\tau +2 g^2-1, -4 (1-\tau) g, -12 g^2,-4 (1+\tau) g, -\tau +2 g^2-1  \right\}, \\
  C_2 &= \frac{1}{g^2} \left\{\Delta ^2 (3 \tau -1)-180 g^4 (1-\tau), -12 g \left(\Delta ^2+60 g^4\right),
       - \Delta ^2 (1+3 \tau )-180 g^4 (1+\tau) \right\}  \\
  C_0 &= \frac{1}{g} \left\{ -6 \left(\Delta ^2+20 g^4\right) \right\}.
\end{align*}
\end{ex}

\section*{Acknowledgments}

This work was partially supported by JSPS Grant-in-Aid for Scientific Research (C) No.20K03560, JSPS Grant-in-Aid for Early-Career Scientists No. 24K16941, and CREST JPMJCR2113, Japan.

\begin{flushleft}
  
\bigskip

  Cid Reyes-Bustos \\
  NTT Institute for Fundamental Mathematics,\\
  NTT Communication Science Laboratories, \\
  NTT, Inc. \\
  3-9-11, Midori-cho Musashino-shi, Tokyo, 180-8585, Japan \\
  email: \texttt{cid.reyes@ntt.com, math@cidrb.me}
  \phantom{a}\\\phantom{a}\\

  Masato Wakayama \\
  NTT Institute for Fundamental Mathematics,\\
  NTT Communication Science Laboratories, \\
  NTT, Inc. \\
  3-9-11, Midori-cho Musashino-shi, Tokyo, 180-8585, Japan \\
  and\\
  Faculty of Social Informatics, ZEN University, \\
  4-12-15 Ginza, Chuo-ku, Tokyo 104-0061, Japan\\
  email: \texttt{masato.wakayama@ntt.com, masato\_wakayama@zen.ac.jp}
\end{flushleft}

\end{document}